\documentclass[letterpaper,12pt,titlepage,oneside,final]{book}
 
\newcommand{\href}[1]{#1} % does nothing, but defines the command so the
\newcommand{\eq}[1]{(\ref{eq:#1})}

\newcommand{\scn}[1]{Sec.~(\ref{sec:#1})}

\newcommand{\order}[1]{\ensuremath{\mathcal{O}(#1)}}
\newcommand{\diby}[2]{\ensuremath{\frac{\partial #1}{\partial #2}}}
\newcommand{\ddiby}[2]{\ensuremath{\frac{\delta #1}{\delta #2}}}
\newcommand{\equa}[1]{\begin{equation} #1 \end{equation}}
\newcommand{\pb}[2]{\ensuremath{\lf\{#1,#2 \rt\}}}
\newcommand{\myexp}[1]{\ensuremath{\exp\lf\{ #1 \rt\}}}
\newcommand{\mean}[1]{\ensuremath{\lf\langle #1 \rt\rangle }}

\def\lf {\ensuremath{\left}}
\def\rt {\ensuremath{\right}}
\def\ra {\ensuremath{\rightarrow}}

\def\qand {\ensuremath{\quad\text{and}}}

\newcommand{\ta}[2]{\ensuremath{\lf.t_\alpha\rt.^{#1}_{#2}}}

\newcommand{\hn}[1]{\ensuremath{\mathcal H_{(#1)}}}
\newcommand{\wn}[1]{\ensuremath{\omega_{(#1)}}}
\newcommand{\sn}[1]{\ensuremath{S_{(#1)}}}
\newcommand{\ct}[1]{\ensuremath{\cosh^{#1} (t/\alpha)}}

\newcommand{\tht}[1]{\ensuremath{\tanh^{#1} (t/\alpha)}}
\newcommand{\ordn}[1]{\ensuremath{\mathcal O(\epsilon^{#1})}}

\def\pia {\ensuremath{\pi_\alpha}}
\def\wa {\ensuremath{\phi^\alpha}}
\def\lin {\ensuremath{\mathcal{H}_\alpha}}
\def\diff {\ensuremath{\mathcal{H}^a}}
\def\diffadm {\ensuremath{\mathcal{H}^a_\text{ADM}}}
\def\ham {\ensuremath{\mathcal{H}}}
\def\hamadm {\ensuremath{\mathcal{H_\text{ADM}}}}
\def\Na {\ensuremath{N^\alpha}}

\def\rcs {\ensuremath{\mathcal{R}}}
\def\acs {\ensuremath{\mathcal{A}}}
\def\cG {\ensuremath{\mathcal{G}}}
\def\covd {\ensuremath{\mathcal{D}_\phi}}

\def\dw {\ensuremath{\dot{\phi}^\alpha}}
\def\htot {\ensuremath{H_\text{T}}}
\def\dq {\ensuremath{\dot {q}}}
\def\dxi {\ensuremath{\mathcal{D}_{\xi,\phi}}}
\def\absk {\ensuremath{\lf| k \rt|}}
\def\wts {\ensuremath{\widetilde{TS}}}

\def\hg {\ensuremath{\mathcal H_\text{gl}}}
\def\ts {\ensuremath{T_\phi \frac{S}{\sqrt g}}}

\def\hd {\ensuremath{\mathcal H_\text{dual}}}

\usepackage{ifthen}
\newboolean{PrintVersion}
\setboolean{PrintVersion}{false} 
\usepackage{amsmath,amssymb,amstext,amsthm} % Lots of math symbols and environments
\newtheorem{proposition}{Proposition}
\newtheorem{principle}{Principle}
\usepackage[pdftex]{graphicx} % For including graphics N.B. pdftex graphics driver 

\usepackage[pdftex,letterpaper=true,pagebackref=false]{hyperref} % with basic options
\hypersetup{
    plainpages=false,       % needed if Roman numbers in frontpages
    unicode=false,          % non-Latin characters in Acrobat’s bookmarks
    pdftoolbar=true,        % show Acrobat’s toolbar?
    pdfmenubar=true,        % show Acrobat’s menu?
    pdffitwindow=false,     % window fit to page when opened
    pdfstartview={FitH},    % fits the width of the page to the window
    pdftitle={Shape dynamics and Mach's principles: Gravity from conformal geometrodynamics},    % title: CHANGE THIS TEXT!
    pdfauthor={Sean Gryb},    % author: CHANGE THIS TEXT! and uncomment this line
    pdfsubject={Shape Dynamics},  % subject: CHANGE THIS TEXT! and uncomment this line
    pdfnewwindow=true,      % links in new window
    colorlinks=true,        % false: boxed links; true: colored links
    linkcolor=blue,         % color of internal links
    citecolor=green,        % color of links to bibliography
    filecolor=magenta,      % color of file links
    urlcolor=cyan           % color of external links
}
\ifthenelse{\boolean{PrintVersion}}{   % for improved print quality, change some hyperref options
\hypersetup{	% override some previously defined hyperref options
    citecolor=black,%
    filecolor=black,%
    linkcolor=black,%
    urlcolor=black}
}{} % end of ifthenelse (no else)

\let\origdoublepage\cleardoublepage
\newcommand{\clearemptydoublepage}{%
  \clearpage{\pagestyle{empty}\origdoublepage}}
\let\cleardoublepage\clearemptydoublepage

\begin{document}

% For a large document, it is a good idea to divide your thesis
% into several files, each one containing one chapter.
% To illustrate this idea, the "front pages" (i.e., title page,
% declaration, borrowers' page, abstract, acknowledgements,
% dedication, table of contents, list of tables, list of figures,
% nomenclature) are contained within the file "uw-ethesis-frontpgs.tex" which is
% included into the document by the following statement.
%----------------------------------------------------------------------
% FRONT MATERIAL
%----------------------------------------------------------------------
% T I T L E   P A G E
% -------------------
% Last updated May 24, 2011, by Stephen Carr, IST-Client Services
% The title page is counted as page `i' but we need to suppress the
% page number.  We also don't want any headers or footers.
\pagestyle{empty}
\pagenumbering{roman}

% The contents of the title page are specified in the "titlepage"
% environment.
\begin{titlepage}
        \begin{center}
        \vspace*{1.0cm}

        \LARGE
        {\bf Shape dynamics and Mach's principles: \\ Gravity from conformal geometrodynamics}

        \vspace*{1.0cm}

        \normalsize
        by \\

        \vspace*{1.0cm}

        \Large
        Sean Gryb \\

        \vspace*{3.0cm}

        \normalsize
        A thesis \\
        presented to the University of Waterloo \\ 
        in fulfillment of the \\
        thesis requirement for the degree of \\
        Doctor of Philosophy \\
        in \\
        Physics \\

        \vspace*{2.0cm}

        Waterloo, Ontario, Canada, 2011 \\

        \vspace*{1.0cm}

        \copyright\ Sean Gryb 2011 \\
        \end{center}
\end{titlepage}

% The rest of the front pages should contain no headers and be numbered using Roman numerals starting with `ii'
\pagestyle{plain}
\setcounter{page}{2}

\cleardoublepage % Ends the current page and causes all figures and tables that have so far appeared in the input to be printed.
% In a two-sided printing style, it also makes the next page a right-hand (odd-numbered) page, producing a blank page if necessary.

% D E C L A R A T I O N   P A G E
% -------------------------------
  % The following is the sample Delaration Page as provided by the GSO
  % December 13th, 2006.  It is designed for an electronic thesis.
  \noindent
I hereby declare that I am the sole author of this thesis. This is a true copy of the thesis, including any required final revisions, as accepted by my examiners.

  \bigskip
  
  \noindent
I understand that my thesis may be made electronically available to the public.

\cleardoublepage
%\newpage

% A B S T R A C T
% ---------------

\begin{center}\textbf{Abstract}\end{center}

We develop a new approach to classical gravity starting from Mach's principles and the idea that the local shape of spatial configurations is fundamental. This new theory, \emph{shape dynamics}, is equivalent to general relativity but differs in an important respect: shape dynamics is a theory of dynamic conformal 3--geometry, not a theory of spacetime. Equivalence is achieved by \emph{trading} foliation invariance for local conformal invariance (up to a global scale). After the trading, what is left is a gauge theory invariant under 3d diffeomorphisms and conformal transformations that preserve the volume of space. There is one non--local global Hamiltonian that generates the dynamics. Thus, shape dynamics is a formulation of gravity that is free of the local problem of time. In addition, the symmetry principle is simpler than that of general relativity because the local canonical constraints are linear and the constraint algebra closes with structure constants. Therefore, shape dynamics provides a novel new starting point for quantum gravity. Furthermore, the conformal invariance provides an ideal setting for studying the relationship between gravity and boundary conformal field theories.

The procedure for the trading of symmetries was inspired by a technique called \emph{best matching}. We explain best matching and its relation to Mach's principles. The key features of best matching are illustrated through finite dimensional toy models. A general picture is then established where relational theories are treated as gauge theories on configuration space. Shape dynamics is then constructed by applying best matching to conformal geometry. We then study shape dynamics in more detail by computing its Hamiltonian perturbatively and establishing a connection with conformal field theory.

\cleardoublepage
%\newpage

% A C K N O W L E D G E M E N T S
% -------------------------------

\begin{center}\textbf{Acknowledgements}\end{center}

I would like to thank Henrique Gomes and Tim Koslowski. It has been both a pleasure and a great learning experience to work with you on developing shape dynamics.

We would be lost without teachers. In my case, this is particularly true. I have been lucky enough to have worked with two individuals that I consider to be mentors and friends: Julian Barbour and Lee Smolin. Thank you for your leadership, knowledge, and guidance.

There are many people that I have collaborated with during the development of this thesis. Timothy Budd, Flavio Mercati, and Karim Thebault have particularly contributed to my understanding of this work. I am indebted to their expertise. Many thanks to Ed Anderson, Laurent Freidel, Louis Leblanc, Fotini Markopoulou, Sarah Shandera, Rafael Sorkin, Rob Spekkens, Tom Zlosnik, and countless others that I have missed for valuable discussions during the course of my PhD. I would like to particularly thank Hans Westman for introducing me to best matching and Niayesh Afshordi for his gracious support and interest in the relational program. I would also like to extend my gratitude to the administrative staff at the Perimeter Institute, including the Outreach department and the Bistro workers, for helping to create and maintain a unique and inspiring working environment.

Without frequent visits to College Farm, nestled in the rolling hills of North Oxfordshire, or the many scribblings on chalk boards throughout the Perimeter Institute, none of this would have been possible. I am grateful to FQXi for funding my trips to College Farm and NSERC for its support of Perimeter. Research at the Perimeter Institute is supported in part by the Government of Canada through NSERC and by the Province of Ontario through MEDT

Finally, I would like to thank my friends and family. Thanks to Damian Pope and Yaacov Iland for their support and understanding. Thanks to my family: Neal, Laurie, and Dad, for all your love. Thanks to Mom, who will live forever in our hearts. You are my strength.
\cleardoublepage
%\newpage

% D E D I C A T I O N
% -------------------

\begin{center}\textbf{Dedication}\end{center}

To Mom, for showing me what is important and for giving me the strength to be free. You lived more than most could hope and are loved more than any could dream.
\cleardoublepage
%\newpage

% T A B L E   O F   C O N T E N T S
% ---------------------------------
%\renewcommand\contentsname{Table of Contents}
\tableofcontents
%\addcontentsline{toc}{chapter}{Table of Contents}
\cleardoublepage
\phantomsection
%\newpage

% L I S T   O F   T A B L E S
% ---------------------------
\addcontentsline{toc}{chapter}{List of Tables}
\listoftables
\cleardoublepage
\phantomsection		% allows hyperref to link to the correct page
%\newpage

% L I S T   O F   F I G U R E S
% -----------------------------
\addcontentsline{toc}{chapter}{List of Figures}
\listoffigures
\cleardoublepage
\phantomsection		% allows hyperref to link to the correct page
%\newpage

% L I S T   O F   S Y M B O L S
% -----------------------------
% To include a Nomenclature section
% \addcontentsline{toc}{chapter}{\textbf{Nomenclature}}
% \renewcommand{\nomname}{Nomenclature}
% \printglossary
% \cleardoublepage
% \phantomsection % allows hyperref to link to the correct page
% \newpage

% Change page numbering back to Arabic numerals
\pagenumbering{arabic}

%----------------------------------------------------------------------
% MAIN BODY
%----------------------------------------------------------------------
%
%------------------------- Intro ---------------------------------------
%----------------------------------------------------------------------
% INTRODUCTION
%----------------------------------------------------------------------
% Here I introduce the idea behind best matching and outline the thesis.
%
%-----------------------------------------------------------------------
\chapter{Introduction}
%-----------------------------------------------------------------------

I recently spoke to a group of grade 9 students at a local high school about the wonders of math and physics. Their teacher, who is a good friend of mine, had just assigned them their Problem of the Day, which was to identify the most ``square'' object out of a collection of differently shaped rectangles. Their first task was to rank the objects in order of ``squareness'', then to find a mathematical criterion for determining the ``squareness'' of an object. This thesis is a general solution to that problem.

What is truly remarkable is that the solution to such a simple problem leads to a new theory of gravitation that is equivalent to general relativity (GR) but has different symmetries that treat local shapes as the irreducible physical degrees of freedom. This new theory represents a fresh starting point for quantum gravity, free of the problem of time. It provides a conformal framework for understanding gravity that is ideally suited for understanding gauge/gravity dualities and a new computational framework for doing cosmology.

Indeed, to achieve such a theory, we will need to solve a slightly more general problem than that posed by my friend to his students: how to quantify the ``difference'' between the \emph{local} shapes formed by configurations of matter in the universe. By ``local'' shapes we mean the shape of objects, treated individually, that are finitely separated in space. The extent of these separations and the definition of the local neighborhood is an issue we will address shortly. Because we want our construction to be as general as possible, we want our definition to work no matter what kind of matter is being considered and what kind of shapes are being formed. This can be achieved by manipulating the space that the imagined shapes live in and by identifying those manipulations that can actually \emph{change} the local shapes. A moment's reflection reveals what manipulations do \emph{not} change the local shapes: coordinate transformations and local rescalings of the spatial metric, or \emph{conformal transformations}. This is because local shape does not depend on position and orientation, which are equivalent to local infinitesimal coordinate transformations, or changes of the local scale. Thus, to solve the general problem that my friend set to his students, we need to find a way to quantify the ``difference'', or ``distance'', between two conformal geometries. Since, mathematically, a metric is what gives a notion of distance, our task is to define a \emph{metric} on the space of conformal geometries, also known as \emph{conformal superspace}, or simply \emph{shape space}.

One may ask what purpose this metric could serve. To address this question, consider the nature of time in classical physics. In our classical experience of the world, time is undoubtedly what flows when genuine change occurs. \emph{Dynamics} is a way of predicting what will change when time flows. Therefore, the key to defining dynamics is identifying a way of quantifying how much genuine change has occurred. We will make a choice, which we will motivate with Mach's principles,\footnote{The plural in ``principles'' is used because we will distinguish between \emph{two} key physically distinct ideas of Mach to motivate our choice.} that places local shapes as the fundamental empirically meaningful quantities in Nature. With this choice, genuine change is given by the change of local shapes. We can then use our metric on shape space to define a dynamics for conformal geometry. Remarkably, it is possible to define a theory of \emph{shape dynamics} in this way that is dynamically equivalent to GR.

\subsubsection{What shape dynamics is}

Shape dynamics is a theory of dynamical conformal geometry that reproduces the known physical solutions of GR. It was discovered by requiring that local shapes represent the physical degrees of freedom of the gravitational field, a requirement directly inspired by Mach's principles. The simplest way to understand the connection between shape dynamics and GR is to think of it as a duality whose mechanism is similar to the mechanism behind $T$--duality in string theory. There exists a kind of \emph{parent} theory, which we feel is more appropriately called a \emph{linking} theory in this context, that is defined on a larger phase space. Shape dynamics and GR represent different gauge fixings of this linking theory and, for that reason, make the same physical predictions.

An equivalent way of understanding the move from GR to shape dynamics is as a dualization procedure that trades one symmetry for another. To understand this trading, we must first understand the symmetry in GR that is traded. GR is a spacetime theory invariant under 4--dimensional coordinate transformations, or diffeomorphisms. However, it is possible to express GR as a theory of dynamic 3--geometry by restricting the spacetime manifold to have a topology $\Sigma \times \mathbb R$, where $\Sigma$ is an arbitrary 3--dimensional manifold. With this topology, the spacetime can be sliced by spacelike hypersurfaces that foliate it. Of course, because of 4d diffeomorphism invariance, there are many choices of foliation leading to the same 4d geometry. In the Hamiltonian formulation of GR, this invariance under refoliations appears as a local gauge symmetry of the theory generated by the Hamiltonian constraint. But the symmetries generated by the Hamiltonian constraint have a split personality: on one hand, they represent local deformations of the spacelike hypersurfaces while, on the other hand, they represent global reparametrizations of the parameter labeling the hypersurfaces. There is, thus, a qualitative difference between the local part of the Hamiltonian constraint, generating refoliations, and the global part, generating reparametrizations. Unfortunately, these two different roles cannot be untangled in general because each choice of foliation requires a different split of the Hamiltonian constraint. This dual nature of the Hamiltonian constraint is the origin of the problem of time.

Shape dynamics can be constructed by trading the refoliation invariance of GR for conformal invariance. The dual nature of the Hamiltonian constraint is resolved by fixing a particular foliation in GR where the split between refoliations and reparametrizations is made. The split personality is resolved by trading \emph{all but} the part of the Hamiltonian constraint that generates global reparametrizations. This means that we must keep the particular linear combination of the Hamiltonian constraint of GR that generates reparametrizations in the foliation we have singled out. In turn, this implies that shape dynamics will be missing one particular linear combination of conformal transformations that corresponds to part of the Hamiltonian constraint we are keeping. This turns out to be the global scale. Since the invariance of GR under 3d diffeomorphisms is untouched (and, as it turns out, unaffected by the trading procedure), we are led to the following picture for shape dynamics: it is a theory with a global Hamiltonian that generates the evolution of the 3--metric on spacelike hypersurfaces. This evolution is invariant under 3d diffeomorphisms and conformal transformations that preserve the global scale. In the case where $\Sigma$ is a compact manifold without boundary, the conformal transformations must preserve the total volume of $\Sigma$. 

The symmetry principle in shape dynamics is considerably cleaner than that of GR. Conceptually, this is clear because refoliation invariance leads, for instance, to relativity of simultaneity, which is more challenging to conceptualize than local scale invariance. More generally, there is no \emph{many--fingered} time in shape dynamics. Time is simply a \emph{global} parameter that labels the spacelike hypersurfaces. Thus, there is no local problem of time. There is still a global problem of time associated with the reparametrization invariance but this problem is considerably easier to deal with. There are also technical simplifications. As we will see, the conformal constraints are \emph{linear} in the momenta in contrast to the Hamiltonian constraints of GR, which are \emph{quadratic} in the momenta. Aside from avoiding operator ordering ambiguities in quantum theory, linear constraints can form Lie algebras. This implies that group representations can be formed simply by exponentiating the local algebra, a drastic improvement over GR. There is a price to pay for these simplifications. The global Hamiltonian of shape dynamics is a non--local functional of phase space. From the point of view of the linking theory, this non--locality is the result of the phase space reduction required to obtain shape dynamics. However, the non--locality is simply a technical challenge and not a conceptual one. In this thesis, we will give some examples where this technical challenge can be overcome.

It cannot be overemphasized that shape dynamics is a gauge theory in its own right and \emph{not} just a gauge fixing of GR. From this perspective, shape dynamics is not a solution to the problem of time of GR but rather a formulation of gravity that is itself free of the problem of time. Although it is true that the first step of the dualization procedure leading to shape dynamics involves fixing a particular spacetime foliation, shape dynamics has a conformal gauge symmetry that GR does not have. This means that there are gauges in shape dynamics that do \emph{not} correspond to the solutions of GR, although they are gauge equivalent. For example, it is always possible to fix a gauge in shape dynamics on compact manifolds without boundary where the spatial curvature is constant. This gauge provides a valuable computational tool that we will exploit to solve the local constraints of shape dynamics.

There are other important differences between shape dynamics and GR resulting from having to fix a foliation to use the dictionary. We will see that the particular foliation that needs to be fixed is such that the spacelike hypersurfaces have constant mean curvature (CMC) in the spacetime in which they are embedded. CMC foliations are used extensively in numerical relativity and are known to foliate many of the physical solutions of GR.\footnote{Precisely which solutions of GR are excluded in shape dynamics is an important but difficult question to answer and is beyond the scope of this thesis. We will, thus, leave precise statements for future investigations.} It is only in CMC gauge where a general procedure for solving the initial value constraints is known to exist and to be unique.\footnote{The same mechanism behind the existence and uniqueness proofs of the initial value problem (see \cite{Niall_73}) is used to prove the existence and uniqueness of the shape dynamics Hamiltonian.} However, not all solutions to GR are CMC foliable. Many of these, like those with closed timelike curves, are clearly unphysical. However, it is still possible that our universe is \emph{not} CMC foliable. Thus, CMC foliability of the universe is a \emph{prediction} of shape dynamics. By excluding potentially unphysical solutions of GR and by providing a cleaner symmetry principle, shape dynamics may have a simpler quantization than GR.

\subsubsection{What shape dynamics may be}

We have just described shape dynamics as a theory of dynamic conformal geometry. The form of the global Hamiltonian used to generate this dynamics is specifically chosen so that theory will make the same predictions as GR. The key new feature introduced by this global Hamiltonian is non--locality. Although, the entire causal structure of GR is encoded in this one global object, the precise interplay between the non--locality of shape dynamics and the causal structure of GR is still a mystery. I believe that unraveling this mystery could be the key to understanding how to quantize gravity.

What we seem to be missing is a further principle to help construct the shape dynamics Hamiltonian without having to rely on GR. It's not clear what such a principle could be but somehow it should impose on shape dynamics the information about the causal structure of the spacetime in the GR side of the duality. In addition, it is reasonable to hope that this new principle will also suggest a way to quantize shape dynamics without a notion of locality. This is a question of utmost importance because of the necessity of a locality principle in quantum and effective field theory.\footnote{It is possible to work in the linking theory which \emph{is} local. However, the linking theory has the same problem of time as general relativity because its constraint algebra contains the hypersurface deformation algebra as a subgroup.} Unfortunately, we do not yet have such a principle.

One possibility, which deserves further exploration, is to revisit the ambiguity in defining \emph{local} shapes mentioned early in this discussion. The original motivation for introducing conformal symmetry was that only local shapes are empirically meaningful. To be more precise, all measurements of length are local comparisons. However, in order to make sense of this observation we need to be precise about how we actually measure the shape degrees of freedom. Concretely, we can imagine that our universe is filled with point particles and that these particles are clumped into small groups that form local shapes. We can make our statement more precise by imagining that we have at our disposal a small system of two (or possibly more) particles that we can use a ruler. If the system we are trying to study is \emph{large} compared with the length of this ruler, then we can define the local shape degrees of freedom as the quantities that can be measured in the system by comparing them to the ruler. Using this definition, it is clear that no local measurement of length made with the ruler will change if we perform a local scale transformation. As we move the ruler from one clump of particles to another, the ruler gets rescaled along with the new clump. If the scale factor varies significantly over the extension of the ruler, then the infinitesimal segments of the ruler will simply get rescaled along with the infinitesimal segments of the system we are comparing to. However, if the system is \emph{small} compared with the size of the ruler, then there may be shape degrees of freedom that cannot be resolved by the ruler: what one ruler sees as two distinct particles a coarser ruler may only see as one. For example, on galactic scales, the solar system is but a point. It is only on smaller scales that one can resolve the planets or, smaller still, the moons, mountains, people, insects, etc... Thus, the shape of the universe changes as the size of the ruler changes. 

The fact that the local shapes resolved in experiments depend strongly on the resolution used to make measurements of these shapes suggests that renormalization group (RG) flow could play an important role in our understanding of shape dynamics. Indeed, it may be possible to exactly mimic the flow of time in shape dynamics by the change in shape resulting from RG flow. Concretely, Hamiltonian flow in shape dynamics could be represented as RG flow in a theory with no time. The conformal constraints of shape dynamics act, in the quantum theory, like the conformal Ward identities of a conformal field theory (CFT). This suggests that shape dynamics may be the ideal theory of gravity for formulating dualities between gravity and CFT. We will show that it is possible to construct a holographic RG flow equation for shape dynamics similar to what is done in standard approaches to the AdS/CFT correspondence. Exploring these connections further may both lead to a deeper understanding of the holographic principle and also may provide a way of \emph{defining} shape dynamics through holographic RG flow in a CFT.

There is one final potentially interesting connection worth noting. Shape dynamics has the same local symmetries as the high energy limit of Ho\v rava--Lifshitz gravity. Interestingly, it is this symmetry that leads to the power counting renormalizability arguments. This is because the conformal symmetry singles out the square of the Cotton tensor as the lowest dimensional term allowed in a quasi--local expansion of the action. However, this term has 6 spatial derivatives compared with the 2 time derivatives in the kinetic part of the action, leading to the $z=3$ anisotropic scaling of the theory. The stability problems of Ho\v rava's theory are avoided in shape dynamics because non--local terms are allowed in the Hamiltonian (this also allows for \emph{exact} equivalence with GR). These stability problems appear in the theory because of the appearance of an extra propagating degree of freedom. This degree of freedom does not appear in shape dynamics because the foliation invariance is simply \emph{traded} for the conformal symmetry. Thus, the local propagating degrees of freedom of shape dynamics are identical to those of GR. Unfortunately, the non--locality also forbids the use of the perturbative power counting arguments to argue that the theory is finite. Nevertheless, it may still be true that the conformal symmetry protects shape dynamics in the UV. Although the non--perturbative renormalizability of GR remains an open question, shape dynamics has a different symmetry. Thus, the question of finiteness of quantum gravity may be more easily addressed in the shape dynamics framework.

\section{Basics}

In Section~(\ref{sec:sd derivation}), we derive shape dynamics using a dualization procedure that we apply to GR. Then, we devote the entire \ref{chap:shape_dynamics}$^\text{th}$ chapter to studying shape dynamics in detail. Nevertheless, it is useful to give an intuitive summary of our results here without attempting to prove anything rigorously.

There are two helpful pictures to keep in mind when trying to understand how shape dynamics is defined. The first is to think of shape dynamics and GR as different theories living on different intersecting surfaces in phase space. The second is to picture them as being different gauge fixings of a larger linking theory. The first is often convenient for conceptualizing while the second is essential for proving things rigorously.

\subsection{Intersecting surfaces}

As has been discussed, shape dynamics is a theory of evolving conformal geometry. The evolution is generated by a global Hamiltonian that has a flow on the constraint surface in phase space generated by 3d diffeomorphism and conformal constraints. The conformal constraints have one global restriction corresponding to the volume preserving condition. This nearly specifies the constraint surface. The remaining task is to find a global Hamiltonian that leads to a dynamics equivalent to that of GR.

In the Hamiltonian formulation, GR is a theory of dynamic geometry whose flow is generated by the diffeomorphism constraints and the usual local Hamiltonian constraints. What can be shown is that the Hamiltonian constraints can be \emph{partially} gauge fixed by the volume preserving conformal constraints. This turns out to be a gauge where the spacelike foliations are CMC. Because of the volume preserving condition, there is still one degree of freedom of the Hamiltonian constraints that is \emph{not} gauge fixed. This is the CMC Hamiltonian.

To be precise, let us call $D\approx 0$ the constraint surface defined by the volume preserving conformal constraints, $\tilde S \approx 0$ the part of the Hamiltonian constraint that is gauge fixed by the CMC condition ($S\approx 0$ would be the full Hamiltonian constraint), $\hg$ the global Hamiltonian of shape dynamics, and $\mathcal H_\text{CMC}$ the CMC Hamiltonian. Geometrically, the fact that $D\approx 0$ is a gauge fixing of $\tilde S \approx 0$ means that the two surfaces have a common intersection that selects a single member of the gauge orbits of $\tilde S \approx 0$. Since the diffeomorphism constraints are common to both theories, they can be trivially taken into account when comparing them. Figure~(\ref{fig:inter surfaces}) represents shape dynamics as living on the surface $D\approx 0$ with flow generated by $\hg$ and GR as living on the surface $\tilde S \approx 0$ with flow generated by $\mathcal H_\text{CMC}$. The common intersection represents GR in CMC gauge.
\begin{figure}
     \begin{center}
	\includegraphics[width=0.55\textwidth]{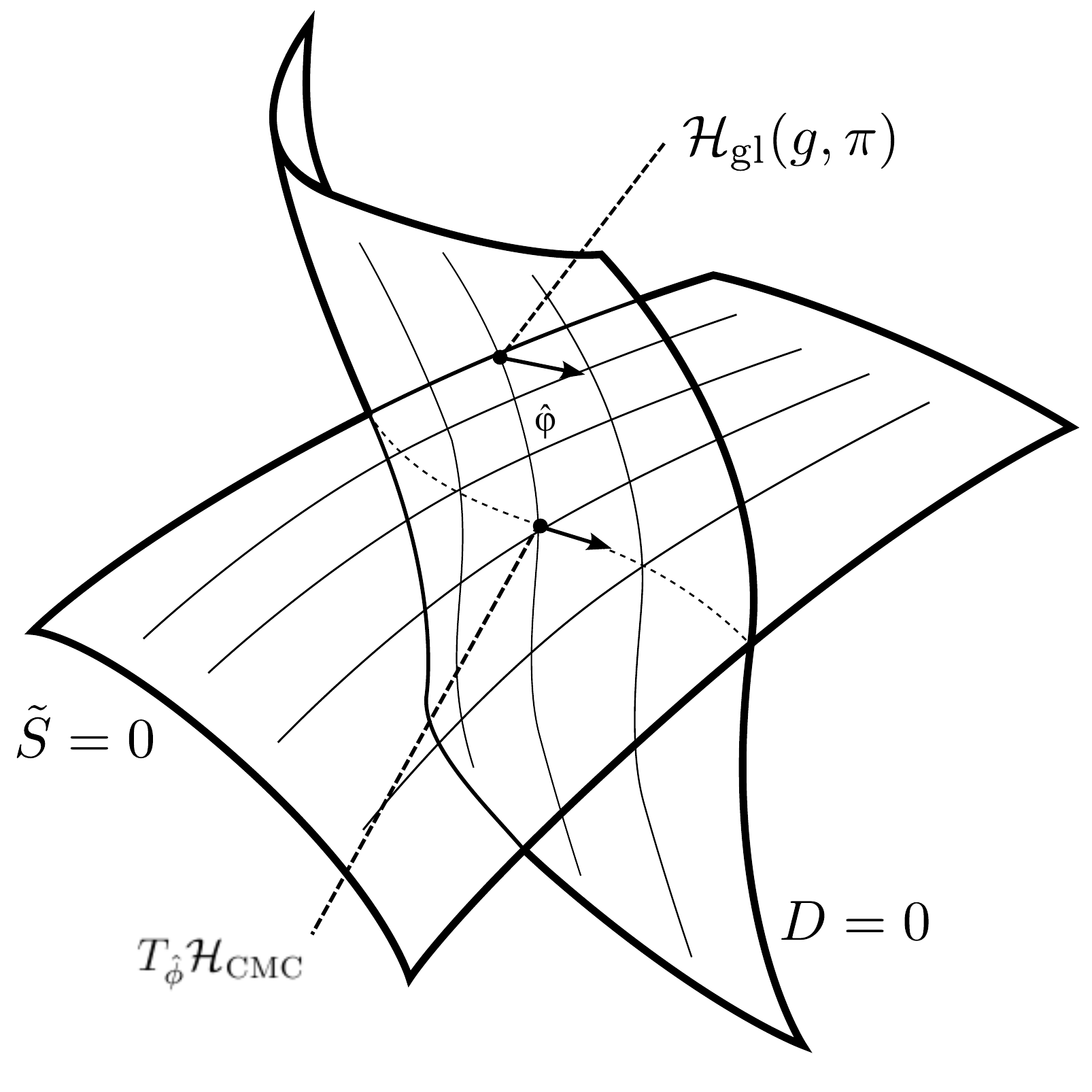}
	\caption{Shape dynamics, represented by $D\approx 0$, intersects GR, represented by $\tilde S \approx 0$. The gauge orbits of both surfaces are shown.}\label{fig:inter surfaces}
     \end{center}
\end{figure}

It is now possible to understand how $\hg$ is defined. Because the constraint surface $D\approx 0$ is integrable, it is possible to move orthogonally to the intersection by moving along the gauge orbits generated by $D$. These are volume preserving conformal transformations in phase space. Thus, for any point on $D\approx 0$, it is possible to find a unique volume preserving conformal transformation $\hat\phi$ that will bring you to the intersection. Then, one can guarantee that $\hg$ is both first class with respect to the $D$'s and generates the same flow as GR by defining it everywhere on $D\approx 0$ to be equal to the value of $\mathcal H_\text{CMC}$ at the intersection. In other words,
\begin{equation}
    \hg = T_{\hat\phi} \mathcal H_\text{CMC},
\end{equation}
where $T_{\hat\phi}$ signifies a volume preserving conformal transformation on phase space. This definition is illustrated in Figure~(\ref{fig:inter surfaces}). This picture is a very useful way to think of the relationship between shape dynamics and GR. We will show it again in Section~(\ref{sec:sd basic}) where we will be much more careful and complete with our definitions.

To understand the dictionary between shape dynamics and GR, note that different gauge fixings of GR are simply different sections of the surface $\tilde S \approx 0$ while different gauge fixings of shape dynamics are different sections of $D\approx 0$. Thus, it is always possible to take a solution of a given theory and use a pair of gauge transformations to express it as an arbitrary equivalent solution of the other theory.

\subsection{Linking theory}

The linking theory provides a powerful tool both for conceptualizing and for rigorously proving many of the statements made in the previous section. The idea is to treat GR and shape dynamics as different gauge fixings of a theory on an enlarged phase space. The dictionary can further be established by choosing the appropriate gauge fixing that brings the solutions to the intersection. Figure~(\ref{fig:link sketch}) shows a diagram of these relations.
\begin{figure}
    \begin{center}
	\includegraphics[width= 0.6\textwidth]{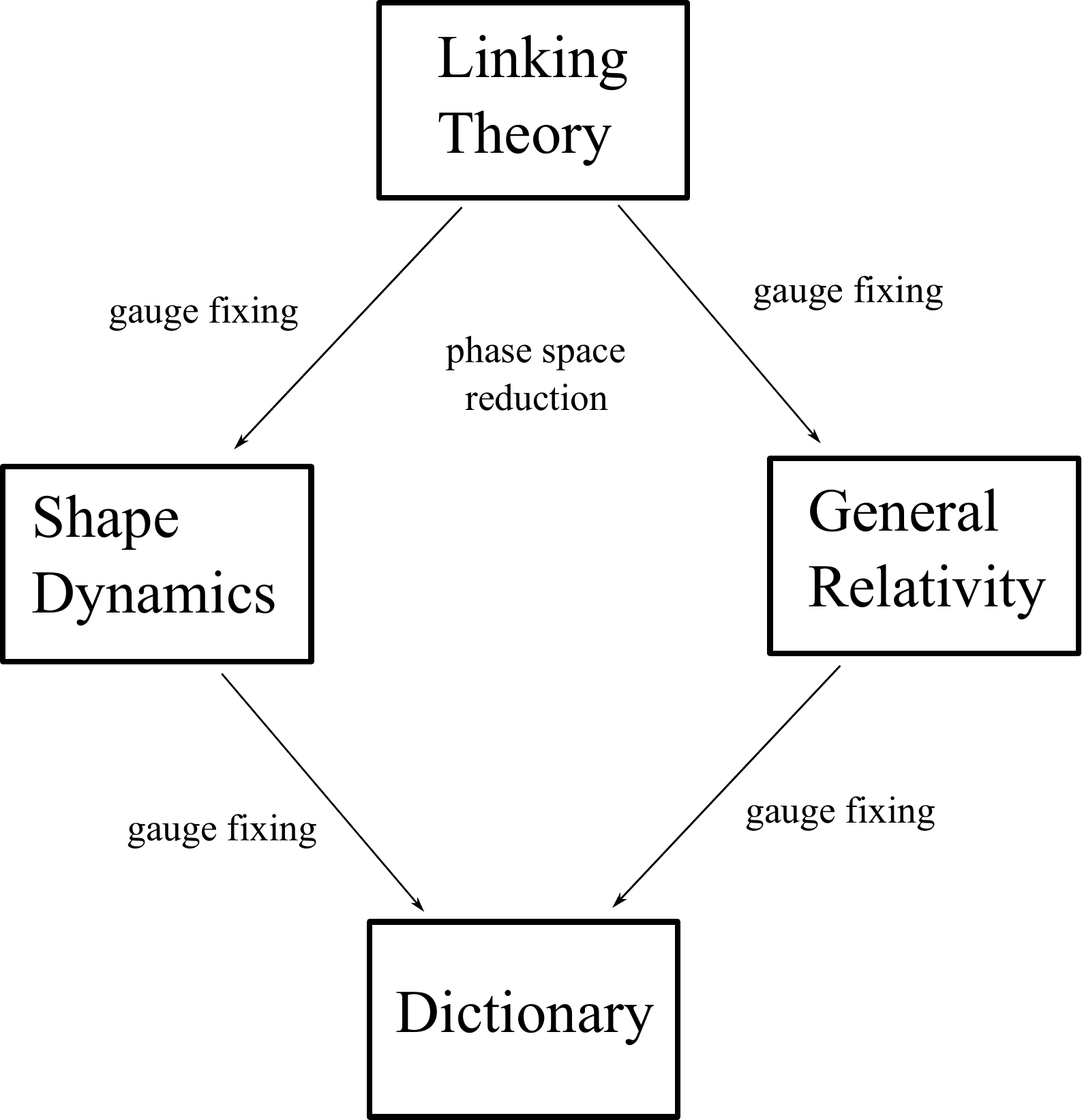}
	\caption{Gauge fixings lead from the linking theory to GR, shape dynamics, and, ultimately, the dictionary relating them.} \label{fig:link sketch}
    \end{center}
\end{figure}
 
The definition of the linking theory can be understood in the following way. Consider the usual diffeomorphism and Hamiltonian constraints of GR. Now consider trivially extending the phase space by including the variable $\phi$ and its conjugate momentum $\pi_\phi$ without changing the original constraints. Adding the first class constraint $\pi_\phi\approx 0$ does nothing to the dynamics of the original variables but ensures that $\phi$ and $\pi_\phi$ are auxiliary. This is the linking theory. It is convenient to perform a canonical transformation on this theory that puts it in a more useful form. This canonical transformation is a volume preserving conformal transformation on phase space (the precise definition is given in Section~(\ref{eq:first link theory})).

We will see that, under this canonical transformation,
\begin{equation}
    \pi_\phi \to \pi_\phi - D.
\end{equation}
 Thus, if we fix a gauge for the transformed Hamiltonian constraint using the gauge fixing condition $\pi_\phi \approx 0$, the first class constraint
\begin{equation}
\pi_\phi - D \approx 0 ~\to ~ D \approx 0,
\end{equation}
and the $\tilde S$'s have been traded for the $D$'s. Performing a phase space reduction leads immediately to shape dynamics. It is easy to see that GR can be obtained from the linking theory by imposing $\phi \approx 0$ as a gauge fixing condition for $\pi_\phi - D \approx 0$, then performing a phase space reduction.

\section{Outline}

One of most promising aspects of shape dynamics is that it is based on simple foundational principles. The road to shape dynamics starts with a careful formulation of Mach's principles. In Chapter~(\ref{chap:bm_foundations}), we start by clearly stating what we believe to be an accurate formulation of Mach's principles that captures the essence of Mach's ideas. Using this, we develop a procedure, called \emph{best matching}, that provides a principle of dynamics that can be used to define relational theories. We first illustrate the procedure by showing how it can be used to construct a simple relational particle model; then, we develop the general procedure. What best matching suggests is that relational theories should be thought of as gauge theories on configuration space and that the dynamics should be given by a geodesic principle on configuration space. We use a historically relevant example, Newton's bucket, to illustrate how best matching eliminates absolute space and time and rigorously implements Mach's ideas.

Chapters (\ref{chap:equiv bm}) and (\ref{chap:non_equivariant_bm}) are devoted to exploring the detailed structure of best matching. In particular, we develop the canonical formulation of best matching. This allows us to clearly identify how best matching removes absolute structures by introducing special constraints. In some cases, these constraints lead to standard gauge symmetries. In other cases, these constraints provide gauge fixings that lead to a dualization procedure. We will separate these distinct cases by exploring the first in Chapter~(\ref{chap:equiv bm}) and by developing the second in Chapter~(\ref{chap:non_equivariant_bm}).

In both cases, we will study general finite dimensional models. These models are important for several reasons: 1) they provide toy models for the geometrodynamic theories we will study later, 2) they can be worked out explicitly, 3) they provide limiting cases of the geometrodynamic models. In some cases, they provide interesting models in their own right. For example, the general finite dimensional models developed in Section~(\ref{sec:finite dim models}) can be used to study mini--superspace cosmologies. However, one of the most important reasons for studying the toy models is to build intuition for the technically more challenging geometrodynamic models we will present later. The simplicity of the toy models brings to light the key aspects of best matching, free of other technical distractions. This will build an arena for conceptualizing that will be useful for the more subtle field theories.

In Chapter~(\ref{chap:equiv bm}), we exploit the new understanding of relational theories provided by best matching to motivate a specific definition of background independence. This definition is then used to study a toy model with a global problem of time. In Chapter~(\ref{chap:non_equivariant_bm}), we develop a dualization procedure for trading symmetries. This provides a good model for the dualization procedure used in geometrodynamics to derive shape dynamics and introduces the idea of a \emph{linking theory}.

In Chapter~(\ref{chap:bm_geo}), we use best matching to construct relational theories where the metric of space is dynamic. We consider three different classes of theories. The first is a na\"ive generalization of best matching as it applies to particle models. We will show that the notion of locality in these models is not restrictive enough to lead to a sensible theory. We will then introduce a modification to the na\"ive best matching principle that leads to a local action. Using this principle it is possible derive GR. Indeed, it is even possible to use our previous definition of background independence to solve the global problem of time by introducing a background global time. Our proposal naturally leads to unimodular gravity. Finally, we will use best matching to construct a conformally invariant geometrodynamic theory that is equivalent to GR. This procedure implements local scale invariance by following the dualization procedure studied in the toy models. The result is shape dynamics, which allows for non--locality but is nevertheless restrictive enough to produce a well defined theory.

We conclude our discussion in Chapter~(\ref{chap:shape_dynamics}) by examining in more detail the structure of shape dynamics. The goal is to understand shape dynamics better as a theory in its own right. We start by describing the global Hamiltonian, then compute it using two different perturbative expansions. The first is an expansion in large volume. This expansion is useful for understanding the connection between shape dynamics and CFT. The second is an expansion of fluctuations about a fixed background. This expansion is useful for doing cosmology using shape dynamics. We end with a calculation of the Hamilton--Jacobi functional in the large volume limit. This result is used to construct the semi--classical wavefunction of shape dynamics and establish a correspondence between shape dynamics and a timeless CFT. Further explorations of this correspondence may provide a deeper understanding of the AdS/CFT correspondence and the meaning of shape dynamics.

\section{History}

Shape dynamics began with the development of best matching in Barbour and Bertotti's original 1982 paper \cite{barbourbertotti:mach} and has taken more than one unanticipated turn since that time. The most important champion of this approach has undoubtedly been Julian Barbour, who has enthusiastically encouraged the development of this idea from its inception to its current form. The best--matching procedure was developed for particle models in \cite{barbour:scale_inv_particles,gergely:geometry_BB1,gergely:geometry_BB2} and many papers by Ed Anderson including \cite{anderson:found_part_dyn,anderson:rel_part_mech_1,anderson:rel_part_mech_2,anderson:triangleland_new,anderson:triangleland_new_2,anderson:triangleland_old}. In geometrodynamics, best matching was used to construct GR in \cite{barbourbertotti:mach,barbour_el_al:rel_wo_rel,anderson:rel_wo_rel_vec,Anderson:cyclic_ADM}. The last two papers present a powerful construction principle for GR.

To the best of my knowledge, the first paper proposing to look for a 3d conformally invariant geometrodynamic theory using best matching was \cite{Barbour_Niall:first_cspv}. In this paper, Niall O'Murchadha proposed non--equivariant best matching and applied it to the full group of conformal transformations in GR. These ideas were elaborated on in \cite{barbour:scale_inv_particles,barbour_el_al:scale_inv_gravity}, then refined in \cite{barbour_el_al:physical_dof} where the volume preserving condition was introduced. Finally, in \cite{Barbour:new_cspv}, the observation was made that the global scale could be replaced by a ratio of volumes. In these papers, the special variation used in best matching (which will be introduced in Section~(\ref{sec:bm variation})) was treated as a kind of gauge fixing condition for the lapse. However, the canonical analysis was incomplete and there were no clues that a dual theory could be constructed from a phase space reduction. The main observation was that a geodesic principle could be defined on conformal superspace that reproduced the predictions of GR in CMC gauge. However, this connection was restricted to a gauge fixing of GR, which we now understand as the intersection of shape dynamics and GR. A summary of these approaches with an excellent description of the conceptual motivations from best matching is given in the short review \cite{JuliansReview}. For an interesting alternative approach to 3d conformal invariance in geometrodynamics, see \cite{Westman:3d_weyl}.

The ideas presented in these papers were inspired by York's solution to the initial value problem \cite{York:york_method_prl,York:cotton_tensor,York:york_method_long}, which used conformal transformations and the CMC gauge of GR to find initial data that solve the Hamiltonian and diffeomorphism constraints of GR. Indeed, the existence and uniqueness theorems developed in \cite{Niall_73} for the solutions of the initial value problem using this approach were a vital inspiration for the uniqueness and existence theorems used to develop shape dynamics.

The current form of shape dynamics was discovered by Henrique Gomes, Tim Koslowski, and myself when we realized that a phase space reduction of what we now call the linking theory would leave a theory invariant under volume preserving conformal transformations. We published our results in \cite{gryb:shape_dyn}. Shortly after, Gomes and Koslowski discovered the linking theory \cite{Gomes:linking_paper}, which significantly helped to clarify the presentation of the dualization procedure. Since then, with input from Flavio Mercati, we have published a calculation of the Hamilton--Jacobi functional in the large volume limit \cite{gryb:gravity_cft}, which proposed a new approach to the AdS/CFT correspondence \cite{LargeReview,Maldacena:ads_cft,Witten,Freidel} and the holographic RG flow equation \cite{Verlinde,Skenderis:holo_RG,Skenderis:holoRG_main,Skenderis:holo_weyl}. There has been much unpublished work on matter coupling, perturbation theory, and Ashtekar variables that has relied on valuable input from Timothy Budd and James Reid (on top of the authors already mentioned). This work should be appearing in the literature shortly. The quantization of $(2+1)$ shape dynamics in metric variables can be found in \cite{Budd:2_plus_1_sd}.

The material of this thesis was based on the work presented in \cite{sg:ym_bm,sg:dirac_algebra,gryb:shape_dyn,gryb:gravity_cft}. However, I have adapted some of the presentation of the results of \cite{gryb:shape_dyn} to include the insights of \cite{Gomes:linking_paper}. In regards to best matching, for pedagogical reasons I have often summarized my own understanding of the procedure based on my reading of the papers above and discussions with Julian Barbour. I hope that this provides a useful new perspective. However, there are new contributions worth noting. The possibility of treating best matching as a gauge theory on configuration space was noticed in \cite{sg:ym_bm} and \cite{Gomes:gauge_theory_riem}. However, a first principles derivation of the best--matching connection and its explicit calculation in toy models is part of a work in preparation by Barbour, Gomes, and myself. I have included these concepts in this thesis. Finally, the complete canonical formulation of equivariant best matching was first presented in \cite{sg:dirac_algebra} and is the main subject of Chapter~(\ref{chap:equiv bm}).

%------------------ FOUNDATIONS OF BEST MATCHING ----------------------
%----------------------------------------------------------------------
% FOUNDATIONS OF BEST MATCHING
%----------------------------------------------------------------------
% This section develops the concepts at the foundation of best matching.
% This is my understanding of best matching which is a procedure that
% originated from Julian Barbour. I should remember to cite relevant
% references and deliniate what is review from what is new.
%
% Jun 2: This is VERY rough... the order will change considerably!
%
%======================================================================
\chapter{Foundations of best matching}\label{chap:bm_foundations}
%======================================================================

During my first visit to Julian Barbour's historic home, College Farm, in Northern Oxfordshire, I asked the inventor of best matching to explain to me the basic idea behind the procedure. He started by drawing two different triangles on his white board and announced: ``the idea behind best matching is to find the \emph{difference} between two different shapes.'' Remarkably, this simple idea is all that is needed to construct a general framework for producing relational models based on Mach's principles that is capable of deriving general relativity and of revealing its conformal dual: shape dynamics.

To see how this is possible, we must proceed step by step. First, we will try to understand the problem that best matching claims to solve. To do this, we will look carefully at a well known example: Newton's bucket. This famous example illustrates the differences between absolute and relative motion and how best matching responds to Newton's arguments for absolute space. After studying the problem, we will show how best matching manages to ``find the \emph{difference} between two shapes.'' We will illustrate this with a simple example of a system of particles moving in 2 dimensions. This example illustrates many of the key features of best matching. We will use it to motivate a general formulation of best matching for finite dimensional systems. Our analysis will suggest that relational theories are best thought of as gauge theories on configuration space. We will see precisely how this beautiful geometric picture emerges.

It will be convenient to distinguish between two different kinds of relational theories: those whose metric on configuration space is constant as we compare different physically equivalent configurations (for reasons that will become clear later, this will correspond to the \emph{equivariant} case) and those whose metric is not constant (this will correspond to the \emph{non--equivariant} case). Equivariant theories are always consistent in the finite dimensional case but non--equivariant theories need extra conditions in order to be consistent. Because of these extra complications, which are crucial to understanding the relation between shape dynamics and general relativity, we will treat non--equivariant theories in a separate section.

\section{Newton's Bucket}

To illustrate the key features of best matching, it is instructive to review an historically relevant example that illustrates the difference between absolute and relative space. The example is commonly know as \emph{Newton's bucket} and was first introduced in Newton's \emph{Principia}. Newton's bucket is a bucket half--filled with water suspended by a rope that can be wound tightly by spinning the bucket around an azimuthal axis. A modern version can be can be crafted from a piece of string and a pickle jar with two holes punched into the lid (see Figure~(\ref{fig:pickel_jar})).
\begin{figure}
    \begin{center}
	\includegraphics[width = 0.3\textwidth]{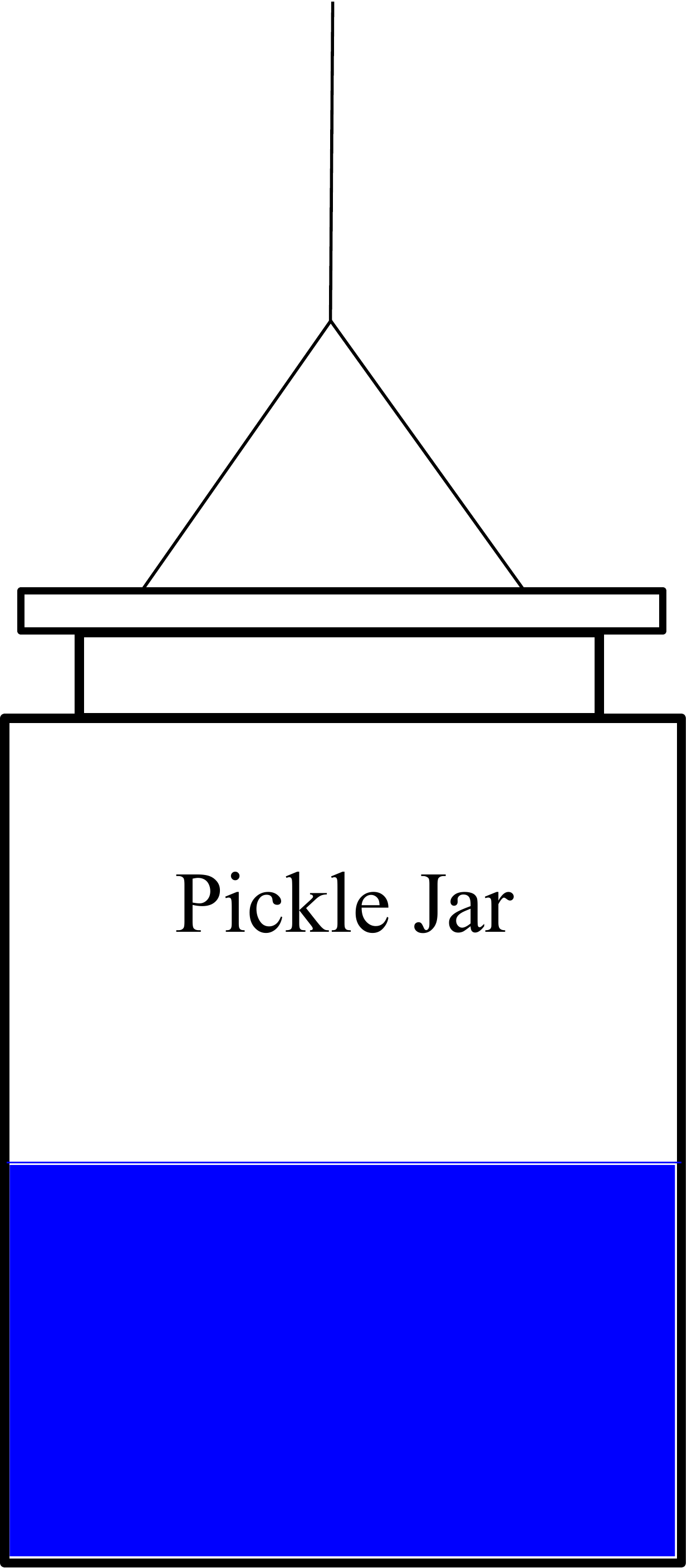}
	\caption{Newton's bucket made with a pickle jar, string, and water.} \label{fig:pickel_jar}
    \end{center}
\end{figure}

The experiment compares the motion of 3 different reference frames: the lab frame, the bucket, and the water. There are 4 simple steps:
\begin{enumerate}
    \item No motion between the lab, bucket, and water.
    \item The bucket is quickly spun so that the water remains static with respect to the lab.
    \item Over time, the bucket and the water spin together.
    \item The bucket is quickly stopped from spinning while the water continues to spin.
\end{enumerate}
The observable phenomenon is the shape of the surface of the water in the bucket, which can be either flat or curved up the walls of the bucket.

It is easy enough to imagine the outcome of this experiment. In the first step, the water will be flat because nothing is happening. In the second step, the water stays flat because it hasn't started to spin yet even though the bucket is. In the third step, the water begins to spin and creeps up the surface of the bucket. In the last step, the water is still curved because it is spinning even thought the bucket has been stopped. Newton uses this to argue that the relative motion between the water and bucket clearly has no impact on the physically observable phenomenon, which is the shape of the surface of the water. It is clear from Table (\ref{tab:newton bucket}), which summarizes the results, that there are always two possible outcomes for each type of relative motion.
\begin{table}\label{tab:newton bucket}
    \begin{center}
	\begin{tabular}{|l|l|l|l|l|}\hline
	    × & 1 & 2 & 3 & 4\\\hline
	    Relative Motion & no & yes & no & yes\\\hline
	    Surface of Water & flat & flat & curved & curved\\\hline
	\end{tabular}
    \end{center}
    \caption{The results of Newton's bucket experiment.}
\end{table}
This implies that the relative motion of the bucket and water does not explain the observed phenomena. Instead it is the only the motion of the water with respect to the lab frame that determines the shape of the water. Newton concludes that the lab is at rest with respect to \emph{absolute space} and that only motions with respect to absolute space are meaningful.

In \emph{The Mechanics} \cite{mach:mechanics}, Mach provides an objection to this argument. He notices that the lab is effectively at rest with respect to the ``fixed stars'' (which we now know to be galaxies). He points out that what Newton's bucket experiment shows is that it is only the \emph{relative} motion of the water with respect to the fixed stars that determines the shape of its surface. But, the observable phenomena should not depend on the relative motion of \emph{only} the bucket and the water but on the relative motion of the bucket and \emph{everything} else in the universe. This obviously includes the fixed the stars that are considerably more massive than the bucket. This extra mass should make the relative motion of the water and the fixed stars more significant than that of the water and bucket, leading to the observed results of the experiment.

Mach, unfortunately, did not provide a precise framework for testing this hypothesis nor did he provide a specific theory that would explain how the massive stars have more impact on the behavior of the water. Nevertheless, the intuitive argument is clear: \emph{all} relative motions between bodies must be considered and those bodies with greater mass have a more significant impact on the overall behavior of the system. This reasoning greatly influenced Einstein and played an important role in the development of GR. Indeed, one could say that Mach predicted the \emph{frame dragging} effects that occur in GR.

In the next sections, we will develop best matching. As we do, it will become clear that best matching provides a precise framework for implementing Mach's explanation of Newton's bucket experiment. The procedure produces a theory that explains exactly how the stars provide the illusion of absolute space. The starting point for this framework is Mach's principles, which we will now state.

\section{Mach's Principles}

It is difficult to find agreement on the exact manner in which to state Mach's principles. There are many different versions that exist in the literature, all based on different interpretations of Mach's writings. Since Mach did not clearly state what his principles are, we have some liberty in how we define them. In this work, we will adapt a definition based on the one carefully outlined in \cite{Barbour:DefMach}. We will distinguish between two different principles: \emph{spatial relationalism} and \emph{temporal relationalism}. These principles originate from one simple idea that I believe to be the core of Mach's principles:
\begin{quote}
    The dynamics of observable quantities should depend only on other observable quantities and no other external structures.
\end{quote}
From this general observation, we can identify two distinct physical principles that realize this idea:
\begin{principle}
    According to Mach, only the spatial relations between bodies matter.
    \begin{quote}
    ``When we say that a body $K$ alters its direction and velocity solely through the influence of another body $K^\prime$, we have inserted a conception that is impossible to come at unless other bodies $A$, $B$, $C$... are present with reference to which the motion of the body $K$ has been estimated.'' \cite{mach:mechanics}
    \end{quote}
    There is no absolute space -- only the spatial relations between these bodies. We will take this to be the principle of {\bf spatial relationalism}.
\end{principle}
\begin{principle}
    For Mach, the flow of time is perceived only through the \emph{changes} of spatial relations. 
    \begin{quote}
    ``It is utterly beyond our power to measure the changes of things by time. Quite the contrary, time is an abstraction, at which we arrive by means of the changes of things... '' \cite{mach:mechanics}
    \end{quote}
    We will refer to the statement that the flow of time should be a measure of change as the principle of {\bf temporal relationalism}.
\end{principle}

The above terminology has been adapted from \cite{Mittelstaedt:machs_2nd} since it clearly distinguishes two very different concepts. One is an ontological statement about what should be observed and how a physical theory should depend on these observables while the other is a definition of time valid for classical systems. Both principles derive from the simple statement that the dynamics of physical quantities should not depend on external structures. As we will see, these concepts of different physical origins manifest themselves differently in the technical description of relational theories such as GR. This distinction will be, thus, important to keep in mind as we develop best matching.

\section{A best--matched toy model}

In this section, we will illustrate the key features of best matching through a simple example. This will motivate the formal constructions of best matching that rigorously implement the principles stated above.

\subsection{Kinematics}

Consider a system of $3$ particles in 2 dimensions. The most general configuration possible in such a system is an arbitrary triangle. Following Mach's first principle, only the spatial relations between these particles are observable. Thus, there are two independent observables in this system and they can be parametrized many different ways. A way to see that there are only two physically meaningful observables is to note that there are only three lengths that can be measured in this system: the length of each side of the triangle. However, since lengths should not be compared to an absolute scale one must use one length as a reference length against which we measure the other two. This reduces the total degrees of freedom to 2. Another way to parametrized the physical degrees of freedom would be to choose the two largest angles. Because all angles must add up to $\pi$, these two angles are sufficient to completely determine the shape of the triangle up to an unobservables scale.

Historically, no satisfactory attempt to construct a dynamical principle in terms of the physical degrees of freedom of a system of particles has been successful (see \cite{barbour:newton_2_mach} for a summary of known attempts, which suffer from anisotropic effective mass). This poses an interesting philosophical question regarding the ubiquity of gauge theories. It is not my intention to address such a philosophical question in this work. Instead, I will simply point out that the only known dynamical principles that lead to sensible particle theories are \emph{not} written in terms of the physically observable quantities but, rather, \emph{redundant} variables. Best matching is a theory of this kind. As an immediate consequence of this redundancy, the implementation of Mach's first principle will require a quotienting of the redundant configuration space in order to indirectly isolate the physical, or \emph{relational}, degrees of freedom. In accordance with standard terminology, we will refer to the quotiented degrees of freedom as the \emph{gauge} degrees of freedom.

We now return to the main question that started this chapter: How can we determine the \emph{difference} between two shapes or, in this case, triangles? As noted above, our strategy will be to make use of Newton's absolute Euclidean coordinates only to quotient these by the unphysical gauge degrees of freedom. First, we will do a counting of (configuration) degrees of freedom to give us an idea of what our gauge group is. There are $3\times 2 = 6$ configuration degrees of freedom describing the motion of 3 particles in 2 dimensions. However, the origin and orientation of the coordinate system used to label the particle positions is completely unphysical. This corresponds to two translational degrees of freedom and one rotation. Also, the size of the triangle is unobservable. This adds up to 4 gauge degrees of freedom leaving 2 physical ones, in agreement with our previous analysis. The gauge group is, thus, the 2d Euclidean group, consisting of rotations and translations, crossed with the group of dilations. The tensor product of these two groups is also a group called the \emph{similarity} group.

\subsection{Matching triangles}

Best matching provides a dynamical procedure for quotienting the Euclidean positions of particles by the similarity group. Consider two snapshots of the 3 particle system, represented by two different triangles in a 2d Euclidean plane. An example of two such triangles is shown in Figure (\ref{fig:two triangles}).
\begin{figure}
    \begin{center} \includegraphics[width=.75\linewidth]{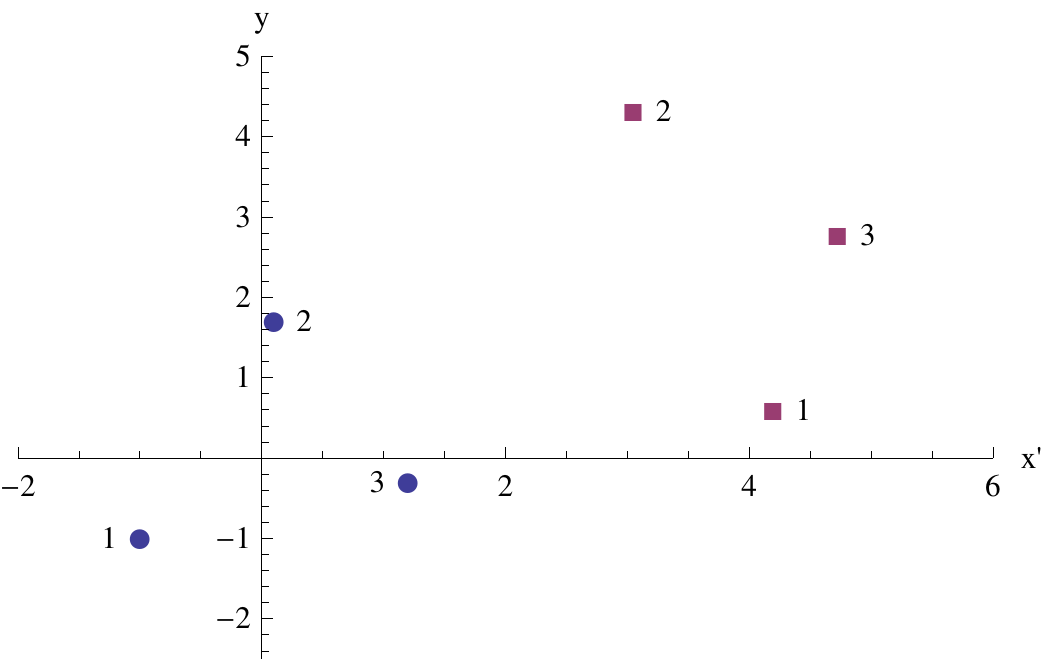} \end{center}
  \caption{Two snapshots of a three particle system at different times.}
  \label{fig:two triangles}
\end{figure}
We need to be able to compare these two triangles without making reference to their origin, orientation, and size. This is achieved by using one triangle as the reference shape then shifting the second triangle with arbitrary translations, rotations, and dilatations. Since we assume that each particle has an identity, we can calculate the ``distance'' between the two shapes by summing the Euclidean distances between each vertex of the triangle from one snapshot to the next. The \emph{best--matched configuration} is the one achieved by minimizing this ``distance'' using only translations, rotations, and dilatations of the second triangle. Figure (\ref{fig:bm triangles}) shows the second triangle as it is shifted into its best--matched position.
\begin{figure}
    \begin{center}\includegraphics[width=0.75\linewidth]{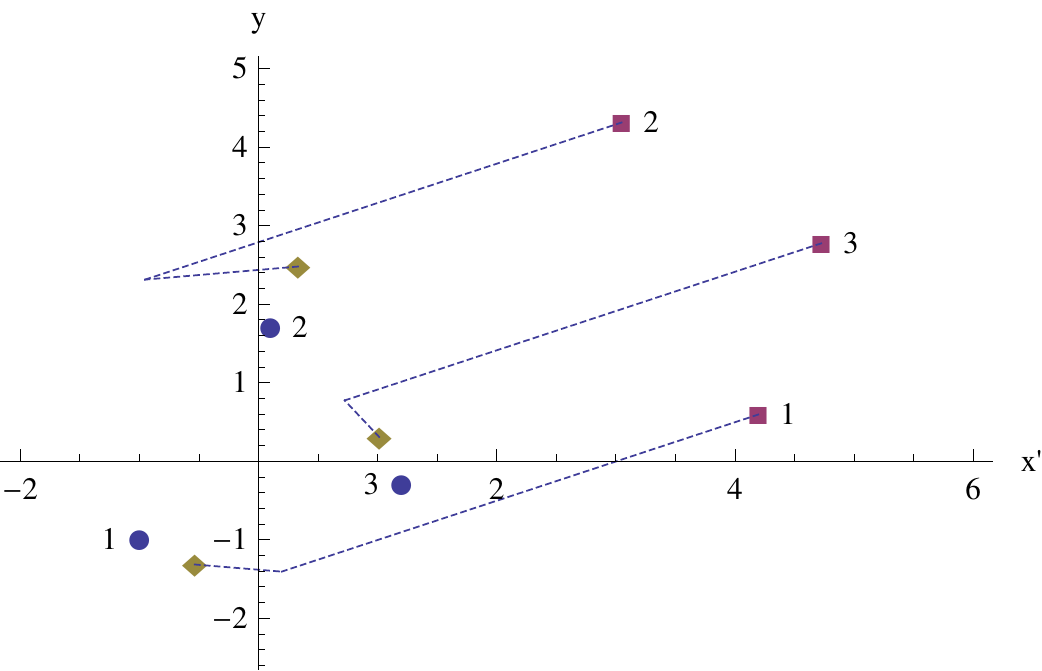} \end{center}
  \caption{The second triangle is to be shifted into its \emph{best--matched} position.}
  \label{fig:bm triangles}
\end{figure}

It is straightforward to express this procedure mathematically. Let $q^i_I(t)$ represent the $i^{\text{th}}$ Euclidean coordinate of the $I^{\text{th}}$ particle of the system at some time $t$. Since we will shortly be considering models where duration will emerge out the framework, it will be convenient to think of $t$ simply as some arbitrary parameter labeling the snapshots. To highlight this, we will call $t \to \lambda$ and think of $\lambda$ as a arbitrary time label which does nothing but order events. We can represent the ``shifting'' as a group action on the $q$'s
\begin{equation}
    q^i_I(\lambda) \to G(\phi^\alpha(\lambda))^i_j q^j_I(\lambda),
\end{equation}
where
\begin{equation}
    G(\phi^\alpha(\lambda))^i_j = \exp{\lf\{ \phi^\alpha(\lambda) t^i_{j,\alpha}\rt\}}
\end{equation}
and $\alpha$ ranges from 1 to 4 (the dimension of the similarity group in 2d). The $\phi^\alpha(\lambda)$ are the group parameters representing the amounts of rotation, translation, and dilatation to be performed and the $t^i_{j,\alpha}$ are the generators of the similarity group listed in Table (\ref{tbl:sim group gen}).
\begin{table}
    \centering
\begin{tabular}{| c | c | c |}
    \hline
    Symmetry     &     number of generators          & $t^i_{j,\alpha}$ \\
    \hline
    translations &    2 $(\alpha = k = 1\hdots 2)$   & $\delta^i_j \partial_k $ \\
    \hline
    rotations    &    1 $(\alpha = 3)$               & $\epsilon^{ik} q_j \partial_k$  \\
    \hline
    dilatations  &    1 $(\alpha = 4)$               & $\delta^i_j q^m \partial_m$ \\
    \hline
\end{tabular}
    \caption{The generators of the similarity group.} \label{tbl:sim group gen}
\end{table}

We imagine that the two snapshots represent the configuration of the system at two infinitesimally separated moments in time. If we define the quantity
\begin{equation}\label{eq:dqi}
    \delta q_I = G(\lambda + \delta\lambda) q_I(\lambda+\delta\lambda) - G(\lambda) q_I(\lambda),
\end{equation}
where spatial indices have been suppressed so that $G$ should be thought of a matrix and $\delta q$ as a column vector, then the condition that the triangles are best matched reduces to
\begin{equation}\label{eq:bm min}
    \text{min}_\phi \lf\{\delta^{IJ} \delta q_I^\intercal\, \eta\, \delta q_J \rt\}
\end{equation}
where $\eta$ is the diagonal unit matrix. The $^\intercal$ takes the transpose and the subscript $\phi$ indicates that a value of $\phi(\lambda)$ must be found that minimizes this quantity at all times $\lambda$. This procedure is reminiscent of a $\chi^2$ minimization of the distance between the vertices of the triangle.

\subsection{Action principle and the best--matching variation}\label{sec:bm variation}

The best--matching procedure is naturally expressed in terms of an action principle on configuration space. For infinitesimal $\delta \lambda$, we can expand \eq{dqi} and keep only the lowest order terms in $\delta\lambda$. We can then rewrite $\delta q_I$ and define the $\covd$ operator and its action on $q_I$ as
\begin{equation}
    \covd q_I \equiv G^{-1} \frac{\delta q_I}{\delta\lambda} = \dot q_I + G^{-1}\dot G q_I,
\end{equation}
where the $\dot{ }$ represents a derivative with respect to $\lambda$. Then, the minimization procedure \eq{bm min} is equivalent to the condition $\frac{\delta S}{\delta \phi} = 0$, where
\begin{equation}
    S = \int d\lambda \sqrt{ \delta^{IJ} (G \covd q_I)^\intercal \eta G \covd q_J }.
\end{equation}
This is clear because the integrand is the square root of the quantity to be minimized in \eq{bm min} under variations of $\phi_\alpha$. The square root is minimized so that the entire procedure is invariant under the choice of $\lambda$. This can be seen by noting that the action $S$ is invariant under reparametrizations of $\lambda$ of the form $\lambda \to f(\lambda)$, where $f$ is an arbitrary smooth function.

It is important to acknowledge that the variation with respect to $\phi_\alpha$ must be performed according to rules that implement the best matching procedure described above. These rules are \emph{not} the ones usually used in action principles because $\phi_\alpha$ is not a physically meaningful variable. Thus, its value at the endpoints of \emph{any} infinitesimal interval along the variation must remain arbitrary. This means that we cannot use the vanishing of $\delta \phi_\alpha$ on any interval of the variation. To see why this must be the case, recall the basic rules of the best--matching procedure. We have two triangles that we want to compare. To do this, we must be able to shift arbitrarily the triangles until they reach the best--matched position. But, this means that we certainly cannot fix the value of $\phi_\alpha$ at one of the endpoints. This would precisely defeat the purpose of the procedure because it would fix a particular origin, orientation, and scale for the system. Instead, we must be able to vary $\phi_\alpha$ \emph{freely} along any interval of the variation.

The mathematical realization of this variation can be stated in the following way. After an integration by parts, the variation of $S$ with respect to $\phi_\alpha$ takes the form
\begin{equation}
    \delta S_\phi = \int d\lambda \lf[ \diby{L}{\phi} - \frac d{d\lambda} \lf( \diby{L}{\dot\phi} \rt) \rt]\delta\phi + \lf. \delta\phi \diby{L}{\dot\phi} \rt|_{\lambda_i}^{\lambda_f},
\end{equation}
where $L$ is the Lagrange density
\begin{equation}
    L = \sqrt{ \delta^{IJ} (G \covd q_I)^\intercal \eta G \covd q_J }
\end{equation}
and $(\lambda_i,\lambda_f)$ are the endpoint values of $\lambda$. The local terms of $\delta S_\phi$ lead to the usual Euler--Lagrange equations for $\phi$. However, to get the boundary term to vanish, we must impose the additional condition
\begin{equation}\label{eq:fep cond}
    \lf. \diby{L}{\dot\phi}\rt|_{\lambda_i}^{\lambda_f} = 0.
\end{equation}
This condition must hold along \emph{any} infinitesimal interval one could chose to do the variation. This is because the best--matching procedure should be independent of which interval one chooses to perform the variation. If we impose the condition \eq{fep cond} for all values of $(\lambda_i,\lambda_f)$, we obtain the \emph{best--matching condition}
\begin{equation}\label{eq:bm cond}
    \diby{L(\lambda)}{\dot\phi} = 0.
\end{equation}
Thus, the best--matching variation of $\phi_\alpha$ is equivalent to a standard variation of $\phi_\alpha$ with the additional condition \eq{bm cond}.

To obtain an interacting theory, it is necessary to slightly generalize the action $S$. To motivate this generalization, note that $\delta^{IJ} \eta_{ij}$ is a flat metric on configuration space. The variational principle for $S$ is then a geodesic principle on configuration space since $S$ is just the length of path on configuration space using this metric.\footnote{In fact, this is not quite true. Only if one treats the $\covd$'s as standard $\lambda$ derivatives would this be true. We will address this difference in Section~(\ref{sec:formal constructions}).} To get a non--trivial theory, we must simply curve the metric on configuration space. The simplest way to do that is to multiply it by a conformal factor. This simple generalization is sufficient to reproduce Newtonian particle mechanics, as we will see. If we call the conformal factor $2(E-V(Gq))$ (the factor of 2 is conventional and it is most natural to think of $V$ as a function of the best matched coordinates, $Gq$) then the Lagrange density becomes
\begin{equation}
    L \to \sqrt{ 2 (E-V(Gq)) \delta^{IJ} (G \covd q_I)^\intercal \eta G \covd q_J }.
\end{equation}
The quantity $\delta^{IJ} (G \covd q_I)^\intercal \eta G \covd q_J$ is just twice the kinetic energy $T$ of the system once it has been best matched (in units where the particle masses have been set to 1) and the total energy, $E$, of the system is a constant determined experimentally. Using this definition, we have
\begin{equation}
    L = 2 \sqrt{ (E-V) T }.
\end{equation}
This action is commonly known as Jacobi's action and is known to reproduce Newtonian particle dynamics when $V$ is interpreted as the usual potential for the system. For an introductory treatment of Jacobi's theory see chapter V.6-7 of Lanczos's book \cite{lanczos:mechanics}.

Best matching as applied to 3 particles in 2d can be stated as follows
\begin{itemize}
    \item First, start with a geodesic principle on configuration space (i.e. Jacobi's action). This sets up the $\chi^2$ type minimization required for best matching.
    \item Make the substitutions $q \to G(\phi) q$. This allows for the appropriate shifting of the coordinates.
    \item Perform a best--matching variation of $\phi$ by imposing the Euler--Lagrange equation and the addition best-matching condition \eq{bm cond}.
\end{itemize}
It is important to point out that this procedure, particularly the best--matching condition, was derived from the simple requirement of finding the ``difference'' between shapes by minimizing the incongruence between them. Later, we will see that the best--matching condition is key to the discovery of shape dynamics. The point to emphasize here is that this condition is not ad--hoc in any way but results from a simple idea motivated by Mach's principles.

\subsection{Linear constraints and Newton's bucket}

The best--matching condition, \eq{bm cond}, can be computed for our system. The result leads to valuable physical insight into the meaning of best matching and how best matching resolves Newton's bucket problem.

Taking partial derivatives and dropping overall factors, it is a short calculation to show that
\begin{equation}\label{eq:toy lin const}
    \diby{L(\lambda)}{\dot\phi} = 0 \Rightarrow \delta^{IJ} (\covd q_I)^\intercal \eta t_\alpha q_J = 0.
\end{equation}
Inserting the values of the generators $t_\alpha$ of the similarity group from Table~(\ref{tbl:sim group gen}), we find that these constraints reduce to
\begin{align}
    \sum_I p^i_I &=0 \label{eq:lin const}\\
    \sum_I \epsilon_{ij} p_I^i q_I^j &= 0 \label{eq:rot const}\\
    \sum_I p^i_I q^i_I &= 0\label{eq:dil const},
\end{align}
where $p^i_I = \dot q^i_I + \dot \phi^i$ is the best--matched linear momentum of the $I^\text{th}$ particle and $\epsilon_{ij}$ is the completely anti--symmetric tensor in 2d.

To understand this result, consider what the linear constraint \eq{lin const}, generated from best matching the translations, accomplishes. The best--matched momentum $p^i_I$ represents the momentum of the vertices of the triangle when the system has been shifted to the best--matched position. Thus, the condition \eq{lin const} says that the total linear momentum of system, when best--matched, is zero. This is precisely the Noether charge associated to translational invariance. In other words, the best--matching procedure requires that the system be shifted translationally such that the total momentum of the system is zero. Unsurprisingly, the analogous thing holds for the rotations and dilatations. The constraint \eq{rot const} says that the total angular momentum (in 2d) of the system is zero when the system has been best--matched. Similarly, \eq{dil const} requires the \emph{dilatational} momentum to vanish. Indeed, we will see that it is a general result: the best matching condition is a constraint, linear in the momentum, that requires the vanishing of the appropriate Noether charge. The meaning of this is clear. In standard mechanics, the value of the Noether charge is set by the initial conditions. In best matching, the initial conditions that set the value of this charge have no physical significance since they correspond to the gauge coordinates of the $q$'s. As a result, the actual value of this charge is physically meaningless. The best--matching condition is a choice where the value of this unphysical charge is set to zero.

We can now return to the example of Newton's bucket and see how best matching provides a concrete model for framing Mach's argument. We can think of our system as a being composed of the water, bucket, and fixed stars. Best matching with respect to the rotations tells us that the total angular momentum of the system must be zero. Since the fixed stars are very massive and distant, their contribution to the angular momentum is significantly greater than that of the bucket or the water. Because of this, the bucket and water can have virtually any realistic amount of rotation without impacting the angular momentum of the whole system and, thus, the dynamics of the rest of the system. This leads to the \emph{illusion} of absolute space because the fixed stars effectively behave like a fixed absolute background for rotation. However, if they were not part of the system and only the water and bucket existed in the universe, then the best--matching condition would imply that the angular momentum of the bucket \emph{must} cancel that of the water. This would lead to only two possibilities: either the bucket and water are not rotating at all and the surface of the water is flat or the bucket and water are rotating in opposite directions and the surface of the water is curved up the walls of the bucket. Unlike the results of Newton's experiment, these two possibilities are completely consistent with a relational theory since there are only two outcomes correlated directly with the relative motion of the bucket and water. Unfortunately, such an experiment is not feasible because the real universe consists of much more than some water in a bucket. Nevertheless, we see that best matching provides an explanation for how the fixed stars actually create the illusion of a fixed background for small subsystems of the universe such as Newton's bucket.

\subsection{Mach's second Principle}

Surprisingly, the best matching procedure we have just developed for implementing Mach's first principle also implements Mach's second principle. As we pointed out, taking the square root  of the minimum distance between shapes, \eq{bm min}, makes $S$ invariant under reparametrizations of the time label $\lambda$. In fact, there is a preferred choice of parametrization where the equations of motion manifestly take the form of Newton's equations for non--relativistic particles in terms of the best--matched coordinates. Performing a variation of $S$ with respect to the shifted quantities $\bar q = G q$, we obtain
\begin{equation}
    \sqrt{\frac V T}  \frac{d}{d\lambda}\lf(\sqrt{\frac V T} \frac{ d\bar q}{d\lambda} \rt) = -\diby{V}{\bar q},
\end{equation}
where, as before, $V(\bar q)$ is the potential and $T(\dot{\bar q})$ is the kinetic energy but both are in terms $\bar q$. If we identify
\begin{equation}\label{eq:eph time1}
    \dot t_N = \sqrt{\frac T {E-V}}
\end{equation}
This becomes
\begin{equation}
    \frac{ d^2 \bar q}{dt_N^2} = -\diby{V}{\bar q}
\end{equation}
which is precisely Newton's law in terms of the best--matched coordinates.

The choice \eq{eph time1} is a particular parametrization that is completely equivalent to Newton's time. This choice, however, is relational since its definition depends purely on the configuration space variables and their relative changes. To be more precise, we can rewrite \eq{eph time1} in the following way by taking out the $\lambda$ dependence (since $t_N$ is reparametrization invariant):
\begin{equation}
    \delta t_N = \sqrt{ \frac 1 2 \frac{ \delta^{IJ} (\delta q_I)^\intercal \eta \delta q_J}{E - V(\bar q)}}.
\end{equation}
This is proportional to the total \emph{change} of the shape of the system. Thus, the best--matching procedure leads directly to a notion of time that is both equivalent to Newton's and that treats time as a measure of the change in the configurations of the universe. This is precisely in accordance with our statement of Mach's second principle.

\section{Formal constructions}\label{sec:formal constructions}

The toy model can now be used to point out the key features of best matching and motivate the general geometric features of the procedure.

\subsection{Mach's first principle and Principal Fiber Bundles}

The first step in best matching is also the most subjective. It involves identifying the ontology of the theory. The basic assumption behind best matching is that the most convenient variables presented to us to study physical theories contain significant redundancies. These redundancies can be eliminated by best matching if they originate from a continuous symmetry generated by a Lie algebra. In the case of our toy model, the most convenient variables to use to study the theory are the Euclidean coordinates. However, the physically meaningful quantities, measurable by observers in the system, are the ratios of the particle separations. This suggests that the redundancy of the variables is parametrized by the similarity group. In general, the first step is to identify the redundant, or \emph{absolute}, configuration space $\acs$ then the symmetry group $\cG$ parametrizing the redundancy. These identifications imply the quotient space $\rcs = \acs / \cG$ representing the reduced, or \emph{relational}, configuration space on which live the physical degrees of freedom.

The redundant configuration space $\acs$ admits a Principal Fiber Bundle (PFB) structure. In the finite dimensional particle models, $\acs$ has a simple tensor product structure $\acs = \rcs \times \cG$ so that the fibers are given by $\cG$ and the base space is $\rcs$. Vertical flow along the fibers is generated by the group action on $\acs$. In our toy model, this is represented by translations, rotations, and dilatations of the triangles. Genuine changes of shape are represented by horizontal flow in the PFB $\acs$. Figure~(\ref{fig:bm shapes}) shows how changes of orientation and scale of the triangles represents motion along the fibers of $\acs$.
\begin{figure}
    \begin{center}
	\includegraphics[width = 0.7\textwidth]{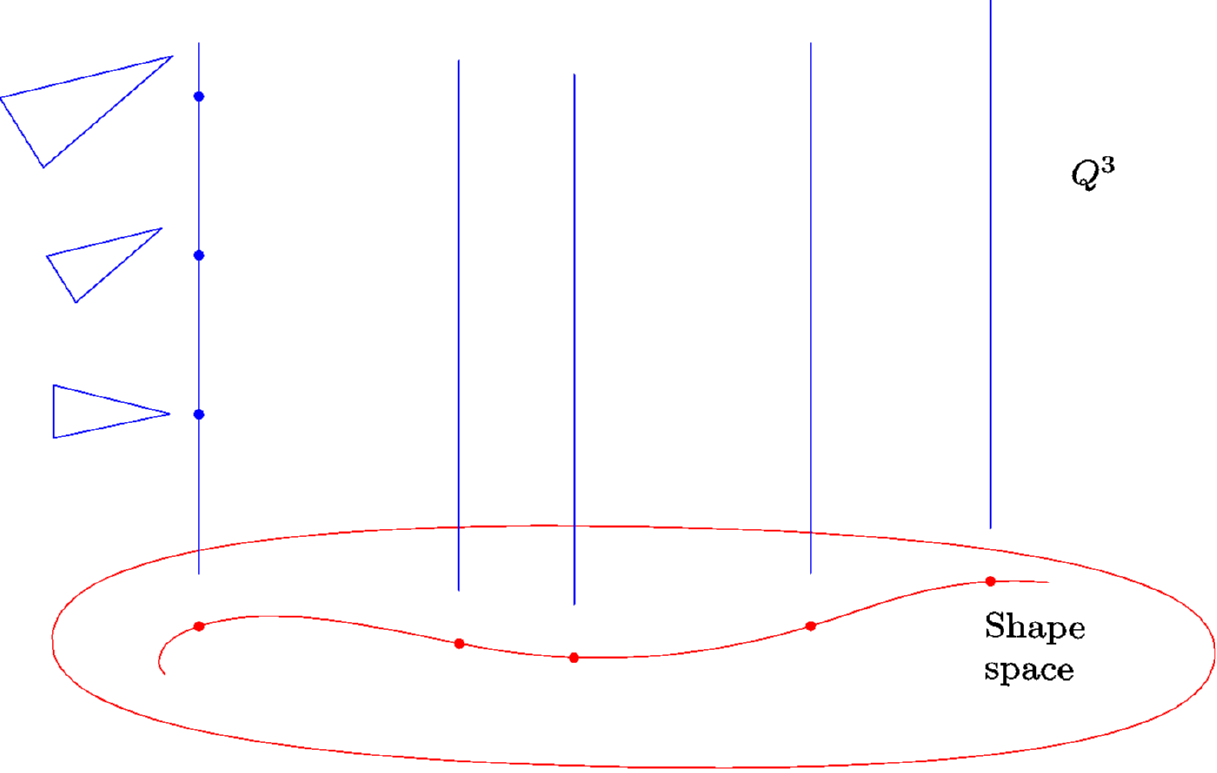}
	\caption{Changes of orientation and scale represent motion along the vertical directions of a fiber bundle over shape space. (Note: $\acs = Q^3$)} \label{fig:bm shapes}
    \end{center}
\end{figure}
Best matching is a procedure that aims to find the ``difference'' between two infinitesimally different shapes. Geometrically, this requires a notion of derivative between two neighboring fibers. In other words, it requires a connection on $\acs$. Indeed, best matching is precisely a method for choosing a connection on $\acs$. Just as a choice of section on a PFB selects one member out of an equivalence class represented by the fibers, best matching selects one member -- the best--matched configuration -- out of an equivalence class generated by the symmetry group $\cG$. In fact, the best--matching condition
\begin{equation}
    \diby{L(\lambda)}{\dot\phi} = 0
\end{equation}
corresponds to the condition for the \emph{horizontal lift} above a particular curve in $\rcs$. Figure~(\ref{fig:fiber bundle}) illustrates how the fiber bundle structure of best matching projects curves onto shape space (many thanks go to Boris and Julian Barbour for providing Figures~(\ref{fig:bm shapes}) and (\ref{fig:fiber bundle})).
\begin{figure}
     \begin{center}
	\includegraphics[width = 0.8\textwidth]{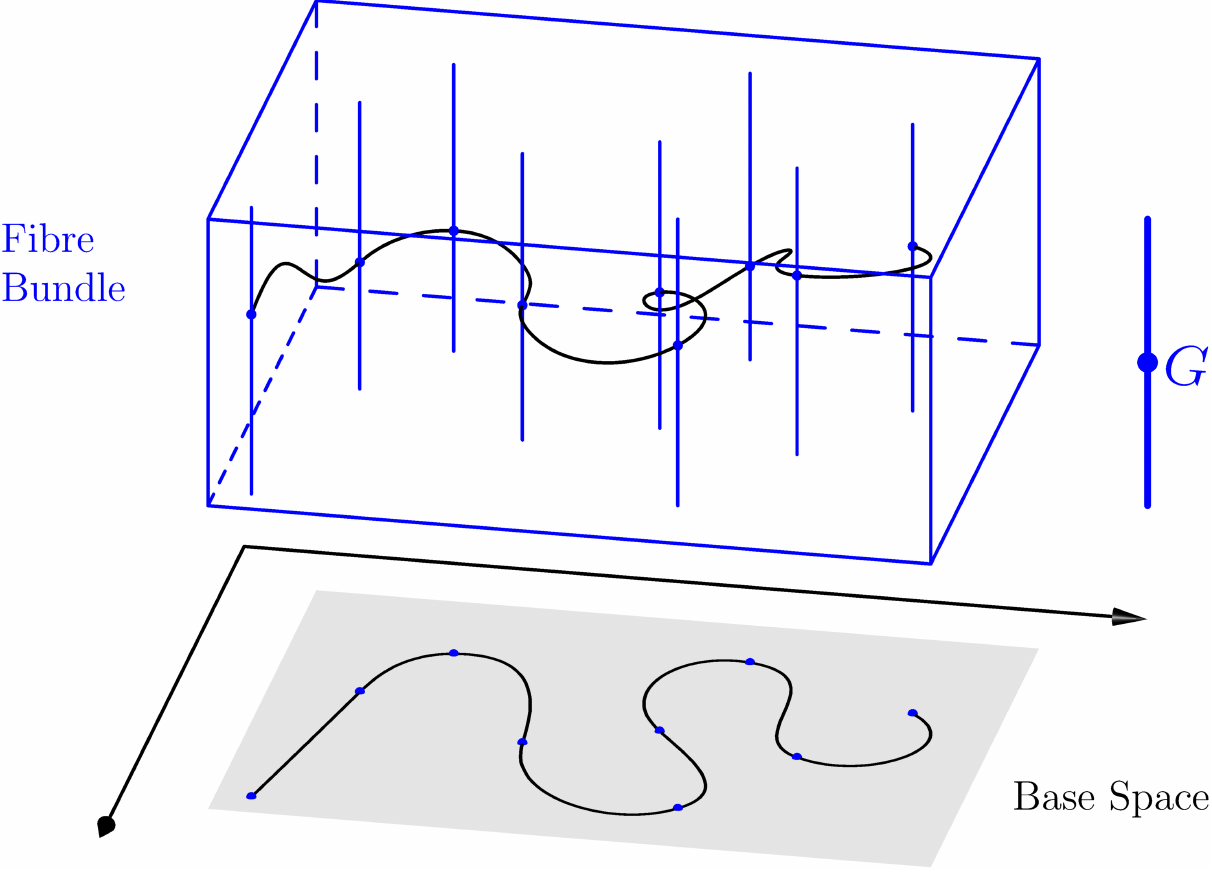}
	\caption{The fiber bundle structure of best matching projects curves down to shape space.} \label{fig:fiber bundle}
    \end{center}   
\end{figure}

In practice, it is usually necessary to extend configuration space by introducing the auxiliary fields $\phi_\alpha$. They can be thought of as the vertical components of the redundant configuration variables. It is now clear that the operator $\covd$, whose action on the $q$'s in our toy model was
\begin{equation}
    \covd q = \dot q + \dot \phi^\alpha t_\alpha q,
\end{equation}
is actually defining the covariant derivative along a trial curve in $\acs$. $\dot\phi_\alpha$ is then the pullback of the connection onto this trial curve. From our knowledge of $\phi_\alpha$, it is possible to compute the full best--matching connection $A$ over $\acs$. This is done in Section~(\ref{sec:bm connection}).

To summarize, Mach's first principle is implemented in best matching by: first, identifying the configuration variables, $\acs$, and their symmetries, $\cG$, then, by using the horizontal lift, or best--matching connection, to compare neighboring configurations in $\acs$. This ensures that only physically meaningful quantities enter the the description of the dynamics because the PFB structure projects the dynamics onto the relational configuration space $\rcs$. This has the effect of making irrelevant any initial conditions specified along the fiber directions since this information is physically meaningless.

\subsection{Mach's second principle and geodesics on configuration space}

In the toy model, $2V(q) \delta^{IJ}\eta_{ij}$ was a conformally flat metric on configuration space. The dynamics implied by best matching produced trajectories that are geodesics of this metric. The ability to define metrics on configuration space suggests a natural way to implement Mach's second principle. According to our definition, time, or more specifically duration, should be a measure of the total change undergone by the configurations. The principal fiber bundle structure obtained from implementing Mach's first principle allows us to project the dynamics onto the reduced configuration space $\rcs$. Then, to satisfy Mach's second principle, duration should be given by a length on $\rcs$. Indeed the Newtonian time, $t_N$, is precisely that. It should be cautioned that the metric used to compute the Newtonian time is \emph{not} the same metric used in the action. Nevertheless, the basic idea is clear: the presence of natural metrics on configuration space allows for both a way to define the dynamics through a geodesic principle and a way to define duration in a Machian way. Thus, the final picture that emerges from best matching is a geodesic principle on $\rcs$.

The fact that we have a geodesic principle on $\rcs$ puts an additional restriction on the number of freely specifiable initial data in the theory. To specify a geodesic, one requires a point and a \emph{direction}. This is one less piece of information than is typically required to specify dynamics on configuration space since, in the standard case, one must specify a point and a \emph{tangent vector}. The one additional piece of information, the length of the tangent vector, specifies the speed at which the system moves through the trajectory. Since geodesics are reparametrization invariant, the speed along the trajectory is physically meaningless. One can then summarize Mach's principles by stating them through the number of freely specifiable initial data required to uniquely specify a classical solution:
\begin{quote}
    The freely specifiable initial data required to uniquely specify a classical solution of a relational theory is a \emph{point} and \emph{direction} on the reduced configuration space $\rcs$.
\end{quote}
This requirement was first stated by Poincar\'e \cite{poincare:principle} as a generalized relativity principle as been coined \emph{Poincar\'e's principle} \cite{barbourbertotti:mach}.

\subsection{The best--matching connection}\label{sec:bm connection}

It is possible to compute the best--matching connection, $A$, over the whole configuration space. Since the $\phi$ fields represent $A$ pulled back onto a trial curve $\gamma$, we can find an expression for $A$ by generalizing \eq{toy lin const} over the whole configuration space. The covariant derivative along a path $\covd$ generalizes to
\begin{equation}
    \covd q^a_I \to \mathcal D^J_b q^a_I = \diby{q^a_I}{q^b_J} + A^{\alpha,J}_b \ta{a}{c} q^c_I,
\end{equation}
where $A^{\alpha,J}_b$ are the components of the best--matching connection. This has both particle and spatial indices in accordance with the index structure we are using for the configuration space coordinates $q^a_I$. Using this, \eq{toy lin const} generalizes to
\equa{\label{eq:A_eom}
  \delta^{IJ} \mathcal D^K_a q^c_I \eta_{cd} \ta{d}{e} q^e_J = 0.
}
for A (we have reinserted spatial indices to avoid any ambiguity in notation).

It is instructive to compute $A$ for simple cases.

\subsubsection{Scale invariant model}

Here we consider the global connection associated with the dilatations. The advantage of this simple case is that an explicit expression can be obtained for $A$. This is enlightening because it tells us what the connection is doing over the entire PRB $\acs$.

For the dilations, there is only one generator
\equa{
  \ta a b \to \delta^a_b.
}
We will rescale the coordinates $q$ so that it is possible to consider particles with different masses, $m_I$, (previously, we used units where $m_I =1$)
\equa{
  q^a_I \to \sqrt{m_I} q^a_I.
}
In these units, it is a short calculation to work out \eq{A_eom} in terms of $A$. This gives:
\equa{
  A_a^I(q) = - \partial^I_a \ln \sqrt{I(q)}.
}
In the above, $I(q)$ is the off--shell generalization of the moment of inertia of a point in $\acs$
\equa{
  I(q) = m_I \delta^{IJ} q_I^\intercal \eta q_J.
}
In row vector notation, where rows label particle numbers, the connections take the simple form,
\equa{
  A^I_a(q) = -\frac{\eta_{ab}}{I}\lf(\begin{array}{c}
m_1 q^b_1\\
m_2 q^b_2\\
\cdots
                      \end{array}\rt).
}

Pulling back this result onto a trial curve $\gamma$, parametrized by $\lambda$, gives
\equa{\label{eq:si_omega}
  \dot\phi (\lambda) = - \frac{d}{d\lambda} \lf( \ln \sqrt{I(\lambda)} \rt),
}
which is identical to the result obtained by directly solving \eq{dil const}. Note that the potential, fortunately, drops out of the equations for $A$.

\subsubsection{Translationally Invariant Models}

As a second example, we will consider the pure translations. For this example, care must be taken to get the correct index structure for the generators. In this case, the $\alpha$ index splits into a spatial index \emph{and} a particle index. However, this notation is redundant since the generators are identical for each particle. In the end, to solve for $A$ we will need to sum over this redundant index. Explicitly, the generators take the form
\begin{equation}
    \ta a b \to \delta^a_b \diby{}{q^n_N},
\end{equation}
where the $\alpha$ index has split into $n$ and $N$. Similarly, $A$ has the components
\begin{equation}
    A^{\alpha, I}_a \to A^{n,I}_{N,a}.
\end{equation}

Using these conventions, one can reduce \eq{A_eom} to the following form
\equa{
  \sqrt{m_N m_A} \delta^{AN} \eta_{an} + \sqrt{m_N m_B} A^{b,A}_{a,B} \delta^{BN}\eta_{bn} = 0,
}
where we avoid summing over particle indices. Because of the redundancy in the notation, this expression contains many more equations then we need. Summing over $N$ allows us to eliminate the redundancy. Performing this sum and rearranging gives
\equa{
  A^{nA}_a = -\frac{m^A\delta^n_a}{\sum_N m_N}.
}
This is a very compact way of expressing the information contained in the best--matching connection for the translations. $A$ is both diagonal and depends only on the masses.

This result can be pulled back to a particular path $\gamma$ on $\acs$ parametrized by the parameter $\lambda$. This leads to
\equa{
  \dot\phi^n = - \frac{\sum_I m_I \dot q^n_I}{\sum_I m_I},
}
which is the center of mass velocity. Thus, the best--matching procedure is telling us explicitly that the velocity of the center of mass is the physically meaningless quantity associated with the translational invariance, completely in agreement with our intuition.

\subsection{Equivalent action}\label{sec:equiv S}

As already discussed, to define a geodesic principle on $\rcs$, it is sufficient to use the length of a path on $\rcs$ as the action. In general, a geodesic action takes the form
\begin{equation}
    S = \int d\lambda \sqrt { g_{ab} \dot q^a \dot q^b }.
\end{equation}
Since our action is actually defined on $\acs$, we must replace $\lambda$ derivatives by covariant derivatives and shift the metric $g_{ab}$ to its best matched position $ \bar g_{ab} = G_a^c g_{cd}(Gq) G^d_b$. Thus,
\begin{equation}
    \bar S = \int d\lambda \sqrt { \bar g_{ab} \covd q^a \covd q^b }
\end{equation}
leads to a proper geodesic principle on $\rcs$ when $\phi_\alpha$ are varied with a best--matching variation.

Because actions with square roots are difficult to deal with mathematically, it is convenient to introduce an auxiliary field $N(\lambda)$, which we will call the \emph{lapse} in analogy to the ADM action of GR, to write the action in a simpler form. Assuming that the metric decomposes conformally as
\begin{equation}
    g_{ab} = V g'_{ab}
\end{equation}
then it is easy to show that the action
\begin{equation}
    S' = \int d\lambda \lf[ \frac 1 N  \bar g'_{ab} \covd q^a \covd q^b - N V \rt]
\end{equation}
reproduces exactly the same equations of motion. We will often use this form of the action for expressing a geodesic principle since it is more manageable mathematically.

There is yet another way to simplify this action that gives it a structure similar to a best matched symmetry. The idea is to note that the lapse, $N$, has the units of velocity. Thus, it is better to think of $N$ as the derivative of some variable $\tau$. If set $N = \dot\tau$ then
\begin{equation}\label{eq:php general action}
    S' = \int d\lambda \lf[ \frac 1 {\dot\tau}  \bar g'_{ab} \covd q^a \covd q^b - \dot\tau V \rt].
\end{equation}
In this form, $S'$ is a reparametrization invariant version of Hamilton's well--known principle, where $L = T - V$. Thus, we will refer to the variation of this action as Parametrized Hamilton's Principle (PHP). However, because of the derivative on $\tau$, we must perform a best matching variation of $\tau$ to get the same equations of motion that we had with the Lagrange multiplier $N$. In this way, it appears that the theory defined by \eq{php general action} is a theory where the reparametrization invariance has been best matched.

\subsection{Equivariant and non--equivariant metrics}

We conclude this chapter by making an important distinction between two kinds of best--matching theories. The first occurs when the metric, $g_{ab}$, on configuration space is equivariant under the symmetry group. This means that the flow generated by $\cG$ is a Killing vector of $g$. An immediate result of this is that
\begin{equation}
    \bar g_{ab} = G_a^c g_{cd}(Gq) G^d_b = g_{ab}.
\end{equation}
Then, the best--matching procedure takes the original action
\begin{equation}
    S = \int d\lambda \sqrt { g_{ab} \dot q^a \dot q^b },
\end{equation}
which is invariant under \emph{global} (in $\lambda$) gauge transformations of the form $q \to Gq$, and sends it to the \emph{locally} gauge invariant action
\begin{equation}
    S' = \int d\lambda \sqrt { g_{ab} \covd q^a \covd q^b }
\end{equation}
by promoting $\lambda$ derivatives to gauge covariant derivatives. Equivariant best matching is, thus, equivalent to doing standard gauge theory on configuration space. However, it provides a powerful conceptual framework, based on Mach's principles, to motivate the local gauging of a symmetry. In this sense, best matching leads to a deeper understanding of gauge theory.

The second possibility is that the metric in \emph{not} equivariant under the action of $\cG$. In this case, we will see that the consistency of the equations of motion is no longer guaranteed. However, in special situations, consistency can be restored in a particular gauge so that a geodesic principle on $\rcs$ can still be defined. This allows for the possibility of constructing a truly equivariant metric by equivariantly lifting the metric on $\rcs$. This leads to a \emph{dual} theory that has the required symmetries. In GR, this procedure will lead to shape dynamics when 3d Weyl symmetry is best matched.

%------------- EQUIVARIANT BEST MATCHING ------------------------------
%----------------------------------------------------------------------
% EQUIVARIANT BEST MATCHING
%----------------------------------------------------------------------
% This section develops the idea of best matching in the case of an
% equivariant supermetric for finite dimensional systems. We give lots
% of pictures and try to be pedogogical.
%
%
%======================================================================
\chapter{Equivariant best matching}\label{chap:equiv bm}
%======================================================================

In this chapter, we will give an in depth description of the mathematical structure of best matching in general finite dimensional systems. Undoubtedly, one cannot fully understand the structure of a theory until one has understood its Hamiltonian formulation. As Dirac put it: ``I feel that there will always be something missing from [alternative methods] which we can only get by working from a Hamiltonian.'' \cite{dirac:lectures} With this in mind, we perform a full canonical analysis of a general class of finite dimensional best--matched theories. In this chapter, we will only consider the case where the metric is equivariant, saving the non--equivariant case for next chapter. This will provide a detailed framework for understanding relational theories and will complement the key results of last chapter.

We will proceed as follows. First, we will formulate some general geometric constructions for the framework. Then, we perform the Legendre transform and compute the canonical equations of motion. We develop in detail the canonical version of the best--matching variation and point out some key differences compared with the Lagrangian approach that allow us to capture the full gauge invariance of the theory. This identification allows us, at the same time, to identify and then eliminate the gauge redundancies, providing us with an explicit formalism to formally compute the gauge independent observables. These new insights provide two valuable tools: 1) the matching procedure can be seen as a canonical transformation on the extended phase space, which will be useful for understanding the non--equivariant theories, and 2) the best--matching condition provides a specific criterion for defining background independence with respect to a given symmetry. We end the chapter by applying our definition of background independence to reparametrization invariant theories. This has interesting implications for the problem of time.

\section{Finite Dimensional Models}\label{sec:finite dim models}

\subsection{Mach's second principle}\label{sec:jacobi_principle}

We begin with some general geometric considerations that will deepen our understanding from last chapter and set up the transition to the Hamiltonian theory.

Mach's second principle is implemented by a geodesic principle on the configuration space, $\acs$ (we will consider the $\rcs$ and Mach's first principle in a moment). This can be achieved by extremizing an action, $S$, that gives the length of the trajectory on configuration space
\equa{\label{eq:jacobi_gen}
	S = \int_{q_\text{in}}^{q_\text{fin}} d\lambda\, \sqrt{g_{ab}(q) \dot{q}^a(\lambda)\dot{q}^b(\lambda)}.
}
$g_{ab}$ is a function only of $q$ and not of its $\lambda$--derivatives. $\lambda$ has been written explicitly in this reparametrization invariant action so that it can be used as an independent variable in the canonical analysis.

We will find it convenient fix a conformal class of the metric by selecting a positive definite function $\Omega(q)$ such that
\equa{\label{eq:conf_decomp}
	g_{ab} = \Omega \gamma_{ab}.
}
In many situations, $-\Omega$ can be interpreted as twice the potential energy of the system. As we have seen, in the dynamics of non--relativistic particles, the configuration space is just the space of particle positions $q^i$. The metric $g_{ab}$ leading to Newton's theory is conformally flat so that
\equa{ 
	\gamma_{ab} = \eta_{ab},
}
where $\eta$ is the flat metric with Euclidean signature.\footnote{The units can be chosen so that all of elements of $\eta$ are 1. Particles with different masses can be considered by replacing $\eta$ with the suitable mass matrix for the system.} In general, the metric $g_{ab}$ is a specified (ie, \emph{non}--dynamical) function on $\acs$. From now on, we will use the action (\ref{eq:jacobi_gen}), making use of the decomposition (\ref{eq:conf_decomp}) only when necessary. This allows us to work directly with geometric quantities on $\acs$.

The variation of $S$ leads to the geodesic equation
\equa{
    \ddot q^a + \Gamma^a_{bc} \dot q^b \dot q^c = \kappa(\lambda) \dot q^a,
}
where $\kappa \equiv d \ln \sqrt{g_{ab} \dot{q}^a\dot{q}^b}/ d\lambda$ and $\Gamma^a_{bc} = \frac{1}{2} g^{ad} (g_{db,c} + g_{dc,b} - g_{bc,d})$ is the Levi-Civita connection on $\acs$.

The choice of the parameter $\lambda$ is important. Normally, one would like to set $\kappa = 0$ with an affine parameter. However, for metrics of the form \eq{conf_decomp} with $\gamma = \text{const}$, there is another special choice of $\lambda$ that simplifies the geodesic equation. If we choose the parameter $\tau$ such that
\equa{
    \frac{d\tau}{d\lambda} = \frac{\sqrt{g_{ab} \dot{q}^a\dot{q}^b}}{\Omega},
}
the geodesic equation becomes
\equa{
    \gamma_{ab} \frac{d^2 q^b}{d\tau^2} = \frac{1}{2}\partial_a \Omega.
}
In the case of non--relativistic particles, $\Omega = -2V$ and $\gamma = \eta$ so that the geodesic equation is Newton's $2^{\text{nd}}$ law. With these choices, $\tau = \sqrt{-\frac{T}{V}}$ is just equivalent to the Newtonian time $t_N$.

\subsection{Mach's first principle}

We can now implement Mach's first principle by introducing the auxiliary fields $\phi_\alpha$ and performing best matching. This sends the $\lambda$ derivatives to covariant derivatives and shifts the metric $g_{ab}$ to its best--matched value $\bar g_{ab} = G_a^c g_{cd}(Gq) G^d_b$. This leads to the action
\equa{\label{eq:jacobi_bm}
    S = \int_\gamma d\lambda\, \sqrt{\bar g_{ab}\, \covd q^c \covd q^d},
}
which must be varied over all paths $\gamma$ on $\acs$ according to the rules of the best--matching variation.

The action \eq{jacobi_bm} can be written in an illuminating form using the fact that our metric is symmetric under $\mathcal G$. The equivariance criterion leads to the existence of global Killing vectors and is expressed by the fact that the Lie derivative in the direction of the symmetry generators $\mathcal L_{t_\alpha q} g = 0$ is zero. Explicitly,
\equa{\label{eq:killing_loc}
    \ta{c}{\lf( a\rt.} g_{\lf. b\rt) c} + \partial_c g_{ab} \ta{c}{d} q^d = 0.
}
where the rounded brackets indicate symmetrization of the indices. This expression can be exponentiated to prove the following relation
\equa{\label{eq:killing_field}
    \bar g_{ab} = g_{ab}(Gq) \, G^a_c G^b_d = g_{cd}(q).
}
Inserting this into \eq{jacobi_bm} gives
\equa{\label{eq:bm_jacobi_S}
    S = \int_{q_\text{in}}^{q_\text{fin}} d\lambda\, \sqrt{g_{ab}(q)\, \covd q^a \covd q^b}.
}
The relation \eq{killing_field} can be inserted into the original action $S$ to show that, with an equivariant metric, the original theory is invariant under $\lambda$ independent $\cG$--transformations. This shows in detail why equivariant best matching is equivalent to gauging a global symmetry.

We can now proceed with the canonical analysis of the best--matching action (\ref{eq:bm_jacobi_S}). The momenta $p_a$, conjugate to $q^a$, and $\pia$, conjugate to $\phi^\alpha$, are
\begin{align}
    p_a &\equiv \diby{L}{\dot{q}^a} = \frac{g_{ab}\, \covd q^b}{\sqrt{g_{cd}\, \covd q^c \covd q^d }}, \label{eq:pa_jac} \qand \\
    \pia &\equiv \diby{L}{\dw} = \frac{g_{ab}\, \covd q^b}{\sqrt{g_{cd}\, \covd q^c \covd q^d }} \ta{a}{e} q^e. \label{eq:pia_jac}
\end{align}
It is easy to verify that these momenta obey the following primary constraints
\begin{align}\label{eq:jacobi_constraint}
    \ham &= g^{ab}\, p_a p_b -1 = 0, \qand \\
    \lin &= \pia - p_a \ta{a}{b} q^b = 0, \label{eq:lin}
\end{align}
where $g^{ab}$ is the inverse of $g_{ab}$. The quadratic \emph{scalar constraint} $\ham$ arises from the fact that $p_a$, according to (\ref{eq:pa_jac}), is a unit vector on phase space. It corresponds to the fact that only a direction in $\acs$ can be specified. Thus, $\ham$ reflects the physical insignificance of the length of $\dot{q}$. The linear \emph{vector constraints} $\lin$ reflect the continuous symmetries of the configurations. They indicate that the phase space of the theory contains equivalence classes of states generated by $\lin$. Note that $\ham$ and $\lin$ arise in very different ways: $\ham$ is a direct consequence of using a geodesic principle on $\acs$ and can be attributed to Mach's second principle while $\lin$ is a consequence of the best matching and can be attributed to Mach's first principle. In the equivariant case, we will see that these constraints are first class. This leads precisely to the restrictions on the freely specifiable initial data required by Poincar\'e's principle. 

Using the fundamental Poisson Brackets (PBs)
\equa{
    \pb{q^a}{p_b} = \delta^a_b, \qand \quad \pb{\wa}{\pi_\beta} = \delta^\alpha_\beta,
}
we find that there are two sets of non--trivial PBs between the constraints. They are
\begin{align}
    \pb{\lin}{\ham_\beta} &= c_{\alpha\beta}^\gamma \ham_\gamma \label{eq:lie_algebra} \text{, and}\\
    \pb{\ham}{\lin} &= \partial_c g^{ab}\, p_a p_b \ta{c}{d} q^d - g^{ab} p_c p_{(a} \ta{c}{b)} \label{eq:ham_lin_pb},
\end{align}
where $c_{\alpha\beta}^\gamma$ are the structure constants of the group. From \eq{lie_algebra}, we see that the closure of the vector constraints on themselves is guaranteed provided $\mathcal G$ is a Lie algebra. The PB's \eq{ham_lin_pb} vanish provided \eq{killing_loc} is satisfied. Thus, the closure of the constraints is guaranteed by the \emph{global} gauge invariance of the action.

Because of the important role played by \eq{killing_loc}, it is illuminating to see the conditions under which \eq{killing_loc} is satisfied for particular models. In translationally invariant non--relativistic particle models the generators given in Table~(\ref{tbl:sim group gen}) are used in \eq{killing_loc}. Being careful about particle and spatial indices (particle indices are labeled by $I$ and spatial indices are indicated by arrows) leads to the following condition on the potential
\equa{
    \sum_I \vec\nabla_I V = 0,
}
where $\vec\nabla_I = \diby{}{\vec q_I}$. This requires that the potential be translationally invariant. It is satisfied by potentials that are functions of the differences between the coordinates. The same argument applied to the rotations leads to a similar result: the potential must be rotationally invariant. The dilatations are different. They imply the following condition on the potential
\equa{\label{eq:dilatation_cc}
    \partial_c V\, q^c = -2V.
}
By Euler's theorem, this implies that the potential should be homogeneous of order $-2$ in $q^c$.

While the gauge invariance of the action is guaranteed for the rotations and translations by the gauge invariance of the potential, it is not for the dilatations. This is because the kinetic term has conformal weight $+2$ under global scale transformations of the $q$'s. Thus, the potential must have conformal weight $-2$ if the action is to be scale invariant. This is just the requirement \eq{dilatation_cc} and is equivalent to the \emph{consistency conditions} obtained in \cite{barbour:scale_inv_particles} but derived from different motivations and in the canonical formalism. The key message to take from this result is that equivariance of the metric does not necessary imply gauge invariance of the potential. If the non--gauge invariance of the potential is compensated by the non--gauge invariance of the kinetic term, the entire action can still be gauge invariant. The dilatations provide an example of this.

It is possible to work out the gauge transformations generated by the linear constraints $\lin$. Computing the PBs $\pb{q}{\lin}$ and $\pb{\wa}{\ham_\beta}$ we find $q$ and $\phi$ transform as
\begin{align}
    q^a &\ra e^{-\zeta^\alpha \ta{a}{b}} q^b \notag \\
    \wa &\ra \wa + \zeta^\alpha \label{eq:gauge_trans}
 \end{align}
under large gauge transformations parameterized by $\zeta^\alpha$. This is what Barbour call \emph{banal invariance} in \cite{barbour:scale_inv_particles}. From the canonical analysis, this is a genuine gauge invariance of the theory. It is simply a statement that the $\phi$ field, because of how it was introduced into the theory, is purely auxiliary. In practice, the $\phi$ field is formally equivalent a St\"uckelberg field \cite{stueckelberg}, in which case the banal invariance is a \emph{splitting} symmetry.

\subsection{Best--matching variation} \label{sec:mach_variation}

Before computing the canonical equations of motion and solving the constraints, we will describe the canonical version of the best--matching variation. This is analogous to variation used in the Lagrangian formulation that we introduced in Section~(\ref{sec:bm variation}). However, there are important differences and new insights that should be highlighted.

As was discussed in detail, the $\phi$'s should be varied freely on the endpoints of \emph{any} interval along the trajectory. To see how this is carried out in the canonical formalism, consider the canonical action:
\equa{\label{eq:omega_pi_action}
    S[q,p,\phi, \pi] = \int d\lambda \lf[ p\cdot\dot{q} + \dot\phi\cdot\pi - h(q,p, \phi, \pi) \rt].
}
We are concerned only with variations of the $\phi$'s and $\pi$'s since the $p$'s and $q$'s are treated as standard phase space variables. We need to determine the conditions under which the action will vanish if the $\phi$'s and the $\pi$'s are varied freely at the endpoints. The variation with respect to the $\pi$'s vanishes provided $\dot\phi = \diby{h}{p} = \pb{q}{h}$ regardless of the conditions on the endpoints. Thus, Hamilton's first equation is unchanged by the free endpoint condition. However, the procedure leading to Hamilton's second equation is modified.

After integration by parts, the variation of \eq{omega_pi_action} with respect to $\phi$ is
\equa{
    \delta_\phi S[q,p,\phi,\pi] = -\int d\lambda \lf[ \diby{h}{\phi} + \dot \pi \rt] \, \delta\phi + \lf. \pi\, \delta \phi \rt|_{\lambda_\text{in}}^{\lambda_\text{fin}} = 0.
}
The first term implies Hamilton's second equation
\equa{
  \dot{\pi} = -\diby{h}{\phi} = \pb{\pi}{h}.
}
However, because $\delta \phi$ is \emph{not} equal to zero on the endpoints, the second term will only vanish if $\pi(\lambda_\text{in}) = \pi(\lambda_\text{fin}) = 0$. This single free endpoint condition, however, is not enough. In order for the $\phi$ fields to be completely arbitrary, the solutions should be independent of where the endpoints are taken along the trajectory. This implies the canonical \emph{best--matching condition}, $\pi(\lambda) = 0$, \emph{everywhere}. The best--matching condition is an additional equation of motion. In Dirac's language, it is a \emph{weak} equation to be applied only \emph{after} taking Poisson brackets.

For metrics satisfying \eq{killing_field}, $\phi$ is a \emph{cyclic} variable. This means that it enters the action only through its dependence on $\dot\phi$. In this case, $\pb{\pi}{h} = 0$ identically so that, by Hamilton's second equation, $\pi$ is a constant of motion. Normally, this constant of motion would be set by the initial and final data. The main effect of applying the best--matching condition is to set this constant equal to zero. When combined with the linear constraints $\lin\approx 0$, the best--matching condition implies
\begin{equation}
    p_a \ta{a}{b} q^b \approx 0
\end{equation}
which is the vanishing of the Noether charge associated to a symmetry of the action under $\cG$. In the Lagrangian language, this results immediately from applying the Lagrangian form of the best--matching condition. In the canonical formalism, this is a two step process. First, we must impose the linear constraints $\lin$, then we must apply the best--matching condition. This important difference highlights the advantages of the canonical framework. We can see that the best--matching condition imposes a new symmetry on the theory: invariance under $\phi \to \phi + \zeta$. Unlike the splitting symmetry (or banal invariance), this gauge invariance has a real physical significance: it generates the background independence of the theory because it projects the theory down to the relational phase space. This observation will be used to establish a precise definition of background independence in Section~(\ref{sec:BI_BD}).

\subsection{Cyclic variables and Lagrange multipliers}

We pause for a brief comment on the meaning of the best--matching variation. As we briefly mentioned, a \emph{cyclic} variable is one that enters a theory only through its time derivative. Conversely, a Lagrange multiplier appears in a theory without any time derivatives. In the equivariant case, $\wa$ drops out of the theory because of the equivariance property. This means that $\wa$ is a cyclic variable. Since
\begin{equation}
    \diby L {\wa} = 0
\end{equation}
identically for cyclic variables, the best--matching condition
\begin{equation}
    \diby L {\dot\phi^\alpha} = 0
\end{equation}
is the only non--trivial equation of motion needed to be satisfied by the theory. However, the best--matching condition on a cyclic variable is completely equivalent to a standard variation with respect to a Lagrange multiplier. This is because, for a Lagrange multiplier $\Lambda$, the Euler--Lagrange equations are
\begin{equation}
    \diby L \Lambda = 0.
\end{equation}
Thus, in the equivariant case, one can think of the best--matching variation as a variation where the connection is treated as a Lagrange multiplier. This is the way that the \emph{lapse} and \emph{shift} are varied in the ADM action. However, we see here that it is more appropriate to think of these as cyclic variables varied by a best--matching variation. This approach is mathematically equivalent but is based on a more solid conceptual picture.

\subsection{Classical equations of motion}\label{sec:part_jac_eom}

We are now in a position to compute the canonical equations of motion of our theory. The definitions of the momenta, (\ref{eq:pa_jac}) and (\ref{eq:pia_jac}), imply that the canonical Hamiltonian vanishes, as it must for a reparametrization invariant theory. Thus, the total Hamiltonian $\htot$ is proportional to the constraints
\equa{
    \htot = N\ham + \Na\lin
}
where the \emph{lapse}, $N$, and \emph{shift}, $\Na$, are just Lagrange multipliers enforcing the scalar and vector constraints respectively. We use this terminology to emphasize that these Lagrange multipliers play the same role as the lapse and shift in general relativity.

The best--matching variation implies
\begin{align}
    \dot{\phi}^\alpha &= \pb{\wa}{\htot} = \Na, \\
    \dot{\pi}_\alpha &= \pb{\pi}{\htot} = 0, \qand \\
    \pia &=0.
\end{align}
The $\wa$'s are seen to be genuinely arbitrary given that their derivatives are equal to the shift vectors. As expected, the $\pia$'s are found to be constants of motion set to zero by the best--matching condition. 

We now perform a standard variation of the $q$'s and $p$'s. A short calculation shows that Hamilton's first equation $\dot{q}^a = \pb{q^a}{\htot}$, can be re-written as
\equa{\label{eq:h1_pq}
    p_a = \frac{1}{2N} g_{ab}\, \lf.G^{-1}\rt.^b_c \diby{}{\lambda}\lf( G^c_d q^d \rt),
}
where we have made use of the definition $\lf.G^{-1}\rt.^a_b = \myexp{-\wa\ta{a}{b}}$. Note that $G$ can be rewritten in terms of the shift vectors using the equation of motion $\dw = \Na$. By (\ref{eq:killing_field}), we find that, in terms of barred quantities, (\ref{eq:h1_pq}) becomes
\equa{\label{eq:h1_pbar}
    \bar{p}_a = \frac{1}{2N} g_{ab}(\bar{q})\, \dot{\bar{q}}^b,
}
where $\bar{p}_a = \lf.G^{-1}\rt.^b_a p_b$.

Hamilton's second equation gives
\equa{
    \dot{p}_a = -N (\partial_a g^{bc}) p_b p_c + \Na p_b \ta{b}{a},
}
which, upon repeated use of (\ref{eq:killing_field}), leads to
\equa{
    \dot{\bar{p}}_a = -N (\bar{\partial}_a g^{bc}(\bar{q})) \bar{p}_b \bar{p}_c.
}
Thus, the equations of motion can now be written purely in terms of the best--matched quantities:
\equa{\label{eq:eom_int}
    \frac{1}{2N}\diby{}{\lambda}\lf( \frac{1}{2N} g_{ab}(\bar{q})\, \dot{\bar{q}}^b \rt) = -\frac{1}{2} \bar{p}_b \bar{p}_c \bar{\partial}_a g^{bc}(\bar{q}).
}

We can now use the conformal flatness of the metric $g_{ab} = \Omega \eta_{ab} = (-2V) \eta_{ab}$ and the scalar constraint $g^{ab}\, p_a p_b = 1 \ra \eta^{ab} p_a p_b = -2V$ to write (\ref{eq:eom_int}) in a more recognizable form. Identifying $\dot\tau \equiv -\frac{N}{V}$, (\ref{eq:eom_int}) reduces to
\equa{\label{eq:newton}
    \frac{\partial^2 \bar{q}^a}{\partial\tau^2} = -\bar{\partial}^a V(\bar{q}).
}
This, as expected, is Newton's $2^\text{nd}$ law with $\tau$ playing the role of Newtonian time and with the $q$'s replaced by their best--matched values. Note that we did not use the conformal flatness of the metric until the last step and then only to write our results in a more recognizable form. We note in passing that Newton's laws are just \eq{newton} written in the \emph{proper time} gauge, analogous to the similar gauge condition used in general relativity, where $N=1$ and $\Na = 0$. This special gauge also corresponds to Barbour's \emph{distinguished representation} \cite{barbour:scale_inv_particles}.

\subsection{Solving the constraints}

It is now possible to use Hamilton's first equation to invert the scalar and vector constraints and solve explicitly for the lapse and shift. This will allow us to write the equations of motion in terms of gauge invariant quantities having eliminated all auxiliary fields $\phi$. Solving for the lapse and shift allows us to write the equations of motion in a gauge invariant form. This will allow us to identify gauge independent quantities.

The shift can be solved for by inserting Hamilton's first equation
\equa{
    \dot q^a = 2N p_b g^{ab} - \Na \lf. t_\alpha \rt.^a_b q^b
}
into the vector constraint $\mathcal H_\alpha = \pia - p_a \ta{a}{b} q^b = 0$ after applying the best--matching condition $\pia = 0$. Inverting the result for $\Na$ gives
\equa{\label{eq:shift}
    \Na M_{\alpha\beta} = \eta_{ab} \dot{q}^a \lf. t_\beta \rt.^b_c q^c,
}
where
\equa{
    M_{\alpha\beta} = \eta_{ab} \ta{a}{c} \lf. t_\beta \rt.^b_d\, q^c q^d.
}
In the above, we have used $g_{ab} = \Omega \eta_{ab}$ and removed as factors $\Omega$ and $N$. The fact that $N$ drops out is what allows the scalar and vector constraints to decouple allowing the system to be easily solved. Theories of dynamic geometry are typically more sophisticated, and this separation is no longer possible. If $M_{\alpha\beta}$ is invertible, then $\Na$ is given formally in terms of its inverse $M^{\alpha\beta}$. In \scn{bid_observables} we shall give simple closed--form expressions for $\Na$ for non--relativistic particle models invariant under translations and dilitations. The inversion of $M_{\alpha\beta}$ for non--Abelian groups, such as the rotations in 3 dimensions, is formally possible but illuminating, closed--form expressions are difficult to produce.

The lapse can be solved for using (\ref{eq:killing_field}) and inserting Hamilton's first equation (\ref{eq:h1_pbar}) into the scalar constraint $\mathcal H = g^{ab} p_a p_b -1 = 0$. This gives
\equa{\label{eq:lapse}
    N = \frac{1}{2} \sqrt{g_{ab}(\bar{q})\, \dot{\bar{q}}^a\dot{\bar{q}}^b}.
}
Having already solved for the shift we can use it to compute $G^a_b(\wa)$ in the above expression using the equation of motion $\dw = \Na$. We can now express all equations of motion without reference to auxiliary quantities.

\section{Canonical best matching}\label{sec:can bm}

In this section, we will illustrate a powerful method for applying the best--matching procedure starting directly from phase space. Although the motivation behind the procedure, and in particular the canonical best--matching condition, comes from the Lagrangian framework where the configurations are fundamental and Mach's principles can be intuitively implemented, the canonical framework is mathematically more powerful and can bring additional insight into the gauge structure of the theory. The inspiration for this approach comes from \cite{Gomes:linking_paper}, which streamlined the dualization procedure presented in \cite{gryb:shape_dyn}, where shape dynamics was discovered. This approach to best matching is thus an important tool for understanding non--equivariant best matching and the structure of shape dynamics.

We start immediately from the phase space $\Gamma(q^a,p_a)$ coordinatized by the configuration variables $q^a$ and their conjugate momenta $p_a$. We equip $\Gamma$ with a symplectic 2--form that induces the Poisson brackets
\begin{equation}
    \pb{q^a}{p_b} = \delta^a_b.
\end{equation}
To define a geodesic principle on configuration space, we require that the initial speed of the system along the classical trajectory not affect the shape of the classical solutions. In the canonical language, this implies a Hamiltonian constraint $\ham$ that constrains $p_a$ to a unit vector on $\Gamma$. Thus,
\begin{equation}
    \ham = g^{ab} p_a p_b -1 \approx 0.
\end{equation}

The next step is to trivially extend the phase space $\Gamma(q,p) \to \Gamma_e(q,p,\wa,\pi )$ by introducing $\wa$ and its conjugate momenta $\pia$. The symplectic structure is extended such that the only new non--vanishing Poisson bracket is
\begin{equation}
    \pb{\wa}{\pi_\beta} = \delta^\alpha_\beta.
\end{equation}
To guarantee that $\wa$ is arbitrary, we can add the constraint
\begin{equation}
    \pia \approx 0,
\end{equation}
which does nothing but set $\dot\phi^\alpha$ to a Lagrange multiplier. Because $\ham$ is independent of $\wa$, $\pia\approx 0$ is trivially first class with respect to $\ham$. Note this is \emph{not} the best--matching condition since we have not yet done any matching. It is simply a constraint that does not affect the original geodesic principle.

The shifting is accomplished on phase space by performing a canonical transformation $ (q,\phi; p, \pi) \to (\bar q,\bar \phi; \bar p, \bar \pi)$ generated by the type--II generating functional $F(q,\phi; \bar p, \bar \pi)$ defined by
\begin{equation}
    F(q,\phi; \bar p, \bar \pi) = \bar p_b G(\phi)^b_a q^a + \wa \bar \pi_\alpha.
\end{equation}
The transformed $q$'s and $\phi$'s can easily be computed
\begin{align}
    \bar q^a &=\diby{F}{\bar p_a} = G^a_b(\phi) q^b \\
    \bar \phi^\alpha &=\diby{F}{\bar \pia} = \wa.
\end{align}
The $q$'s are just shifted by the action of $\cG$ and $\phi$ is unchanged. The transformed momenta are more interesting. Using
\begin{align}
 p_a &=\diby{F}{q^a} = G^b_a(\phi) \bar p_b \\
    \pia &=\diby{F}{\bar\phi^\alpha} = \bar\pi_\alpha + p_a \ta{a}{b} q^b,
\end{align}
we find that
\begin{align}
    \bar p_a &= (G^{-1})^b_a p_b  \\
    \bar \pi_\alpha &= \pia - p_a \ta a b q^b.
\end{align}
Clearly, $\bar p_a$ is the same $\bar p_a$ we found convenient to define when working out the classical equations of motion.

The Hamiltonian constraint transforms to
\begin{equation}
    \ham \to \bar p_a \bar p_b \bar g^{ab} - 1.
\end{equation}
Because of the equivariance property, this reduces to
\begin{equation}
    \ham \to p_a p_b g^{ab} - 1.
\end{equation}
Thus, in the equivariant case, $\ham$ is invariant under the best--matching canonical transformation. This also implies that it is first class with respect to the best--matching condition $\pia \approx 0$, a key difference from the non--equivariant case. As we will see, if $g$ is non--equivariant, the best--matching condition will only close with the Hamiltonian constraint under special circumstances.

Under the canonical transformation generated by $F$, the first class system of constraints $(\ham\approx 0, \pia\approx 0)$ transforms to
\begin{align}
    \ham &\approx 0 & \lin &= \pia - p_a \ta a b q^b \approx 0.
\end{align}
These are exactly the constraints obtained from the Legendre transform of the original best--matching action. At this point, the canonical transformation has done nothing to the original theory. To implement Mach's first principle, we must perform a best--matching variation with respect to $\wa$ by imposing the best--matching condition $\pia \approx 0$. Because of the equivariance property of $g$, $\wa$ does not appear anywhere in the Hamiltonian of the transformed theory. Thus, $\pia \approx 0$ is first class with all constraints. Since $\wa$ doesn't appear in the theory, we can integrate out $\wa$ and $\pia$ by imposing $\pia = 0$ strongly. This reduces the phase space back to the original one, leaving the constraints
\begin{align}
 \ham &\approx 0 & p_a \ta a b q^b &\approx 0.
\end{align}
Since the second constraint is the vanishing of the generalized momentum associated to the symmetry $\cG$, we have reproduced the results obtained from the Lagrangian approach.

\section{Gauge--independent observables} \label{sec:bid_observables}

The simple form of \eq{eom_int} and (\ref{eq:newton}) suggests there might be something fundamental about the best--matched coordinates $\bar{q}^a = G^a_b q^b$. In fact, as can be easily checked, they commute with the primary, first class vector constraints $\lin$. The $\bar{q}$'s are then invariant under the gauge transformations \eq{gauge_trans} generated by $\lin$. They do \emph{not}, however, commute with the quadratic scalar constraint. For this reason, they are non--perennial \emph{observables} in the language of Kucha\v r \cite{kuchar:time_int_qu_gr,Kuchar:can_qu_gra}, who argues that such quantities are \emph{the} physically meaningful observables of reparameterization invariant theories. Barbour and Foster take this argument further in \cite{barbour_foster:dirac_thm}, showing how Dirac's theorem fails for finite--dimensional reparameterization invariant theories. The reason for ignoring the non--commutativity of the observables with the scalar constraint is that the scalar constraint generates physically \emph{distinguishable} configurations. This is in contrast to the linear vector constraints, which generate physically \emph{indistinguishable} states. In this work, we will use Kucha\v r's language to describe these observables and see that, in all cases where the constraints can be solved explicitly, the $\bar{q}^a$ are manifestly relational observables.

The best--matched coordinates, $\bar{q}^a$, have a nice geometric interpretation. Using Hamilton's first equation for $\dw$, the $\bar q$'s can be written in terms of the shift as
\equa{\label{eq:holonomy}
    \bar{q}^a = \myexp{\wa\ta{a}{b}}q^b = \mathcal{P}\myexp{\int \Na\ta{a}{b}\, d\lambda} q^b,
}
where $\mathcal P$ implies path--ordered integration. Thus, the best--matched coordinates are obtained by subtracting the action of the open--path holonomy of the lapse (thought of as the pullback of the connection over $\acs$ onto the classical path) on the $q^b$'s. This subtracts all vertical motion of the $q$'s along the fiber bundle.

\subsection{Special cases}

The significance of the $\bar{q}$'s is more clearly seen by solving the constraints for specific symmetry groups. First, consider non--relativistic particle models invariant under translations. The $q$'s represent particle positions in 3--dimensional space. The $a$ indices can be split into a vector index, $i$, ranging from 1 to 3, and a particle index, $I$, ranging from 1 to the total number of particles in the system. Then, $a = iI$. For $\eta_{ab}$ we use the diagonal mass matrix attributing a mass $m_I$ to each particle. With the generators of translations (see Table~(\ref{tbl:sim group gen})), (\ref{eq:shift}) takes the form
\equa{
    \vec{N} = \frac{\sum_I m_I \vec{q}_I}{\sum_I m_I} \equiv \dot{\vec{q}}_{\text{cm}}.
}
The shift is the velocity of the center of mass $\vec{q}_{\text{cm}}$. Aside from an irrelevant integration constant, which can be taken to be zero, the auxiliary fields $\phi^\alpha$ represent the position of the center of mass. Inserting this result into \eq{holonomy}, the corrected coordinates are
\equa{
    \bar{q}^a = q^a - q^a_{\text{cm}}.
}
They represent the difference between the particles' positions and the center of mass of the system. This is clearly a relational observable. Furthermore, the non--physical quantity is the position of the center of mass since the theory is independent of its motion.

We can also treat models invariant under dilatations.\footnote{See \cite{barbour:scale_inv_particles} for more details on these models.} In this case, (\ref{eq:shift}) is easily invertible since there is only a single shift function, which we will call $s$. Using the same index conventions as before and the generators of dilatations from Table~(\ref{tbl:sim group gen}) we find
\equa{
    s = \diby{}{\lambda} \lf( -\frac{1}{2}\ln I \rt),
}
where $I = \sum_I m_I (\vec{q}_I)^2$ is the moment of inertia of the system. Aside from an overall integration constant, which can be set to zero, the auxiliary field is $-1/2$ times the log of the moment of inertia. Using \eq{holonomy}, the corrected coordinates are the original coordinates normalized by the square root of the moment of inertia
\equa{
    \bar{q}^a = \frac{q^a}{\sqrt{I}}.
}
Because $I$ contains two factors of $q$, $\bar{q}$ will be invariant under rescalings of the coordinates. Thus, the corrected coordinates are independent of an absolute scale.

The quantity $\tau$, which plays the role of the Newtonian time, can now be computed. It is a function of the lapse and the potential. Since the lapse is an explicit function of the corrected coordinates, it will be observable. Using the definition $\dot\tau \equiv -\frac{N}{V}$ and (\ref{eq:lapse}), $\tau$ is simply
\equa{
    \tau = \int d\lambda\, \sqrt{-\frac{T(\bar{q})}{V(\bar{q})}},
}
where $T = \frac{1}{2} m_{ab} \dot{\bar{q}}^a\dot{\bar{q}}^b$ is the relational kinetic energy of the system. $\tau$ is independent of $\lambda$ and observable within the system. Thus, once the constraints have been solved for, the equations of motion (\ref{eq:newton}) are in a particularly convenient gauge--independent form. This definition of $\tau$ corresponds to Barbour and Bertotti's \emph{ephemeris time} \cite{barbour:nature, barbourbertotti:mach}.

\section{Background dependence and independence}\label{sec:BI_BD}

The presence of a symmetry of the configurations of $\acs$ allows for a distinction between two types of theories:\footnote{These definitions should apply equally to equivariant and non--equivariant theories.}
\begin{itemize}
  \item those that attribute physical significance to the exact location of the configuration variables along the fiber generated by the symmetry. We will call these theories \emph{Background Dependent (BD)}.
  \item those that \emph{do not} attribute any physical significance to the exact location of the configuration variables along the fiber. These theories will be called \emph{Background Independent (BI)}.
\end{itemize}
Based on these definitions, it would seem odd even to consider BD theories as they distinguish between members of an equivalence class. These theories are useful nevertheless whenever there is an \emph{emergent} background that breaks the symmetry in question at an effective level. This is the situation that arises in Newton's bucket experiment. Although one \emph{could} treat this system in a relational BI way, this would be cumbersome because the exact motions of the fixed stars would have to be known in order to solve the angular momentum constraints. Instead, it is simplest to treat this, as Newton did, as a system with an absolute background given by the rest frame of the fixed stars. In this way, we see that BD theories shouldn't be fundamental but are, nevertheless, useful to describe the effective physics of a subsystem of the universe.

Best matching provides a framework for making our definitions of BD and BI more precise. Whenever there are symmetries in the configurations  it is possible to introduce auxiliary fields $\wa$ whose role is nothing more than to parametrize the symmetry. Indeed, making the $\phi$'s dynamical can be very useful since, as we have seen, the full power of Dirac's formalism \cite{dirac:lectures} can be used to study the dynamical effects of the symmetry. In addition, introducing the $\wa$ fields gives us the freedom to distinguish between BD and BI theories as follows:
\begin{itemize}
  \item \emph{BD theories} are those that vary the $\wa$ fields in the standard way using \emph{fixed} endpoints. This requires the specification of appropriate initial and final data, which is considered to be physically meaningful.\footnote{So as not to add redundancy to the boundary conditions, we can set $\wa(\lambda_{\text{in}}) = \wa(\lambda_{\text{fin}}) = 0$ without loss of generality. The boundary conditions on the $q$'s will then contain all the information about the absolute position of the $q$'s along the fiber.}
  \item \emph{BI theories} are those that vary the $\wa$ fields using a best--matching variation.
\end{itemize}

Given these definitions, we can understand the physical difference between BD and BI theories by considering the form of the vector constraints (\ref{eq:lin}). In the BD case, the $\pia$'s are constants of motion. These constants are determined uniquely by the initial conditions on the $q$'s. However, in the BI theory, the constants of motion are irrelevant and are seen as unphysical. The initial data on $\mathcal R$ cannot affect their value. As a result, the BD theory requires more inputs in order to give a well defined evolution. The difference is given exactly by the dimension of the symmetry group. This is precisely in accordance with Poincar\' e's principle.

\section{Time and Hamilton's principle}

As we have seen in Section~(\ref{sec:equiv S}), Parametrized Hamilton's Principle (PHP) is an alternative to a square root action principle for determining the dynamics of a system. It is still a geodesic principle on configuration space but the square root is disposed of in place of a mathematically simpler action. The cost of having this simpler action is the introduction of an auxiliary field whose role is to restore the reparameterization invariance. PHP has the advantage over square root actions in that it singles out a preferred parametrization of the geodesics through a choice of normalization of the scalar constraint. For a standard normalization, the preferred parameter is just Barbour and Bertotti's ephemeris time. In geometrodynamics, this quantity will be related to the proper time of a freely falling observer. In the case of the particle models, PHP is just the action principle for the parameterized particle treated, for example, in \cite{lanczos:mechanics}.

Because time is now dynamical, we can use the definitions of BI and BD from \scn{BI_BD} to distinguish between theories that have a background time and those that are timeless. As we would expect, Newton's theory, which contains explicitly an absolute time, can be obtained from PHP with a background time. Alternatively, a timeless theory is obtained from PHP by keeping the time background independent.

\subsection{Action and Hamiltonian}

To simplify the discussion, we will ignore for the moment the spatial symmetries. This will avoid having to deal with the linear constraints. Comparison to the equations of previous sections can either be made by setting the shift, $\Na$, equal to zero or by unbarring quantities. It can be verified that neglecting the spatial symmetries does not affect the discussions regarding time  \cite{sg:mach_time}.

PHP is defined by the action
\equa{\label{eq:ham_action}
	S_H = \int_{q_{\text{in}}}^{q_{\text{fin}}} d\lambda\, \frac{1}{2}\lf[ \frac{1}{\dot\tau} \gamma_{ab} \dot{q}^a \dot{q}^b + \dot\tau \Omega \rt].
}
The Lagrangian has the form of Hamilton's principle $T - V$ (recall that $\Omega = -2V$) but the absolute time $\tau$ has been promoted to a dynamical variable by parameterizing it with the auxiliary variable $\lambda$. This explains the name: Parametrized Hamilton's Principle. The particular normalization used takes advantage of the conformal split of the metric and singles out BB's ephemeris time as a preferred parameter for the geodesics.\footnote{Alternatively, one could split the action as $S_H = \int d\lambda\, \lf[ \frac{1}{\dot\tau} \Omega\gamma_{ab} \dot{q}^a \dot{q}^b - \dot\tau \rt]$ without changing the equations of motion. This action would single out an affine parameter for the geodesics. It corresponds to dividing the scalar constraint by $\Omega$.} For simplicity, we restrict ourselves to metrics of the form: $\gamma_{ab} = \eta_{ab}$.

We can perform a Legendre transform to find the Hamiltonian of the system. Defining the momenta
\begin{align}
	p_a &= \frac{\delta S_H}{\delta \dot{q}^a} = \frac{1}{\dot\tau} \eta_{ab} \dot q^b, \qand \label{eq:p_a}\\
	p_0 &= \frac{\delta S_H}{\delta \dot\tau} = -\frac{1}{2} \lf[ \frac{\eta_{ab}\dq^a\dq^b}{\dot\tau^2} - \Omega \rt] \label{eq:p_0}
\end{align}
we note that they obey the scalar constraint
\equa{
	\ham_{\text{Ham}} \equiv \frac{1}{2}\lf( \eta^{ab} p_a p_b - \Omega \rt) + p_0 = \frac{\Omega}{2} \ham_\text{Jacobi} + p_0 = 0.
}
The appearance of the $p_0$ term is the only difference, other than the factor $\Omega$, between the timeless geodesic theory and PHP (the factor 2 is purely conventional). The factor $\Omega$ can be absorbed by a field redefinition of the lapse and has no bearing on physical observables. It is a relic of our choice of ephemeris time to parametrize geodesics. Using the definitions \eq{p_a} and \eq{p_0}, we find that the canonical Hamiltonian is identically zero, as it must be for a reparameterization invariant theory. Thus, the total Hamiltonian is
\equa{\label{eq:pnm_ham}
	H_\text{T} = N\ham = N \lf( \frac{1}{2}\eta^{ab} p_a p_b -\frac{1}{2} \Omega + p_0 \rt).
}

\subsection{BI theory}

In PHP, time is promoted to a configuration space variable. The symmetry associated with translating the origin of time is reflected in the invariance of the action under time translations $\tau \ra \tau + a$, where $a$ is a constant. In \cite{sg:mach_time}, it is shown that applying the best--matching procedure to this symmetry is equivalent to treating $\tau$ itself as an auxiliary field. To make the theory background independent with respect to the temporal symmetries, we follow the procedure outlined in \scn{BI_BD} and impose the best--matching condition after evaluating the Poisson brackets. In this case, the best--matching condition takes the form $p_0 = 0$.

We pause for a brief observation. Since the best--matching variation of a cyclic variable is equivalent to the standard variation of a Lagrange multiplier, we can replace $\dot\tau$ with $N$ when doing a background independent formulation of PHP. Then, the action (\ref{eq:ham_action}) bears a striking resemblance to the ADM action. This illustrates why the ADM action is background independent as far as time is concerned. However, the ADM action hides the possibility of introducing a background time (following the procedure given in the next section). Considering this, it might be more enlightening to think of the lapse as a cyclic variable subject to best--matching variation as is done in \cite{barbour:scale_inv_particles} and \cite{Anderson:cyclic_ADM}.

In order to compare this to the standard timeless geodesic theory, it is instructive to work out the classical equations of motion
\begin{align}
 	\dq^a &= \pb{q^a}{H_\text{T}} = N \eta^{ab} p_b, \\
	\dot p_a&= \pb{p_a}{H_\text{T}} = N \partial_a \lf( \frac{\Omega}{2} \rt) = -N \partial_a V  \\
	\dot\tau &= \pb{\tau}{H_\text{T}} = N, \qand \label{eq:ham1_tau}\\
	\dot p_0 &= \pb{p_0}{H_\text{T}} = 0. \label{eq:ham_p0}
\end{align}
(\ref{eq:ham1_tau}) reinforces the fact that $\tau$ is an auxiliary. It is straightforward to show that the above system of equations implies
\equa{
	\frac{\partial^2 q^a}{\partial\tau^2} = -\partial^a V(q).
}
This is Newton's $2^{\text{nd}}$ law. Solving the scalar constraint gives an explicit equation for $\tau$,
\equa{\label{eq:ham_tau}
	\dot\tau = \sqrt{\frac 1 \Omega \eta_{ab} \dq^a \dq^b } = \sqrt{\frac{T}{-V}},
}
using the definitions for $V$ and $T$ given in \scn{jacobi_principle}. This is precisely the expression for the ephemeris time $\tau$ defined in the usual timeless geodesic theory. It should be noted that the best--matching condition implies that the integration constant of (\ref{eq:ham_p0}) is zero. Use was made of this to deduce (\ref{eq:ham_tau}). From this it is clear that the BI theory is classically equivalent to the timeless geodesic theory.

One can take this further and compare the two theories quantum mechanically. Noticing that the canonical action is linear in $\dot\tau$ and $p_0$, we can integrate out $\tau$ without affecting the quantum theory and use the best--matching condition $p_0 = 0$ to reduce the scalar constraint to
\equa{
	\ham = g^{ab} p_a p_b - 1 = 0
}
(after factoring $\Omega$). This is the scalar constraint (\ref{eq:jacobi_constraint}) of the timeless geodesic theory. With $\tau$ now defined by (\ref{eq:ham_tau}), the canonical theories are identical. Thus, their canonical quantizations should also match. For more details on the equivalence of these theories quantum mechanically, see \cite{sg:mach_time}, where the path integrals for these theories are worked out in detail.

\subsection{BD theory}\label{sec:toy BD theory}

For the BD theory, we do not impose the best--matching condition. Integration of (\ref{eq:ham_p0}) implies $p_0 \equiv -E$. The only effect that this has on the classical theory is to alter the formula for $\tau$ to
\equa{
	\dot\tau = \sqrt{\frac{T}{E-V}}.
}
Now an initial condition is imposed on $\tau$ that fixes the value of $E$ and violates Poincar\' e's principle. Thus, $\tau$ is equivalent to a Newtonian absolute time. Note that inserting a background time would have been impossible if we started with the ADM form of PHP.

Strictly speaking, there is a difference between $E$, defined as the negative of the momentum canonically conjugate to time, and $E'$, which is just the constant part of $V = -E' + V'$. Together they form what we would normally think of as the total energy $E_\text{tot} = E + E'$ of the system. $E'$ is freely specifiable and plays the role of a fundamental constant of nature while $E_\text{tot}$ is fixed by the initial conditions on $\tau$. In the classical theory, it is unnecessary to make a distinction between $E_\text{tot}$ and $E$. However, in the quantum theory, this distinction is important because of the possible running of constants of nature like $E'$. In general relativity, the role of $E'$ is played by the cosmological constant. As a result, this distinction may be relevant to the cosmological constant problem \cite{Smolin:unimodular_grav}.

\subsection{A problem of time}

In the classical theory, it seems that there is only a very subtle difference between the BI and BD theories. The difference amounts to the ability to impose boundary conditions on $\tau$ that constrain the total energy. However, the quantum theories are drastically different. Using Dirac's procedure, we promote the scalar constraint to an operator constraint on the wavefunction $\Psi$. In the BD theory, Dirac's procedure applied to the Hamiltonian (\ref{eq:pnm_ham}) gives the time dependent Schr\"odinger equation
\equa{
	\hat\ham \Psi = \lf[ \frac{1}{2}\eta^{ab} \hat p_a \hat p_b + V(\hat q) + \hat p_0 \rt] \Psi  = 0.
}
In a configuration basis, $p_0 = -i\diby{}{\tau}$. Thus, the above is indeed the standard Schr\"odinger equation. However, in the BD theory, the best--matching condition requires $p_0 = 0$ leaving instead the time \emph{independent} Schr\"odinger equation
\equa{
	\hat\ham \Psi = \lf[ \frac{1}{2}\eta^{ab} \hat p_a \hat p_b + V(\hat q) - E' \rt] \Psi  = 0,
}
where we have explicitly removed the constant part of the potential. While it is easy to define an inner product in the BD theory under which evolution will be unitary this is not the case in the BI theory. This makes it difficult to define a Hilbert space for the BI theory (at least at the level of the entire universe). The difficulties associated with this can be called a \emph{problem of time} similar to what happens in quantum geometrodynamics.\footnote{In geometrodynamics, there are \emph{additional} complications associated with foliation invariance or many-fingered time. As far as I know, these have no analogues in the finite--dimensional models.} It is interesting to note that, in finite--dimensional models, one can eliminate this problem of time by artificially introducing a background time. In Section~(\ref{sec:unimolduar}), we study the effects of applying the same procedure to geometrodynamics and are led to unimodular gravity. Clearly the issue of background independence is of vital importance in the quantum theory. This will have important implications in any quantum theory of gravity.

%--------------NON--EQUIVARIANT BEST MATCHING -------------------------
%----------------------------------------------------------------------
% NON--EQUIVARIANT BEST MATCHING
%----------------------------------------------------------------------
% This section explains how one can do best matching even when the
% metric is non--equivariant on the PFB. We introduce an equivariantization
% procedure that leads to the idea of a linking theory and a symmetry
% trading mechanism.
%
% Jun 2: This is VERY rough... the order will change considerably!
%
%======================================================================
\chapter{Non--equivariant best matching}\label{chap:non_equivariant_bm}
%======================================================================

In the last chapters, we saw how Mach's principles can be clearly stated and implemented using best matching. We then saw that best matching is naturally seen as a way of defining a connection on a principal fiber bundle constructed by quotienting a redundant configuration space, $\acs$, by some suitably defined gauge group, $\cG$. In this way, relational theories are seen to be best understood as gauge theories on configuration space. One normally motivates the gauge theory approach by starting with a global symmetry of the action and then \emph{gauging} this symmetry by making it local using gauge covariant derivatives. The best--matching procedure does not follow this paradigm. Instead, the gauging process is derived from a natural procedure motivated by Mach's principles. This raises the question: does the deeper understanding of the gauge principle provided by best matching suggest any natural generalizations? Indeed, there is at least one possibility that emerges. This involves best matching an action that is \emph{not} globally invariant under the desired symmetry. This generalization of the gauge principle can lead to the construction of dual theories that have the desired symmetry. In particular, GR is a theory of this type with respect to 3d conformal transformations. Its conformal dual is shape dynamics.

In this chapter, we will illustrate the key properties of non--equivariant best--matched theories and show how they can lead to a dual representation. These features will be illustrated using a toy model that possesses many key features of shape dynamics. This toy model will be used to motivate a general dualization procedure that one can follow starting from a non--equivariant theory. Then, we will show how canonical best matching provides a powerful framework that draws a link between the original theory and its dual through a series of gauge fixing conditions required by the best--matching variation. To begin with, we will make some general geometric observations that will set the foundation for our discussion.

\section{Non--equivariance}

In Chapter~(\ref{chap:equiv bm}), we only considered metrics, $g_{ab}$, on configuration space that are equivariant under the action of the group $\cG$. This implied the equivariance condition
\begin{equation}
    G^c_a g_{cd}(Gq) G^d_b = g_{ab}(q), \label{eq:equiv cond}
\end{equation}
which, in turn, implied that the original action
\begin{equation}
    S = \int d\lambda \sqrt { g_{ab} \dot q^a \dot q^b }
\end{equation}
is invariant under a global $\lambda$ independent gauge transformation $q^a \to G^a_b(\wa) q^b$. This can easily be seen by noting that for $\dot\wa = 0$, only the metric transforms to the LHS of \eq{equiv cond}. Thus, the equivariance condition immediately implies that $S$ is invariant.

One can understand this from geometric considerations on the PFB $\acs$. Locally, \eq{equiv cond} takes the form of the vanishing of the Lie derivative of the metric in the direction of the flow generated by $\cG$. This means that the fibers are the integral curves of Killing vector fields of $g_{ab}$. Now consider a particular curve, $\gamma$, living on a section of $\acs$ defined by the best--matching condition (ie, a curve on the horizontal lift of $\acs$). The value of the best--matched action is defined as the length of $\gamma$ in terms of the metric $g_{ab}$. If one shifts $\gamma$ by a \emph{constant amount} along the fibers then, since these are Killing directions of $g_{ab}$, the length of $\gamma$, and thus the value of $S$, will remain unchanged. Thus, there is a dim($\cG$) parameter family of paths $\gamma$, related by global gauge transformations, that minimize $S$. Nevertheless, these all project down to the \emph{same} path in the base space $\rcs$. In the equivariant approach, it is this projectability that makes it possible to define a unique geodesic principle on the reduced configuration space $\rcs$. In the non--equivariant approach, this argument breaks down. Nevertheless, it is still possible to use a best matching connection to define a geodesic principle on the reduced configuration space.

In the non--equivariant case, what is lost is the invariance of the action under constant translations of $\gamma$ in the vertical directions of $\acs$. However, it is still possible that there exists a particular choice of $\gamma$, corresponding to a preferred section in $\acs$, that minimizes $S$. One can then define a metric on $\rcs$ by taking the value of $g_{ab}$ on this preferred section. This still provides a geodesic principle on $\rcs$ that does not require any specification of initial data along the vertical directions of $\acs$. Thus, the theory will still satisfy Mach's principles as stated by Poincar\'e (i.e., that the freely specifiable initial data of the theory must correspond to a point and a direction on the reduced configuration space). In other words, if there exists a particular section for which the length of $\gamma$ is extremized, then this extremization condition selects a position in the vertical direction of $\acs$ for which no freely specifiable initial data can be given. Thus, the variational procedure itself performs the required quotienting from $\acs$ to $\rcs$.

There are two interesting consequences to this. First, it is possible to define a background independent relational theory with respect to a symmetry that doesn't exist as a global symmetry of the original action. This seems contradictory the spirit of the gauge principle and many definitions of background independence. Second, it provides a way of constructing a dual theory that \emph{does} have the desired global symmetry. This can be accomplished in the following way. Once the preferred section on $\acs$ has been found that minimizes $S$, a metric on $\rcs$ can be defined by taking the value of $g_{ab}$ on this section. It is then possible to equivariantly lift this metric back to $\acs$ to get a truly equivariant metric over all of $\acs$. That new metric can be used to define the dual theory on $\acs$. This is the configuration space picture of what happens in the dualization procedure used to construct shape dynamics. In the coming sections, we will be mainly concerned with the phase space picture that is more powerful mathematically. However, the configuration space picture is still useful because it gives a nice (though somewhat restrictive) intuitive picture for what is going on. 

To conclude this section, we point out what conditions need to be satisfied in order for non--equivariant theories to be consistent. The main thing that can go wrong, from the point of view of the Hamiltonian picture, is that the Hamiltonian constraint is no longer invariant under $\cG$--transformations. The negative effect of this is that the best--matching constraint, imposed by the best--matching variation, will no longer be first class with respect to it. This could lead to an inconsistent system because the constraint algebra is not closed. However, in fortunate situations, the constraints reduce to a first class system in a particular gauge. This is a standard technique for dealing with second class systems, although there is no guarantee that such gauges exist. The fact that this is the case in GR is still a mystery. Whether this is a fortunate accident or whether there is some deeper principle that we have not yet discovered at work behind this mechanism is a question of top priority.

\section{Toy model}

We begin our discussion of non--equivariant best matching by developing, in detail, a toy model that contains many of the key features of shape dynamics. Our strategy will be to first recognize the non--equivariance then outline a step--by--step strategy for finding a consistent theory then constructing the dual theory. We will see that the Hamiltonian constraints will be first class with respect the best--matching condition. To obtain a first class system, we can find a particular gauge for the Hamiltonian constraints where the system is first class. From this, it is possible to construct the dual theory and compare its classical equations of motion to the original. Accordingly, we start by \emph{gauge fixing} some of the symmetries of the original theory then \emph{replace} them with different symmetries. Remarkably, this achieves a \emph{trading} of symmetries. Thus, what we will present is an algorithm for trading one symmetry for another.

The toy model consists of $n$ independent harmonic oscillators in $d$ dimensions. The action is\footnote{Uppercase indices are only summed when explicitly shown. Also, $v^2 \equiv v^a v^b \delta_{ab}$ for vectors and $u^2 \equiv u_a u_b \delta^{ab}$ for covectors.}
\begin{equation}
    S = \int_{t_0}^{t_1} dt\, \sum_{I=1}^{n} \frac 1 2 \left[ \lf(\frac{dq_I}{dt}(t) \rt)^2  - k q^2_I(t) \right].
\end{equation}
The masses have been set to 1, for convenience. We will need to consider negative spring constants $k<0$ to avoid imaginary solutions. Unstable solutions will not worry us here since we are interested mainly in the symmetries of the toy model, which we do not consider a physical theory.

We write the above action in a reparametrization invariant form by introducing the auxiliary fields $N_I$ and the arbitrary label $\lambda$
\begin{equation}
    S = \int^{\lambda_1}_{\lambda_0} d\lambda\, \sum_{I=1}^{n} \left[ \frac{\dot q^2_I(\lambda)}{2N_I}  + N_I \lf( \frac{\absk}{2} q^2_I(\lambda) + \mathcal{E}_I \rt) \right],
\end{equation}
where we have restricted to strictly negative $k$. We choose to parametrize the trajectory of each particle independently and allow the time variables to be varied dynamically. The $N_I$'s represent the $\lambda$ derivatives of the time variables and mimic the local lapse functions of geometrodynamics, whose terminology we will borrow. The total energy of each particle $\mathcal{E}_I$ essentially specifies how much time should elapse for each particle in the interval $\lambda_1 - \lambda_0$ so that this theory is equivalent to $n$ harmonic oscillators with different ``local'' time variables. We have introduced the local reparametrization invariance as a way of studying this feature of geometrodynamics in a finite dimensional toy model. It is not meant to have any direct physical significance.

After performing the Legendre transform $S = \int d\lambda \sum_{I=1}^n \lf( \pi^I_a \dot q^a_I - H_0(q, \pi)\rt)$, where $\pi^I_a = \frac{\delta S}{\delta \dot q^a_I}$, the local reparametrization invariance leads to $n$ first class constraints, $\chi_I$. The total Hamiltonian is
\begin{align}\label{original hamiltonian}
    H_0 &= \sum_{I=1}^n N^I \chi_I, \text{ where} \\
    \chi_I &= \frac 1 2 \lf( \pi^2_I - \absk q^2_I \rt) - \mathcal{E}_I
\end{align}
and the phase space $\Gamma(q,\pi)$ is equipped with the Poisson bracket
\begin{equation}
    \pb{q^a_I}{\pi_{b}^J} = \delta^a_b\delta_{I}^J.
\end{equation}
Our goal, as explained below, will be to exchange all but one of the first class $\chi_I$'s for new first class conformal constraints.

\subsection{Identify the symmetries to best match}

Our objective is to best--match this theory with respect to an analogue of local conformal invariance. However, we must be careful when choosing the precise symmetry to best match. A na\" ive choice would be to best match with respect to the symmetry, $\mathcal S$, defined by identifying the configurations
\begin{equation}
    q_I^a \to e^{\phi_I} q^a_I,
\end{equation}
so that each particle gets its own conformal factor, $\phi_I$. Then, we would be trading \emph{all} the $\chi_I$'s for this new symmetry. This choice, however, will lead to frozen dynamics because one global $\chi$ must be left over to generate dynamical trajectories on phase space\footnote{We draw the reader's attention to \cite{barbour_foster:dirac_thm}, which suggests that such a global $\chi$ does not generate a gauge transformation but a genuine physical
evolution.}. Thus, we need to leave at least one linear combination of $\chi$'s unchanged. Because of this, our first -- and most subjective -- step is to identify a subset of $\mathcal S$ to best match.

A natural choice is the symmetry $\mathcal S / \mathcal V$, where $\mathcal V$ represents all configurations that share the same total moment of inertia. Since the total moment of inertia sets the global scale of the system, $\mathcal S / \mathcal V$ identifies all locally rescaled  $q$'s that share the same global scale. This symmetry requires invariance under the transformation
\begin{equation}\label{eq: q symmetry}
    q_I^a \to e^{\hat\phi_I} q^a_I,
\end{equation}
where $\hat\phi$ obeys the identity
\begin{equation}\label{eq: q identity}
    \mean{\hat\phi} = 0.
\end{equation}
The \emph{mean} operator, is defined as $\mean{\cdot} \equiv \frac 1 n \sum_{I=1}^n \cdot_I$, where $\cdot$ signifies inclusion of the field we would like to take the mean of. The idea is to trade \emph{all but one} of the \emph{local} Hamiltonian constraints for constraints that generate local rescalings. In this particle model, ``local'' means ``individual'' for each particle. As is straightforward to show, by imposing the constraint \eq{ q identity}, we restrict ourselves to rescalings that do not change the total moment of inertia of the system. This restriction on the rescalings mimics the volume preserving condition in shape dynamics. After performing the symmetry trading, we will still have a single Hamiltonian constraint left over to generate dynamics. This is crucial for obtaining a non--trivial theory.

We can impose the identity (\ref{eq: q identity}) explicitly by writing $\hat\phi$ in terms of the local scale factors $\phi_I$
\begin{equation}\label{eq: hat phi def}
    \hat\phi_I = \phi_I - \mean{\phi}.
\end{equation}
Thus, we parametrize the group $\mathcal S / \mathcal V$ redundantly by $\mathcal S$. The fact that the $\chi_I$'s are \emph{not} invariant under this symmetry implies that we have a non--equivariant best--matching theory with respect to this symmetry.

\subsection{Perform canonical best matching}\label{sec:toy can bm}

The next step is to best match the symmetry $\mathcal S / \mathcal V$ following the procedure outlined in Section~(\ref{sec:can bm}). This is achieved first by trivially extending the phase space $\Gamma(q,\pi) \to \Gamma_\text{e}(q, \pi, \phi, \pi_\phi)$ to include the canonically conjugate variables $\phi_I$ and $\pi_\phi^I$, then by applying the canonical transformation $T$
\begin{align}
    q^a_I &\to T q^a_I \equiv \bar q^a_I = e^{\hat\phi_I} q^a_I \\
    \pi^I_a &\to T \pi^I_a \equiv \bar \pi^I_a = e^{-\hat\phi_I} \pi_a^I \\
    \phi_I & \to T \phi_I \equiv \bar \phi_I = \phi_I \\
    \pi^I_\phi & \to T \pi^I_\phi \equiv \bar \pi^I_\phi = \pi^I_\phi - \lf[ q^a_I \pi^I_a - \mean{q\cdot\pi} \rt]
\end{align}
generated by
\begin{equation}\label{eq: can trans q}
    F = \sum_I \lf( q^a_I \exp(\hat\phi_I) \bar\pi^I_a + \phi_I \bar \pi^I_\phi \rt)
\end{equation}
As before, the non--zero Poisson bracket's introduced in $\Gamma_\text{e}$ are
\begin{equation}
    \pb{\phi_I}{\pi_\phi^J} = \delta^J_I.
\end{equation}
We note that the introduction of $(\phi, \pi^\phi)$ is analogous to the restoration of $U(1)$--gauge symmetry in massive electrodynamics through the introduction of a St\"uckelberg field $\phi$: the original photon field $A_\mu$ is replaced by a $U(1)$--transformed $A_\mu+\partial_\mu \phi$, where $\phi$ is not viewed as a gauge parameter but as a new field with its own gauge transformation property, such that the transformed photon mass term $\frac{m^2}2 (A^\mu+\partial^\mu \phi) (A_\mu+\partial_\mu \phi)$ is gauge-invariant. Because of this close analogy, we will often refer to $\phi$ as a St\"uckelberg field in the remaining text.

The trivial constraint $\pi_\phi^I \approx 0$ in the original extended theory transforms to
\begin{equation}
    T\pi_\phi^I \equiv C^I = \pi_\phi^I - \lf( \pi^I_a q^a_I - \mean{\pi\cdot q} \rt)\approx 0.
\end{equation}
Because $T$ is a canonical transformation, the $C^I$'s are trivially first class with respect to the transformed constraints $T\chi_I$
\begin{equation}
    \chi_I \to T\chi_I = \frac 1 2 \lf( e^{-2\hat\phi} \pi_I^2 - e^{2\hat\phi} \absk q_I^2 \rt).
\end{equation}
In fact, since $\pi_\phi^I \approx 0$ commuted with any functions $f(q,p)$ in the original theory, $C^I\approx 0$ commutes with \emph{any} functions, $f(\bar q,\bar p)$, of the best--matched quantities $\bar q$ and $\bar p$. This should be expected: the canonical transformation (\ref{eq: can trans q}) did not change the theory. But, by extending the phase space, we have added two phase space degrees of freedom per particle. The role of the $C_I$'s is precisely to remove these degrees of freedom.  This can only happen if the $C_I$'s are first class. Consequently, the infinitesimal gauge transformations generated by the $C_I$'s on the extended phase space variables
\begin{align}
    \delta_{\theta_J} q_I &= -\delta^{IJ}\hat\theta_I q_I & \delta_{\theta_J} \pi_I &= \delta^{IJ}\hat\theta_I\pi_I \nonumber\\
   \label{eq: C_I extended trans} \delta_{\theta_J} \phi_I &= \delta^{IJ}\theta_J & \delta_{\theta_J} \pi^I_\phi &= 0,
\end{align}
where $\delta_{\theta_J} \cdot = \pb{\cdot}{ \theta^J C_J}$ (no summation), represents at this stage a completely artificial symmetry.

Before writing down the remaining constraints, we note that one of the $C_I$'s is singled out by our construction. Averaging all the $C_I$'s gives
\begin{equation}\label{eq: pi identity}
    \mean{\pi_\phi} \approx 0.
\end{equation}

However, the action of this constraint is trivial for arbitrary functions $f$ on the extended phase space $\Gamma_\text{e}$. To see this, note that the gauge transformations $\delta_\theta \cdot = \pb{\cdot}{\theta\mean{\pi_\phi}}$ generated by (\ref{eq: pi identity}) are trivial for the phase space variables $q$, $\pi$, $\hat\phi$, and $\pi_\phi$. The only transformation that is not automatically zero is $\delta_\theta \hat\phi_I = \theta \pb{\hat\phi_I}{\mean{\pi_\phi}}$, which can easily be seen to vanish. Since the $\phi$'s only enter the theory through the $\hat\phi$'s, we see that the constraints (\ref{eq: pi identity}) act trivially on all quantities in our theory. This is a reflection of the fact that there is a redundancy in the $\hat\phi$ variables. Eq~(\ref{eq: pi identity}) is then an identity and not an additional constraint. This reflects the redundancy in the original $\hat\phi$ variables expressed in terms of $\phi$'s.

It will be convenient to separate this identity from the remaining constraints by introducing the numbers $\alpha^I_i$, where $i = 1\dots (n-1)$, such that the set $\lf\{\mean{\pi_\phi}, C_i \rt\}$, where $C_i = \sum_I \alpha^I_i C_I$, forms a linearly independent basis for the $C_I$'s.\footnote{Note that the ability to explicitly construct this basis is a feature of the finite dimensional model. This possibility sets the toy model apart from the situation in geometrodynamics, which is more subtle.} We are then safe to drop the constraint $\mean{\pi_\phi} = 0$ from the theory. The remaining constraints are
\begin{align}
    T\chi_I &\approx 0 , \text{   and} & C_i \approx 0.
\end{align}
where $I = 1,\dots,n$ and $i = 1,\dots,n-1$.

\subsection{Impose and propagate best--matching constraints}\label{sec: Mach conditions}

In order to complete the best--matching procedure, we impose the \emph{best--matching} constraints
\begin{equation}
    \pi_\phi^I \approx 0.
\end{equation}
One of the best--matching constraints, $\mean{\pi_\phi} \approx 0$, is first class but trivially satisfied. It can be dropped. The remaining $\pi_\phi^I$'s can be represented in a particular basis as $\pi_\phi^i = \alpha^i_I \pi_\phi^I$ and are second class with respect to the $T\chi$'s. To see this, consider the smearing functions $f^I$. Then, the Poisson bracket
\begin{equation}
    \sum_{I,J=1}^{n} \pb{f^IT\chi_I}{ \alpha^j_J\pi^J_\phi} = \sum_{J=1}^n \lf[ \mean{\alpha^j_J f^J \lf(  \pi^2_J + \absk q^2_J \rt)} - \alpha^j_J f^J \lf( \pi^2_J + \absk q^2_J \rt) \rt]
\end{equation}
is non--zero for general $f^I$ (in the above, the mean is taken with respect to the $j$ components). Nevertheless, the constraints $\pi_\phi^i$ can be propagated by the Hamiltonian by using the preferred smearing $f^I = N^I_0$ that solves the $n-1$ equations
\begin{equation} \label{eq: lf particle}
    \sum_{J=1}^n \alpha^j_J N_0^J \lf( \absk q^2_J + \pi^2_J \rt) = \mean{\alpha^j_J N_0^J \lf( \absk q^2_J + \pi^2_J \rt)}.
\end{equation}
These lapse fixing equations lead to consistent equations of motion.

We can now understand the reason for imposing the identity $\mean{\hat\phi} = 0$. Were $\hat\phi$ unrestricted, the $n\times n$ matrix $\pb{T\chi_I}{\pi_\phi^J}$ would have a trivial kernel $N^I_0 = 0$. This uninteresting solution leads to frozen dynamics. Alternatively, the $(n-1)\times n$ matrix $\pb{T\chi_I}{\pi_\phi^i}$ has a one dimensional kernel. Thus, (\ref{eq: lf particle}) has a one parameter family of solutions parameterized by the global lapse. This leaves us, in the end, with a non--trivial time evolution for the system.

The fact that we have a remaining global lapse indicates that we have a remaining linear combination of $T\chi_I$'s that is first class. We can split the $T\chi_I$'s into a single first class constraint
\begin{equation}
    T\chi_\text{f.c.} = \sum_I N_0^I T\chi_I
\end{equation}
and the second class constraints
\begin{equation}
    T\chi_i = \sum_I \beta_i^I T\chi_I,
\end{equation}
where the $\beta_i^I$ are numbers chosen so that the above set of constraints is linearly independent. Note that, in the finite dimensional case, it is easy to ensure that this new basis is equivalent to the original set of $\chi_I$'s. In geometrodynamics, the situation is more subtle because we are dealing with continuous degrees of freedom. Using this splitting, we now have the first class constraints
\begin{align}\label{eq 1st class part}
    T\chi_\text{f.c.} &\approx 0, \text{   and} & C_i &\approx 0
\end{align}
as well as the second class constraints
\begin{align}\label{eq 2nd class part}
    T\chi_i &\approx 0, \text{   and} & \pi^i_\phi &\approx 0.
\end{align}

\subsection{Eliminate second class constraints}

To see that this procedure has succeeded in exchanging symmetries, we can eliminate the second class constraints by defining the Dirac bracket
\begin{equation}
    \pb{\cdot}{\cdot\cdot}_\text{Db} = \pb{\cdot}{\cdot\cdot}_\text{Pb} + \sum_{i,j = 1}^{n-1} \pb{\cdot}{\pi_\phi^i} C_i^j \pb{T\chi_j}{\cdot\cdot} - \pb{\cdot}{T\chi_i} C^i_j \pb{\pi_\phi^j}{\cdot\cdot},
\end{equation}
where
\begin{equation}
     C^i_j = \pb{\pi_\phi^i}{T\chi_j}^{-1}.
\end{equation}
The existence and uniqueness of the Dirac bracket depend on the invertibility of $C$, which, in turn, depends on the existence and uniqueness of the lapse fixing equations (\ref{eq: lf particle}). In this toy model, (\ref{eq: lf particle}) are just arbitrary linear algebraic equations. Thus, $C$ will generally have a non--vanishing determinant.

The advantage of using the Dirac bracket is that the second class constraints $\pi_\phi^i$ and $T\chi_i$ become first class and can be applied \emph{strongly} and eliminated from the theory. This procedure is greatly simplified because one of our second class constraints, namely $\pi_\phi^i = 0$, is proportional to a phase space variable. Thus, the method discussed by Dirac in \cite{Dirac:CMC_fixing} can be applied to eliminate the second class constraints. For $\pi_\phi^i$, this involves setting $\pi_\phi^i = 0$ everywhere in the action. For the $T\chi_i$'s, this involves treating $T\chi_i$ as an equation for $\phi$. The solution, $\phi_0(q, \pi)$, of
\begin{equation}\label{eq phi_0 particle}
    \lf. T\chi_i\rt|_{\phi^I = \phi^I_0} =0
\end{equation}
is then inserted back into the Hamiltonian. This leaves us with the dual Hamiltonian
\begin{equation}
    H_\text{dual} = N T\chi_\text{f.c.}(\phi_0) + \sum_{i=1}^{n-1} \Lambda^i \sum_{I=1}^n \alpha^i_I\lf( q^a_I \pi^I_a - \mean{q \cdot \pi} \rt),
\end{equation}
where $N$ and $\Lambda^i$ are Lagrange multipliers. We can write this in a more convenient form by using the redundant constraints $D_I = q^a_I \pi^I_a - \mean{q \cdot \pi}$
\begin{equation} \label{eq dual ham particle}
    H_\text{dual} = N T\chi_\text{f.c.}(\phi_0) + \sum_{I=1}^{n} \Lambda^I D_I
\end{equation}
remembering that one of the $D_I$'s is trivially satisfied.

The remaining Hamiltonian is now dependent only on functions of the original phase space $\Gamma$. Furthermore, the Dirac bracket between the dual Hamiltonian and any functions $f(q,\pi)$ on $\Gamma$ reduces to the standard Poisson bracket. This can be seen by noting that the extra terms in the Dirac bracket $\pb{f(q,\pi)}{H_\text{dual}}_\text{Db}$ only contain a piece proportional to $\pb{f(q,\pi)}{\pi_\phi}$, which is zero, and a piece proportional to, $\pb{\pi_\phi^i}{H_\text{dual}}$, which is weakly zero because $H_\text{dual}$ is first class. Thus, this standard procedure eliminates $\phi$ and $\pi_\phi$ from the theory.

The theory defined by the dual Hamiltonian (\ref{eq dual ham particle}) has the required symmetries: it is invariant under global reparametrizations generated by $\lf. \chi_\text{f.c.}\rt|_{\phi_0}$ and it is invariant under scale transformations that preserve the moment of inertia of the system. This last invariance can be seen by noting that the first class constraints $D_I$ generate exactly the symmetry (\ref{eq: q symmetry}).

Before leaving this section we will give two explicit examples of how one can construct $T\chi_\text{f.c.}(\phi_0)$ explicitly. In the first example, we solve the second class $T\chi_i$'s for $\phi_0$ exactly and then plug the solution into the definition of $T\chi_\text{f.c.}$. This immediately gives the general solution for $n$ particles. Unfortunately, this procedure cannot be implemented explicitly in geometrodynamics because the analogue of $T\chi_i(\phi_0) = 0$ is a differential equation for $\phi_0$. There is, however, a slightly simpler procedure that one can follow for constructing $T\chi_\text{f.c.}(\phi_0)$ provided one can find initial data that satisfy the initial value constraints $\chi_i$ of the original theory. If such initial data can be found, then $\phi_0 = 0$ is a solution for these initial conditions. Because the theory is consistent, the initial value constraints will be propagated by the equations of motion. Thus, $\phi_0 = 0$ for all time. We can then use this special solution and compute $\chi_\text{f.c.}$ using its definition. This alternative method is given in the second example for the $n=3$ case and agrees with the more general approach.

\subsubsection{General $n$}

For simplicity, we take $\mathcal E_I = 0$. This will lead to transparent equations, which we value more in the toy model than physical relevance. In the particle model, it is possible to find $\phi^I_0$ explicitly for an arbitrary number of particles, $n$, by solving $T\chi_I(\phi_0) = 0$. It is easiest to solve for $\hat\phi^I_0$ directly. The solution is\footnote{Negative spring constants $k<0$ are required to give real solutions to this equation.}
\begin{equation}\label{eq phi_0}
    e^{2\hat\phi^I_0} = \sqrt{\frac{\pi^2_I}{\absk q^2_I}}.
\end{equation}
We can choose a simple basis for the second class $\chi_I$'s: $\{ \chi_i| i \neq j\}$, for some arbitrary $j$. The first class Hamiltonian is then
\begin{equation}
    \chi_\text{f.c.} = N^j_0 \chi_j(\hat\phi^I_0).
\end{equation}
Using the identity $\hat\phi^j = - \sum_{I\neq j} \hat\phi^I$, we can rewrite this as
\begin{equation}
    \chi_\text{f.c.} \propto \exp\lf( 2\sum_{I\neq j} \hat\phi^I_0 \rt) \pi^2_j + \exp\lf( -2\sum_{I\neq j} \hat\phi^I_0 \rt) \absk q^2_j.
\end{equation}
Inserting (\ref{eq phi_0}), the first class $\chi$ becomes
\begin{equation}\label{eq chi_fc n particle}
    \chi_\text{f.c.} \propto \prod_{I=1}^n \pi^2_I  - \absk^{n} \prod_{I=1}^n q^2_I.
\end{equation}
It can be readily checked that (\ref{eq chi_fc n particle}) is invariant under the symmetry
\begin{align}
    q_I &\to e^{\hat\phi_I} q_I & \pi_I &\to e^{-\hat\phi_I} \pi_I,
\end{align}
where, $\hat\phi^I = \phi^I - \mean{\phi}$.

In geometrodynamics, solving for $\hat\phi$ explicitly will not be possible because it will require the inversion of a partial differential equation. However, if one is only interested in mapping solutions from one side of the duality to the other, then it is sufficient to find initial data that satisfy the $\chi_I$'s with $\phi = 0$ and use the preferred lapse to construct the global Hamiltonian. This procedure is outlined in the next example.

\subsubsection{Example: $n = 3$}

We analyse the 3 particle case because this involves subtleties that do not arise in the 2 particle case but must be dealt with in the general $n$ particle case.

If we choose the basis $\{\pi^i_\phi\} = \{\pi^1_\phi, \pi^2_\phi\}$, then the lapse fixing equations (\ref{eq: lf particle}) with $\phi = 0$ have the unique solution
\begin{align}
    N^1_0 &= N^3_0 \lf( \frac{\absk q^2_3 + \pi^2_3}{\absk q^2_1 + \pi^2_1} \rt) &  N^2_0 &= N^3_0 \lf( \frac{\absk q^2_3 + \pi^2_3}{\absk q^2_2 + \pi^2_2} \rt).
\end{align}
Inserting this solution into $\lf. \chi_\text{f.c.}\rt|_{\phi = 0} =\sum_I N^I_0 \chi_I$, we find
\begin{equation}
    \chi_\text{f.c} \propto \lf( \pi^2_1 \pi^2_2 \pi^2_3 - \absk^3 q^2_1 q^2_2 q^2_3  \rt) - 2\chi_1 \chi_2 \chi_3.
\end{equation}
The consequence of setting $\phi^I = 0$ instead of using it to solve $T\chi_I(\phi_0^I) = 0$ is that we pick up a term that is proportional to $\chi_1$ and $\chi_2$, which are \emph{strongly} zero. For this reason, we need to ensure that we have initial data that solve the initial value constraints ($\chi_1 = 0$ and $\chi_2 = 0$ in this case) for $\phi = 0$ in order to arrive at the correct global Hamiltonian. Assuming that such data have been found, the last term is zero and $\chi_\text{f.c.}$ reduces to
\begin{equation}\label{eq chi_fc 3 particle}
    \chi_\text{f.c} \propto \lf( \pi^2_1 \pi^2_2 \pi^2_3 - \absk^3 q^2_1 q^2_2 q^2_3 \rt).
\end{equation}
This agrees with our general result from last section.

\subsection{Construct explicit dictionary}

We can now compare our starting Hamiltonian with that of the dual theory:
\begin{align}
    H_0 &= \sum_I N^I \chi_I \\
    H_\text{dual} &= N\lf. \chi_\text{f.c.}\rt|_{\phi_0} + \sum_I \Lambda^I \lf( \pi^I_a q^a_I - \mean{ \pi_I \cdot q_I} \rt) \notag \\
                  &= N\lf[ \prod_{I=1}^n \pi^2_I  - \absk^{n} \prod_{I=1}^n q^2_I \rt] + \sum_I \Lambda^I \lf( \pi^I_a q^a_I - \mean{ \pi_I \cdot q_I} \rt).
\end{align}
The first theory is locally reparametrization invariant while the second has local scale invariance and is only globally reparametrization invariant. Note the highly non--local nature of the global Hamiltonian of the dual theory.

The equations of motion of the original theory are
\begin{align}
    \dot q^a_I &= N^I \pi^I_b \delta^{ab} & \dot \pi^I_a &= N^I \absk q^b_I \delta_{ab}.
\end{align}
Those of the dual theory are
\begin{align}
    \dot q^a_I &= \lf[2N \prod_{J\neq I} \pi^2_J \rt] \pi^I_b \delta^{ab} + \frac{n-1}{n}\Lambda_I q^a_I \\
    \dot \pi^I_a &= \lf[2N \prod_{J\neq I} \absk q^2_J \rt] \absk q^b_I \delta_{ab} + \frac{1-n}{n}\Lambda_I \pi^I_a.
\end{align}
We can now read off the choice of Lagrange multipliers for which the equations of motion are equivalent
\begin{align}\label{eq N dic}
    N^I &= 2N \prod_{J\neq I} \pi^2_J \approx 2N \prod_{J\neq I} \absk q^2_J \\
    \Lambda^I &= 0.
\end{align}
In the second equality of (\ref{eq N dic}), the $\chi_I$'s have been used.

\section{General symmetry trading algorithm}\label{sec:trading algorithm}

Based on the above description of the procedure, we can sketch out a general algorithm for trading the first class symmetry $\chi$ for $D$:
\begin{enumerate}
    \item Extend the phase space of the theory to include the Stueckelberg field $\phi$ and its conjugate momentum $\pi_\phi$ then artificially introduce the symmetry to be gained in the exchange by performing a canonical transformation of the schematic form:
    \begin{equation}
	F(\phi) = q \exp(\phi) \bar \pi + \phi\bar \pi_\phi.
    \end{equation}
	The trivial constraint $\pi_\phi \approx 0$ transforms to
    \begin{equation}
	    T \pi_\phi \equiv C = D - \pi_\phi\approx 0.
    \end{equation}
    \item Perform a best--matching variation by imposing the conditions
	\begin{equation}
	    \pi_\phi \approx 0.
	\end{equation}
    If they are first class, then the procedure just ``gauges'' a global symmetry. If they are second class with respect to the transformed first class constraints, $T\chi$, of the original theory, then they can be exploited for the exchange if the Lagrange multipliers, $N^0$, can be fixed uniquely such that
    \begin{equation}
	\pb{T\chi(N^0)}{\pi_\phi} = 0.
    \end{equation}
    When this is possible, the $\pi_\phi$'s can be treated as special gauge fixing conditions for $\chi$. It is possible, as in our case, that there will be a part of $T\chi$ that is still first class with respect to $\pi_\phi$. This should be separated from the purely second class part.
    \item Define the Dirac bracket to eliminate the second class constraints. This can be done provided the operator $\pb{T\chi}{\pi_\phi}$ has an inverse (which also means that $N^0$ exists and is unique). The second class constraints can be solved by setting $\pi_\phi = 0$ everywhere in the Hamiltonian and by setting $\phi = \phi_0$ such that $T\chi(\phi_0) = 0$. The second class condition and the implicit function theorem guarantee that this can be done. When this is done, the Dirac bracket reduces to the Poisson bracket on the remaining phase space variables. The $\pi_\phi$ terms drop out of the $C$'s and the $\chi$'s have been traded for $D$'s.
    \item To check consistency, construct the explicit dictionary by reading off the gauges in both theories that lead to equivalent equations of motion.
\end{enumerate}

Guided by this basic algorithm, we will now construct a formal geometric picture to illustrate why this algorithm works. 

\subsection{Geometric picture}

Let us now examine the mathematical structure behind the trading of gauge symmetries that we encountered by constructing the dual theory. We start out with a gauge theory $(\Gamma,\{.,.\},H,\{\chi_i\}_{i\in\mathcal I})$, where $\Gamma$ is the phase space of the theory supporting the Poisson-bracket $\{.,.\}$, the Hamiltonian $H$, and the constraints $\chi_i$. We demand that the constraints are first class, i.e. for all $i,j\in \mathcal I$ there exist phase space functions $f_{ij}^k$ as well as $u_i^k$ s.t.
\begin{equation}\label{equ:first-class-constraints}
 \begin{array}{rcl}
   \{\chi_i,\chi_j\}&=&\sum_{k\in\mathcal I}f^k_{ij} \chi_k\\
   \{H,\chi_i\}&=&\sum_{k\in\mathcal I} u^k_i \chi_k.
 \end{array}
\end{equation}
The constraints define a subspace\footnote{Although we assert ``space'' we abstain from topologizing $\Gamma$ or any of its subsets in this thesis.} $\mathcal C = \{x\in \Gamma: \chi_i(x)=0\,\, \forall i\in\mathcal I\}$ and the first line of equation \eqref{equ:first-class-constraints} implies that the derivations $\{\chi_i,.\}$ are tangent to $\mathcal C$, while the second line implies that the Hamilton vector field $\{H,.\}$ is tangent to $\mathcal C$. The $\{\chi_i,.\}$ generate a group $\mathbb G$ of gauge-transformations on $\mathcal C$ that lets us identify $\mathcal C$ through an isomorphism $i:E\to\mathcal C$ with a bundle\footnote{This bundle is, in general, not a fibre bundle, since different points $x$ in $\mathcal C$ generally have different isotropy groups $Iso(x)$.} $E$ over $\mathcal C/\mathbb G$. The fibres of $E$ are gauge orbits of a point $x\in \mathcal C$ and thus isomorphic to $\mathbb G/Iso(x)$. According to Dirac, we identify the fibres with physical states and observe that the total Hamiltonian $H_{tot}=H+\sum_{i\in\mathcal I}\lambda^i \chi_i$ depends on a set of undetermined Lagrange multipliers $\{\lambda_i\}_{i\in\mathcal I}$. Fixing a gauge means to find a section $\sigma$ in $E$.

This is most easily done by imposing a set of gauge-fixing conditions\footnote{In general the index set for $\phi_i$ could be different from $\mathcal I$.} $\{\phi_i\}_{i \in \mathcal I}$ such that the intersection of $\mathcal G=\{x\in\Gamma:\phi_i(x)=0\,\,\forall i\in \mathcal I\}$ with $\mathcal C$ coincides with $i(\sigma)$. To preserve $i(\sigma)$ under time evolution we have to solve $\left.\{H_{tot},\phi_i\}\right|_{\mathcal C}=0\,\,\forall i\in\mathcal I$ for the Lagrange-multipliers $\lambda^j=\lambda^j_o$. The gauge-fixed Hamiltonian is
\begin{equation}H_{gf}:=H+\sum_{i\in\mathcal I}\lambda_o^i \chi_j\end{equation} and the equations of motion are generated by $\{H_{gf},.\}$, while the initial value problem is to find data on $i(\sigma)$.\medskip

Given a gauge theory $(H,\{\chi_i\}_{i \in \mathcal I})$ on a phase space $(\Gamma,\{.,.\})$ we define a {\bf dual} gauge theory as a gauge theory $(H_d,\{\rho_j\}_{j\in \mathcal J})$ on $(\Gamma,\{.,.\})$ if and only if there exists a gauge fixing in the two theories such that the initial value problem and the equations of motion of both theories are identical.

One particular way to construct a dual theory is to use gauge-fixing conditions $\{\phi_i\}_{i\in\mathcal I}$ defining $\mathcal G$ such that for all $i,j\in\mathcal I$ there exist phase space functions $g^k_{ij}$ satisfying the integrability condition
\begin{equation}
 \{\phi_i,\phi_j\}=\sum_{k\in \mathcal I} g^k_{ij} \phi_k.
\end{equation}
The $\{\phi_i,.\}$ thus generate a group $\mathbb H$ of transformations on $\mathcal G$ that lets us identify $\mathcal G$ with a bundle $F$ over $\mathcal G/\mathbb H$ using an isomorphism $j:F\to \mathcal G$. The fibres of $F$ at $x\in \mathcal G$ are $\mathbb H/Iso(x)$. The Hamiltonian $H_d$ of the dual theory has to satisfy
\begin{equation}\label{equ:H-dual-conditions}
 \begin{array}{rcl}
   \left.\{H_d,.\}\right|_{\sigma}&=&\left.\{H_{gf},.\}\right|_{\sigma}\\
   \left.\{H_d,\phi_j\}\right|_{\mathcal G}&=&0.
 \end{array}
\end{equation}
These conditions as well as the integrability condition can be fulfilled by construction using the canonical St\"uckelberg formalism to implement symmetry under an Abelian group $\mathbb H$ with the subsequent elimination of the St\"uckelberg field by substituting it with the solution to the constraint equations in the original theory. The nontrivial conditions for the St\"uckelberg formalism to yield the anticipated trading of gauge symmetries are: 1) that the generators of the group furnish a gauge--fixing for the original gauge symmetry and 2) that the transformed constraint equations admit a solution in terms of the St\"uckelberg field. In the particle model, the first condition is satisfied by the invertibility of (\ref{eq: lf particle}) in terms of $N_0$ while the second condition is guaranteed by the invertibility of $T\chi_i(\phi_0) = 0$ in terms of $\phi_0$.

\section{Non--equivariance and \emph{linking} theories}\label{sec:linking toy}

The symmetry trading algorithm just presented can be viewed in a more powerful way using the concept of a \emph{linking} gauge theory. We will adopt the general idea introduced in \cite{Gomes:linking_paper} to the non--equivariant toy model presented in this chapter.

Consider the totally constrained Hamiltonian gauge theory on extended phase space $\Gamma_\text{e}(q,p, \phi,\pi_\phi)$ with the symplectic structure given in Section~(\ref{sec:toy can bm}) and the first class constraints
\begin{align}
    \chi_I &\approx 0  &  \pi^I_\phi &\approx 0.
\end{align}
Since the $\chi_I$'s do not depend on the $\phi$'s, one can straightforwardly apply the phase space reduction $\Gamma_\text{e} \to \Gamma$ by setting $\pi^I_\phi = 0$ strongly. This shows explicitly the trivial statement that the extended theory is equivalent to the original one. For reasons that will become apparent in a moment, we call the extended theory a \emph{linking theory} when the metric is non--equivariant with respect to the symmetry parametrized by the $\phi$'s.

We can perform the canonical transformation $T$, defined in Section~(\ref{sec:toy can bm}), in the linking theory. This transforms the constraints to
\begin{align}
    T \chi^I &\approx 0  &  C^I \equiv \pi^I_\phi - D^I &\approx 0,
\end{align}
where
\begin{equation}
    D^I = q^a_I p_a^I - \mean{q\cdot p}.
\end{equation}
We split $T \chi_I$ into $T\chi_\text{f.c.}$ and $T\chi_i$ and $C^I$ into $C^i$ and $\mean{C}$ (which is redundant). This gives the constraints
\begin{align}
    \chi_\text{f.c.} &\approx 0  &  \mean{C} &= 0 \\
    T \chi^i &\approx 0  &  C^i &\approx 0,
\end{align}
where $\chi_\text{f.c}$ is first class with respect to the $D^i$'s. (Note that the equals sign in the above equation is not a typo: this constraints is automatically satisfied.)

The reason for calling this a ``linking'' theory is that it provides a \emph{link} between the original theory and the dual theory through a particular choice of gauge fixings. If we choose the gauge $\hat\phi_I = 0$ (which is trivially a valid gauge choice for this theory), then the map $T$ becomes the identity and the constraint $\pi_\phi^i = \chi^i$ must be fixed strongly. Thus, the theory reduces to the original theory on $\Gamma$. Alternatively, the gauge $\pi_\phi^i = 0$ can be fixed by solving $T\chi^i = 0$ strongly for $\phi$. This leaves $\chi_\text{f.c.}\approx 0 $ and $D^i \approx 0$, which is precisely the dual theory. Thus, both the original theory and the dual theory represent different gauge fixings of the linking theory. This is a powerful way of thinking about the dualization procedure. Figure~(\ref{fig:toy linking theory}) shows how the linking theory leads to the original or dual theory through different gauge fixing.
\begin{figure}\label{fig:toy linking theory}
    \begin{center}
	\includegraphics[width=0.9\textwidth]{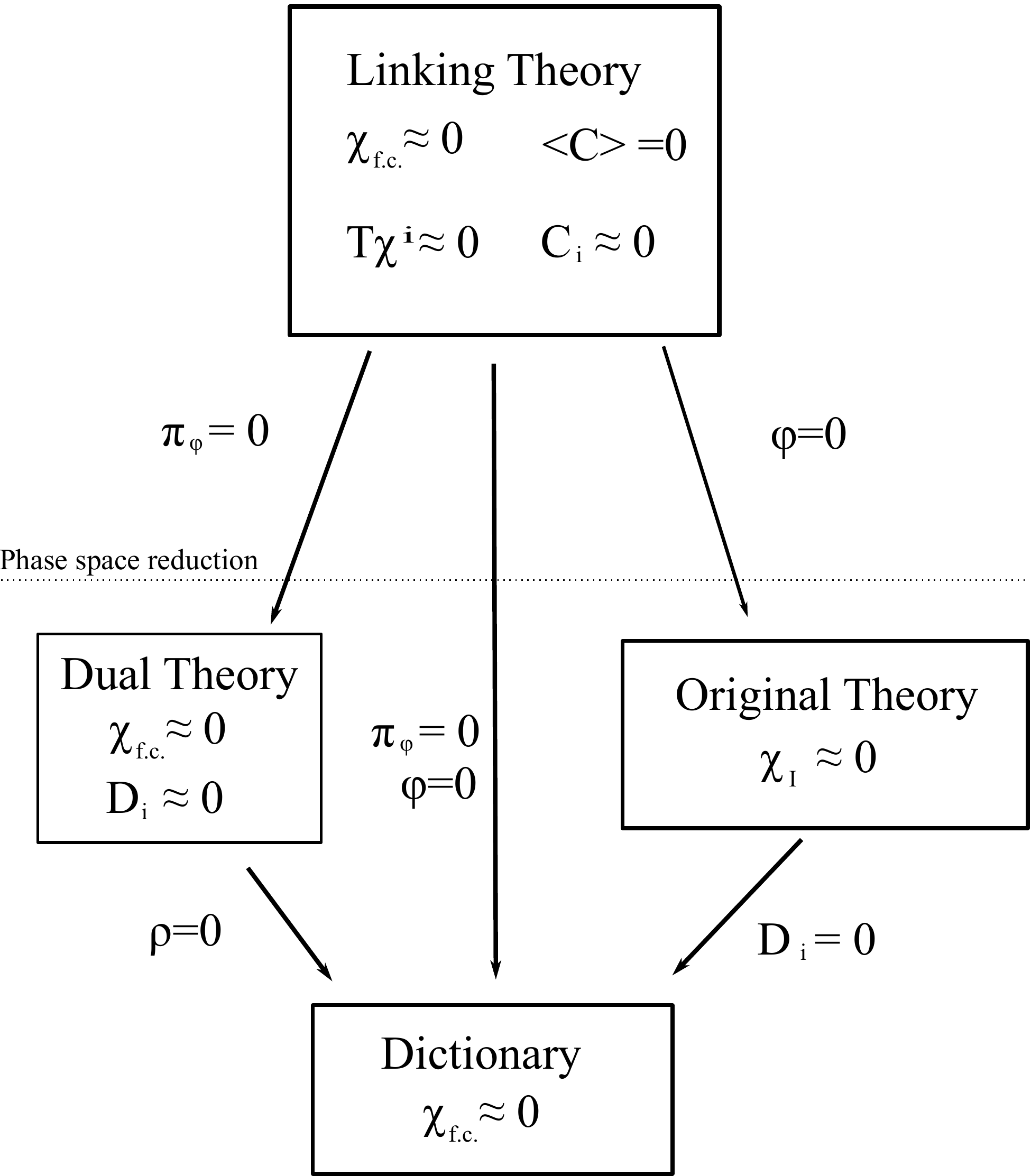}
    \end{center}
    \caption{Different gauge fixings of the \emph{linking theory} distinguish between the original and dual theories.}
\end{figure}

%-------------- BEST MATCHING IN GEOMETRODYNAMICS ---------------------
%----------------------------------------------------------------------
% BEST MATCHING IN GEOMETRODYNAMICS
%----------------------------------------------------------------------
% This section explains best matching as it applies to geometrodynamics.
% The idea is to start with ADM then show how unimodular gravity can be
% obtained and finally finish with the derivation of SD.
%
% Jun 2: This is VERY rough... the order will change considerably!
%
%======================================================================
\chapter{Best matching in geometrodynamics}\label{chap:bm_geo}
%======================================================================

In this chapter, we study different relational geometrodynamic theories. By ``geometrodynamic,'' we mean theories that involve a dynamical Euclidean metric on an arbitrary spatial manifold $\Sigma$. We will be mainly dealing with 3--geometry but the generalization to $n$ dimensions and Lorentzian metrics is either trivial or straightforward. By ``relational,'' we mean that the theories satisfy Mach's principles as they have been outlined throughout the text. This means that we will be using best matching to eliminate what we consider to be unphysical information contained in the metric. Accordingly, the first step is to identify what we believe to be the empirically meaningful information contained in the metric. We will argue that this is nothing but the information about the \emph{shape} of local configurations. Mathematically, this corresponds to best matching the 3--metric with respect to 3--dimensional diffeomorphisms and conformal Weyl transformations. The theory we are led to in this fashion is shape dynamics.

After identifying the configuration space to work with and the symmetry group to best match, the only remaining ambiguity is the choice of the metric on $\acs$. For shape dynamics, this choice is highly non--trivial. We do not yet have a simple understanding of the choice of metric that leads to shape dynamics. Instead, we motivate our choice by picking a metric, equivariant with respect to the conformal diffeomorphisms, such that a geodesic principle in terms of this metric reproduces the predictions of GR. It is remarkable that such a metric exists but understanding this choice better could be the key to unlocking the relationship between Mach's principles and the causal structure of spacetime.

The reason for this, as we will see, is that GR, in ADM form, can be derived from a version of best matching with respect to the 3--diffeomorphisms. This version of best matching is \emph{not} a true geodesic principle but a sort of local version of one. However, if we best match this ``local'' geodesic theory with respect to the conformal transformations, we get a non--equivariant theory whose dual is shape dynamics, which is a \emph{genuine} geodesic principle on conformal superspace. To achieve equivalence with GR, however, we \emph{must} abandon locality. This is a serious price to pay for a geodesic principle but the gain is a set of simple linear constraints and a clear conceptual picture motivated by Mach's principles.

To motivate the transition to shape dynamics we will start by discussing different ways of constructing relational geometrodynamic theories. First, we will sketch a na\"ive choice that maintains a certain degree of locality \emph{and} a genuine geodesic principle. Among theories of this kind is Ho\v rava's projectable lapse theory at high energy \cite{Horava:lif_point}. As a next step, we will abandon a genuine geodesic principle in favor a more restrictive one that is manifestly local. Using this, we will be able to derive GR in ADM form following \cite{barbourbertotti:mach}. We will then use the framework to briefly address the problem of time by using our definition of background independence to insert a background time into GR. The theory we obtain is equivalent to unimodular gravity. Finally, we will insist on a genuine geodesic principle by applying our symmetry trading algorithm to the ADM form of GR. This will lead to shape dynamics.

\section{Identification of the symmetry group}

As is always the case, the first and most subjective step in best matching is the identification of the configuration variables and the symmetry group. We will make a choice that is strongly inspired by Mach's ideas. Consider a complicated system of particles existing in the universe. We claim that only the \emph{local} shapes produced by the distribution of these particles is measurable and that these local shapes are determined purely from the conformal geometry.

The idea is simple: if we can go locally to a frame that is approximately flat then we should be able to best--match the shapes of any particles that live in this locally flat neighborhood. Thus, all we must do is shift our thinking about best matching from a global shifting procedure to a local one. But how should we think of the $\phi$ fields in this case? For the translations, we are looking for a smooth way to shift the locally flat references frames (or the frame fields) by a certain amount in an arbitrary direction. This is simply an arbitrary diffeomorphism of the metric. For rotations, this is the requirement that the frames fields be rotated by arbitrary amounts. However, the metric is invariant under rotations of the frame fields, so this has no effect on the metric itself. In other words, applying the best matching procedure to the Euclidean group \emph{locally} is equivalent to shifting the spatial metric by an arbitrary diffeomorphism. This is a way of stating that there is coordinate ambiguity when comparing a system of particles from one moment to the next. Thus, the local shapes should depend on the geometry and not the coordinate information contained in the metric. More specifically, this translates into the requirement that the theory should be best matched under the diffeomorphism group whose local algebra is generated by the Lie derivative of the metric and parameterized by the vector fields $\xi$ such that
\begin{equation}
    \delta g_{ab} = \mathcal L_\xi g_{ab}.
\end{equation}

Lastly, the local shapes should not depend on the local scale. This is because all measurements are local comparisons. If one uses a particular ruler to measure the length of an object, the result of the measurement is unchanged if the local scale is tripled in size. This is obvious because both the ruler and the object will be rescaled. But, this principle also holds independently for different local neighborhoods. If you carry the object and the ruler to a different point in space where everything is only half the size, then still nothing changes in the outcome of the measurement. This requirement translates to an invariance of the metric under local Weyl transformations of the form
\begin{equation}
    g_{ab}(x) \to e^{4\phi(x)} g_{ab}(x),
\end{equation}
where the factor of 4 is conventional in 3 dimensions.\footnote{In $d$ dimensions, the conventional factor is $\frac 4 {d-2}$. This factor ensures that the curvature scalar transforms in a simple way under this transformation.}

Putting these results together, we are looking for a theory that depends only on the conformal geometry because it is best matched with respect to the conformal diffeomorphisms. We can now identify the fiber bundle to be used in our procedure. The configuration space of all 3--metrics is called \emph{Riem}, the quotient space of Riem with the 3--diffeomorphisms is \emph{superspace}, and the quotient of superspace with the conformal transformation is \emph{conformal superspace}. Thus, we are looking for a best--matching theory where $\acs = \text{Riem}$, $\rcs = $conformal superspace, and $\cG$ is the group of conformal diffeomorphisms. Note that Riem is not a \emph{principal} fiber bundle over conformal superspace because of the presence of metrics with global isometries. One must mod out by these isometry groups to get a genuine PFB.

\section{Equivariant geodesic theories}\label{sec:global_sqrt}

In this section, we consider the simplest options for forming genuine geodesic theories where Riem is treated as a fiber bundle over conformal superspace. Unfortunately, these na\"ive choices do not lead to GR and likely suffer from instability problems. Nevertheless, it is useful to start with these simple models to see why they cannot work.

We want a genuine equivariant geodesic principle on Riem. For this, we need to specify a metric on Riem. Such a metric should be a functional of the spatial metric, $g$, and should feed on two symmetric Rank(2) tensors, $u$ and $v$. For an up to date account of how to define metrics on Riem, see \cite{Giulini:superspace}. We will only consider those metrics $\mathcal G$ that split into an ultra--local piece
\equa{\label{eq:gen_dewitt}
	G[u, v, g] \equiv \int_\Sigma d^nx \sqrt{g}\, G^{abcd}(x)u_{ab}(x) v_{ab}(x) \equiv \int_\Sigma d^nx \sqrt{g}\, (g^{ac} g^{bd} - \alpha g^{ab}g^{cd})u_{ab} v_{cd},
}
and a conformal piece $\mathcal V[g, \partial g, \hdots] = \int d^nx\sqrt{g}\, V$ such that $\mathcal{G}[u,v, g,\partial g,\hdots] = \mathcal V[g, \partial g, \hdots]\cdot G[u, v, g]$. Note that $G^{abcd}$ is the most general ultra--local Rank(4) tensor that can be formed from the metric. It represents a one parameter family of supermetrics labeled by $\alpha\in \mathbb R$. For $\alpha = 1$, we recover the usual DeWitt supermetric. The supermetric $G^{abcd}$ plays a role similar to the flat metric $\eta_{ab}$ in the toy models. The scalar function $V(g(x),\partial g(x),\hdots)$ is analogous to the conformal factor of the toy models and, for this reason, is often called the potential. However, it differs from the potential of the finite--dimensional models in that it can depend on the \emph{spatial} derivatives of the metric. 

We can perform best matching with respect to the diffeomorphisms by infinitesimally shifting the metric by
\equa{\label{eq:shift2}
	\bar{g}_{ab} = g_{ab} + \mathcal{L}_\xi g_{ab}
}
and by doing a best matching variation with respect to $\xi$. Equivariance of the action with respect to this symmetry is guaranteed by choosing it to be a scalar. Just like the toy models, the transformation \eq{shift2} is equivalent to introducing the gauge covariant derivative
\equa{
	\mathcal{D}_\xi g_{ab} = \dot{g}_{ab} + \mathcal{L}_{\dot\xi} g_{ab} = \dot{g}_{ab} + \dot\xi_{(a;b)},
}
which replaces all occurrences of $\frac{d}{d\lambda}$ in the action. In the above, semi--colons represent covariant differentiation on the tangent bundle of $\Sigma$ using the metric compatible connection.

Similarly, best matching with respect to the conformal transformations implies the additional term
\begin{equation}
    4\dot\phi g_{ab}
\end{equation}
to the covariant derivative and a best--matching variation with respect to $\phi$. Thus, the full covariant derivative becomes
\begin{equation}
    \mathcal{D}_{\xi,\phi} g_{ab} = \dot{g}_{ab} + \dot\xi_{(a;b)} + 4\dot\phi g_{ab}.
\end{equation}
Equivariance with respect to the conformal transformations requires the potential to be conformally invariant because this covariant derivative makes the kinetic term conformally invariant. In 3 dimensions, the lowest dimensional operator that one can form out of the metric and its derivatives is the Cotton tensor
\begin{equation}
    C_{ab} = \nabla_c \lf( R_{da} - \frac 1 4 R g_{da} \rt) \epsilon^{cde} g_{eb},
\end{equation}
where $\epsilon^{abc}$ is the totally antisymmetric tensor in 3d. Thus, the simplest conformally invariant scalar one can form from the metric and a \emph{finite} number of its derivatives is the square of the Cotton tensor
\begin{equation}
    V = C_{ab} C^{ab} \equiv C^2,
\end{equation}
which has 6 spatial derivatives of the metric. In principle, the potential can be any conformally invariant scalar formed from the metric. However, since this is the lowest dimensional operator compatible with our symmetries, there is an anisotropic scaling of $z = 3$ between space and time. This occurs because the $C^2$ has 6 spatial derivatives compared with the 2 time derivatives in the kinetic term. Power counting with this scaling indicates that the $C^2$ term is the only relevant local operator in 3 dimensions. It would then seem like a natural choice for the potential of the bare theory. This is similar to what happens in Ho\v rava--Lifshitz gravity \cite{Horava:lif_point,Horava:gravity2}.

We can now write down a geodesic principle on Riem. A direct analogy with the toy models gives
\begin{align}
  S_{\text{global}} &= \int d\lambda \sqrt{\mathcal G[\mathcal{D}_\xi g,\mathcal{D}_\xi g, g, \partial g, \hdots]} \label{eq:global sqrt} \\
		    &= \int d\lambda \sqrt{\int_\Sigma d^n x \sqrt{g}\, G^{abcd} \mathcal{D}_{\xi,\phi} g_{ab} \mathcal{D}_{\xi,\phi} g_{cd}} \cdotp\sqrt{\int_\Sigma d^n x' \sqrt{g}\, V(g,\partial g, \hdots)}.\label{eq:global_jacobi}
\end{align}
Clearly, \eq{global_jacobi} is a non--local action as it couples all points in $\Sigma$ at a given instant. However, the potential is still required to be a \emph{local} functional of the metric. This is what singles out the $C^2$ term as the only relevant coupling at high energy. If we allow for non--local functionals of the metric, then this is no longer the case and many other potentials are allowed. Shape dynamics is a theory of this kind. It has an equivariant potential but one that is a non--local functional of the metric. This allows it to be conformally invariant without introducing anisotropic scaling or relying on the $C^2$ term in the UV.

\subsection{Hamiltonian}

The momenta $(\pi^{ab}, \zeta^a, \pi_\phi)$ canonically conjugate to $(g_{ab}, \xi_a, \phi)$ obtained from a Legendre transform of the action (\ref{eq:global_jacobi}) are
\begin{align}
    \pi^{ab} &= \frac{\delta S}{\delta \dot g_{ab}} = \sqrt{\frac{\mathcal V}{\mathcal T}} \sqrt{g} G^{abcd} \dxi g_{cd} \\
    \zeta^{a} &= \frac{\delta S}{\delta \dot \xi_{a}} = -2\nabla_b \lf( \sqrt{\frac{\mathcal V}{\mathcal T}} \sqrt{g}G^{abcd} \dxi g_{cd} \rt) \\
    \pi_\phi &= \frac{\delta S}{\delta \dot \phi} = 4 g_{ab} \sqrt{\frac{\mathcal V}{\mathcal T}} \sqrt{g} G^{abcd} \dxi g_{cd} = 4 \pi,
\end{align}
where $ \mathcal T = G[\dxi g,\dxi g, g] = \int d^nx\,\sqrt{g} G^{abcd} \dxi g_{ab}\, \dxi g_{cd}$ is the kinetic term and $\pi \equiv g_{ab} \pi^{ab}$. This leads to the primary constraints
\begin{align}
    \mathcal H^a(x) &= \zeta^a(x) + 2\nabla_b \pi^{ab}(x) \approx 0 \\
    \mathcal C(x) &= \pi_\phi(x) - \pi(x) \approx 0,
\end{align}
which, combined with the best matching conditions $\zeta^a \approx 0$ and $\pi_\phi \approx 0$, are the standard ADM diffeomorphism constraint and the conformal constraint, respectively. Although these constraints are clearly local, there is another primary constraint that is an integral over all of space. This constraint is the \emph{zero mode} of the standard ADM Hamiltonian constraint
\equa{
    \mathcal H^{(0)} = \int d^nx\, \lf[ \frac{1}{\sqrt g} G_{abcd} \pi^{ab} \pi^{cd} - \sqrt{g} V \rt] \equiv \int d^nx\,\mathcal H.
}
It guarantees that the metric on Riem $\mathcal G[\dxi g,\dxi g, g, \partial g, \hdots]$ is non--negative. The total Hamiltonian is
\equa{
    H_\text{tot} = N(\lambda) \mathcal H^{(0)} + \int d^nx \lf[ N^a(\lambda,x) \mathcal H_a(\lambda, x) + \rho(\lambda,x) C(\lambda, x) \rt],
}
where the \emph{lapse} $N$, \emph{shift} $N^a$, and $\rho$ are Lagrange multipliers.

The lapse function is only $\lambda$ -- and not $x$ -- dependent. It is said to be projectable. Because of this, the theory does not obey the full Dirac--Teitelboim algebra \cite{Teitelboim:DT_algebra} and is invariant only under \emph{foliation preserving} diffeomorphisms and not the full $3+1$ diffeomorphism group. Because of this and because the simplest potential, $V = C^2$, has 6 derivatives of the metric, this na\"ive theory \emph{cannot} be GR. If one chooses the value of $\alpha$ to be the conformal value $\alpha = 1/3$, then the kinetic term is conformally invariant without using covariant derivatives. In this case, the conformal constraint $\pi = 0$ is unnecessary and the theory becomes equivalent to the high energy formulation of Ho\v rava--Lifshitz gravity. However, this theory is believed to suffer from instabilities and may not be well defined. As a result, this simple choice doesn't seem to lead to a sensible theory. Nevertheless, it is interesting to note that Ho\v rava's theory is very naturally motivated by best matching. In the next sections, we will try to obtain sensible theories, first by modifying the best matching procedure to make it more consistent with locality. This will lead us to GR in ADM form. Then, we will abandon non--locality completely and allow for potentials that are non--local functionals of the metric. This will lead us to shape dynamics.

%======================================================= ADM from best matching ================================

\section{GR from best matching}\label{sec:geo_jacobi}

In this section, we will slightly modify our geodesic principle by taking the square root of \eq{global sqrt} \emph{inside} the integral. Physically, this seems like the more natural choice because the action principle is now local. On the other hand, the mathematical structure is less appealing because we no longer have a proper metric on Riem. Also, we lose a direct analogy with the finite--dimensional models since we can no longer write the action in terms of a quantity that gives the ``distance'' between two infinitesimally separated geometries. Furthermore, as we will see, using a local square root produces a local scalar constraint that restricts one degree of freedom at every point. In order for the new theory to be consistent, the local scalar constraint must be first class with respect to the remaining constraints. This puts a severe restriction on the possible forms of the potential (this issue is studied in detail in \cite{anderson:rel_wo_rel_vec,barbour_el_al:rel_wo_rel}). With the right choice of potential, this extra gauge freedom manifests itself as foliation invariance and leads to many technical and conceptual issues, particularly in the quantization.\footnote{For a review of the difficulties associated with foliation invariance and other issues associated to time, see \cite{kuchar:time_int_qu_gr,Isham:pot_review}.} Despite these complications, this modification leads to a sensible classical theory: GR in ADM form. It is, thus, important to understand how this comes about and the role that the local square root plays in producing a consistent foliation invariant theory.

To obtain GR, it will not be necessary to best match with respect to the conformal constraints. In this approach, the conformal constraints are replaced by the local Hamiltonian constraint. Bringing the square root inside the spatial integration of \eq{global sqrt} and removing the $\phi$ fields gives
\equa{\label{eq:local_geo}
  S_{\text{local}} = \int d\lambda\, d^nx \, \sqrt{g\,G^{abcd} \mathcal{D}_\xi g_{ab} \mathcal{D}_\xi g_{cd}\cdotp V(g,\partial g)}.
    }
The quantity $VG^{abcd}$ is the infinitesimal ``distance'' between two \emph{points} of two infinitesimally separated geometries. It is a kind of \emph{pointwise} metric on Riem. Thus, there is no clean geodesic principle on the reduced configuration space.

The action \eq{local_geo} has been analyzed in detail in \cite{barbour_el_al:physical_dof,barbourbertotti:mach,barbour_el_al:scale_inv_gravity}. For the special choices $\alpha = 1$ and
\equa{\label{eq:geo potential}
  V(g,\nabla g, \nabla^2 g) = 2\Lambda - R(g,\nabla g, \nabla^2 g),
}
where $R$ is the scalar curvature of $\Sigma$ and $\Lambda$ is a constant, the constraint algebra is known to close. With these choices, \eq{local_geo} is the Baierlein--Sharp--Wheeler (BSW) action of GR with cosmological constant \cite{bsw:bsw_action} whose Hamiltonian is equivalent to that of ADM \cite{adm:adm_review}. We have, thus, recovered GR in ADM form.\footnote{Equivalence is achieved because the best matching variation of $\dot\xi$ is equivalent to the usual variation of the shift vector $N^i$. See, for example, \cite{Anderson:cyclic_ADM}.}

%============================================ Unimodular Gravity ============================================

\section{ADM and best matching} \label{sec:unimodular}

While it is useful to have a derivation of GR in BSW form from best matching, it would be convenient to dispose of the awkward square root altogether. This can be done by using the alternative action principle discussed in Section~(\ref{sec:equiv S}).\footnote{For a demonstration of this following a Routhian reduction see \cite{Anderson:cyclic_ADM}.} Since the action principle \eq{php general action} can be understood as a best matching of the reparametrization invariance, it provides a natural framework for introducing a notion of background time in GR by lifting the best--matching condition on $\tau$. Interestingly, this procedure leads directly to unimodular gravity.

To implement an action of the form \eq{php general action} in geometrodynamics, we use the kinetic term and potential outlined in \scn{geo_jacobi}. Using a local action principle and introducing the auxiliary field $\tau^0(\lambda, x)$, the analogue of (\ref{eq:php general action}) is
\equa{\label{eq:geo_ham}
	S_H = \int d\lambda\, d^n x\, \sqrt{g} \frac{1}{2} \lf[ \frac{1}{\dot\tau^0} G^{abcd} \mathcal{D}_\xi g_{ab} \mathcal{D}_\xi g_{cd} - \dot\tau^0 (2\Lambda' - R) \rt],
}
where we have used a prime to distinguish $\Lambda'$ from another $\Lambda$ that we will consider later (this is completely analogous to $E$ versus $E'$ encountered in the particle models of Section~(\ref{sec:toy BD theory})). It can be verified that using a local function, $\tau^0$, of $x$ is equivalent to taking a local square root instead of the usual geodesic principle. Similarly, a global $\tau^0$ is equivalent to a global square root. 

The ADM theory can be obtained by doing a short canonical analysis of the action \eq{geo_ham}. The momenta are:
\begin{align}
	\pi^{ab} &= \diby{L}{\dot{g}_{ab}} = \frac{\sqrt{g}}{\dot\tau^0} G^{abcd}(\dot{g}_{cd} + \dot\xi_{(c,d)}), \\
	\zeta^a &= \diby{L}{\dot\xi_a} = -\nabla_b \lf( \frac{\sqrt{g}}{\dot\tau^0} G^{(ab)cd}(\dot{g}_{cd} + \dot\xi_{(c,d)}) \rt), \qand \\
	p_0 &= \diby{L}{\dot\tau^0} = -\frac{\sqrt{g}}{2}\lf( \frac{1}{(\dot\tau^0)^2} G^{abcd} \mathcal{D}_\xi g_{ab} \mathcal{D}_\xi g_{cd} + (2\Lambda' -R) \rt),
\end{align}
where $p_0$ is the momentum density conjugate to $\tau_0$. The scalar constraint is
\equa{\label{eq:ham_constraint_ham}
	\ham = \frac{1}{\sqrt{g}} G_{abcd}\pi^{ab}\pi^{cd} + \sqrt{g}(2\Lambda'-R) +  2p_0 = \hamadm + 2p_0 = 0,
}
where $\hamadm$ is just the scalar constraint of the ADM theory. There is also a vector constraint associated with $\zeta^a$. It is
\begin{equation}
 	\diff = \nabla_b \pi^{(ab)} + \zeta^a = \diffadm + \zeta^a = 0.
\end{equation}
The constraint $\diffadm$ is ADM's usual vector constraint.

The canonical Hamiltonian is zero as it should be in a reparameterization invariant theory. Thus, the Hamiltonian is
\equa{\label{eq:ham_tot_ham}
	H = N\ham + N_a \diff = H_\text{ADM} + 2Np_0 + N_a \zeta^a.
}
$H_\text{ADM}$ is the ADM Hamiltonian. However, this may not be the full Hamiltonian since we need to check for secondary constraints. To do this, we introduce the fundamental equal--$\lambda$ PB's
\begin{align}
  \pb{g_{ab}(\lambda,x)}{\pi^{cd}(\lambda, y)} &= \delta^c_a\delta^d_b\,\delta(x,y), \\
  \pb{\xi_a((\lambda,x)}{\zeta^b(\lambda, y)} &= \delta^b_a\,\delta(x,y), \qand \\
  \pb{\tau^0(\lambda,x)}{p_0(\lambda,y)} &= \delta(x,y).
\end{align}
Then, the constraint algebra reduces to
\begin{align}
  \pb{g^{-1/2} \ham(x)}{\ham(y)} &= \lf[ (g^{-1/2}\diffadm)(x) + (g^{-1/2}\diffadm)(y) \rt] \delta(x,y)_{;a} \\
  \pb{g^{-1/2} \ham(x)}{\diffadm(y)} &= g^{-1/2} \hamadm(x)^{;a} \delta(x,y) \label{eq:ham_diff}\\
  \pb{g^{-1/2} \diff(x)}{\ham^b(y)} &= (g^{-1/2}\diffadm)(x)\delta(x,y)^{;b} + (g^{-1/2}\hamadm^b)(y)\delta(x,y)^{;a}.
\end{align}
At this point, the discussions for standard and best matching variations diverge.

\subsection{Best matching variation: time--independent theory}

After taking PB's we can apply the best matching conditions for $p_0$ and $\zeta^a$
\begin{align}
  p_0\approx & 0 \\
  \zeta^a\approx & 0.
\end{align}
Then, the vector and scalar constraints imply
\begin{align}
  \hamadm\approx & 0 \\
  \diffadm\approx & 0.
\end{align}
Thus, the constraint algebra is first class and the total Hamiltonian is given by \eq{ham_tot_ham}.

At this point, we can not use the best matching conditions to recover the ADM theory because they are only weak equations. To see that the ADM theory is indeed recovered, we work out the classical equations of motion. The terms in \eq{ham_tot_ham} that are new compared with the ADM theory commute with $g_{ab}$ and $\pi^{ab}$. Thus, they do not affect the equations of motion for $g_{ab}$ or for $\pi^{ab}$ other than replacing the lapse $N$ with $\dot\tau^0$ and the shift $N_a$ with $\dot\xi_a$. Since the remaining equations of motion just identify
\begin{align}
 	\dot\tau^0 &= \pb{\tau^0}{\htot} = 2N, \qand \\
	\dot\xi_a & = \pb{\xi_a}{\htot} = N_a,
\end{align}
the theories are classically equivalent. In other words, it is trivial to integrate out the auxiliary fields $\tau^0$ and $\xi_a$ along with their conjugate momenta. This simply disposes of the best matching conditions and replaces $\tau^0$ and $\xi_a$ with the lapse and shift. It is now easy to see that the quantum theories will also be equivalent since the quantization of the best matching conditions imply that the quantum constraints are identical to those of the ADM theory.

\subsection{Fixed endpoints: unimodular gravity}\label{sec:unimolduar}

In this section, we consider the effect of fixing the endpoints of $\tau^0$. According to the definition of background dependence from \scn{BI_BD}, this will introduce a background time. We will, however, not fix a background for the diffeomorphism invariance. Thus, we still have the best matching condition $\zeta^a \approx 0$ for the variation of $\xi_a$.

The constraint algebra is no longer first class after lifting the best matching condition $p_0\approx 0$ because the scalar constraints no longer close with the vector constraints. From (\ref{eq:ham_diff}) and \eq{ham_constraint_ham},
\equa{
  \pb{g^{-1/2} \ham(x)}{\diffadm(y)} = -(g^{-1/2} p_0)^{;a} \delta(x,y),
}
which implies the secondary constraint
\equa{\label{eq:cosmo_constraint}
  -\nabla_a \Lambda = 0,
}
where $\Lambda = -g^{-1/2} p_0$ is the undensitized momentum conjugate to $\tau^0$. The constraint algebra is now first class. Using the Lagrange multipliers $\tau^a$, the total Hamiltonian is
\equa{
  \htot = H_\text{ADM} + 2Np_0 + N_a \zeta^a - \tau^a \nabla_a \Lambda.
}

The secondary constraint \eq{cosmo_constraint} assures that $\Lambda$ is a spatial constant. Given the equations of motion $\dot\tau^0 = N$ and $\dot\Lambda = 0$, one might expect that the $\dot\tau^0\Lambda$ term in the action is analogous to adding a cosmological constant term to the potential. Indeed this is what happens. Since the action is linear in $\zeta^a$, we can integrate out $\zeta^a$ by inserting the equation of motion $\dot\xi_a = N_a$ and the best matching condition $\zeta^a = 0$. This leads to
\begin{multline}
	S_{\text{uni}} = \int d\lambda\, d^nx \, \lf[ \dot g_{ab} \pi^{ab} + \dot\tau p_0 + \sqrt{g} \tau^a\nabla_a \Lambda \rt. \\
	\lf. -2 N_a \nabla_b \pi^{ab} - N\lf( \frac{1}{\sqrt{g}} G_{abcd}\pi^{ab}\pi^{cd} - \sqrt{g} (R -  2\Lambda_\text{tot}) \rt) \rt],
\end{multline}
which is identical to the action of unimodular gravity considered by Henneaux and Teitleboim \cite{Henneaux_Teit:unimodular_grav}. Unimodular gravity was originally proposed as a possible solution to the problem of time and was developed extensively in \cite{Unruh:unimodular_grav,Unruh_Wald:unimodular,brown:gr_time}.

Note that $\Lambda_\text{tot} = \Lambda + \Lambda'$. It is the observable value of the cosmological constant. In this context, it will depend on the boundary conditions imposed on the cosmological time
\equa{
  T = \int_{\lambda_\text{in}}^{\lambda_\text{fin}} d\lambda\int_\Sigma d^nx\, \sqrt{g}\, \dot\tau^0.
}
In \cite{Smolin:unimodular_grav}, it is shown that the fact that $\Lambda_\text{tot}$ is an integration constant protects its value against renormalization arguments that predict large values of $\Lambda'$. This provides a possible solution to the cosmological constant problem.

These results show that unimodular gravity is obtained by inserting a background time into general relativity using the definition of background dependence given in Section~(\ref{sec:BI_BD}). The quantization of this theory is known to lead to a time \emph{dependent} Wheeler--DeWitt equation \cite{brown:gr_time}. This supports the claim that we have inserted a genuine background time. Although there are some hints that unimodular gravity contains unitary cosmological solutions (see \cite{Sorkin:forks_unimodular,sorkin:unmodular_cosmology}), it is clear that unimodular gravity will not be able to solve all problems of time in quantum gravity. As was pointed out by Kucha\v r in \cite{Kuchar:unimodular_grav_critique}, the background time in unimodular gravity is global whereas foliation invariance in general relativity presents several additional challenges. These complications are introduced by the local square root and, therefore, would not have analogues in the finite--dimensional models and the projectable--lapse theories. 

Interestingly, shape dynamics is a projectable--lapse theory and does not suffer from the same complications associated to foliation invariance. One could then use this principle to produce a time dependent WDW--like equation for quantum shape dynamics. However, simply inserting a background should not be thought of as a genuine solution to the problem of time because background dependent theories violate Mach's principle and should not be thought of as fundamental (unless one has other good reasons for believing in an absolute time). Instead, one should think of background dependent theories as having emerged, under special conditions, out of a fundamental background--independent theory.

% ============================= Shape Dynamics ===============================================
%
% This is from the original paper. Some work will have to go into making this flow better....
% I still need to add a section on the linking theory.

\section{Shape Dynamics}\label{sec:sd derivation}

In this section, we derive shape dynamics by considering the non--equivariant best--matching theory obtained by best matching the action \eq{local_geo} with the potential \eq{geo potential} with respect to conformal transformations. As we will see, if we restrict to conformal transformations that preserve the volume of $\Sigma$, we can use the symmetry trading algorithm presented in Section~(\ref{sec:trading algorithm}) to trade foliation invariance for volume preserving conformal invariance. The theory we obtain is shape dynamics: it has a projectable lapse and is equivariant with respect to volume preserving conformal transformations.

\subsection{Notation}

Before applying the symmetry trading procedure, we pause to reestablish notation. It will be convenient to define a more compact notation for smearing functions that will require a slightly different nomenclature for the ADM constraints. To avoid ambiguity, we will use this section to clearly state our conventions.

We start with the ADM formulation of general relativity on a compact spatial manifold $\Sigma$ without boundary (and if confusion arises we will assume the topology of $\Sigma$ to be $S^3$). The phase space $\Gamma$ is coordinatized by 3--metrics $g$, represented locally by a symmetric 2--tensor $g:x\mapsto g_{ab}(x)dx^adx^b$, and its conjugate momentum density $\pi$, represented locally by a symmetric 2--cotensor $\pi:x\mapsto \pi^{ab}(x)\partial_a\partial_b$ of density weight 1. Given a symmetric 2--cotensor density, $F$, and a symmetric 2--tensor $f$ we denote the smearing by
\begin{equation}
 F(g):=\int_\Sigma d^3 x F^{ab}(x) g_{ab}(x)\,\,\textrm{ and }\,\,\pi(f):=\int_\Sigma d^3x \pi^{ab}(x) f_{ab}(x).
\end{equation}
We will not explicitly state differentiability conditions for $(g,\pi)$ or details about the Banach space we use to model $\Gamma$, we just assume existence of suitable structures to sustain our construction. The non--vanishing Poisson bracket is
\begin{equation}
   \{F(g),\pi(f)\}=F(f):=\int_\Sigma d^3 x F^{ab}(x) f_{ab}(x),
\end{equation}
and the Hamiltonian is
\begin{equation}
 H(N,\xi)=\int_\Sigma d^3x\left(N(x)S(x)+\xi^a(x)H_a(x)\right),
\end{equation}
where the Lagrange multipliers $N$ and $\xi^a$ denote the lapse and shift respectively. The constraints are
\begin{equation}
 \begin{array}{rcl}
   S(x)&=&\pi^{ab}(x)\frac{G_{abcd}(x)}{\sqrt{|g|(x)}}\pi^{cd}(x)-\sqrt{|g|(x)}R[g](x)\\
   H_a(x)&=&-2g_{ac}(x)D_b\pi^{bc}(x),
 \end{array}
\end{equation}
where $D$ denotes the covariant derivative w.r.t. $g$, $G_{abcd}$ denotes the inverse supermetric and $R[g]$ the curvature scalar. Denoting smearings as $C(f)=\int_\Sigma d^3x C(x)f(x)$ and $\vec C(\vec v):=\int_\Sigma d^3x v^a(x)C_a(x)$, we obtain Dirac's hypersurface--deformation algebra
\begin{equation}
 \begin{array}{rcl}
   \{\vec H(\vec u),\vec H(\vec v)\}&=&\vec H([\vec u,\vec v])\\
   \{\vec H(\vec v),S(f)\}&=& S(v(f))\\
   \{S(f_1),S(f_2)\}&=&\vec H(\vec N(f_1,f_2)),
 \end{array}
\end{equation}
where $[.,.]$ denotes the Lie-bracket of vector fields, so the first line simply states that the $\vec H(x)$ furnish a representation of the Lie-algebra of vector fields on $\Sigma$ and where $N^a(f_1,f_2):x \mapsto g^{ab}(x)\left(f_1(x)f_{2,b}(x)-f_{1,b}(x)f_2(x)\right)$.

\subsection{Identify the symmetries to best match}

We now apply the general symmetry trading procedure outlined in Section~(\ref{sec:trading algorithm}) to derive shape dynamics. We treat in detail the spatially compact case but the asymptotically flat case is developed in \cite{Gomes:linking_paper}. For more details on the mathematical structure of the dualization and a more detailed proof of the dictionary, see \cite{Gomes:linking_paper,gryb:shape_dyn}.

Let us spell out the symmetry to be gained in exchange for foliation invariance. Na\" ively, we would trade $S(x)$ for constraints that generate general conformal transformations of $g$. However, as in the toy model, trading all such symmetries would lead to frozen dynamics. One global constraint must be left over to generate global reparametrization invariance. In analogy to the toy model, we restrict to conformal transformations that do not change the global scale. In geometrodynamics, the analogue of the moment of inertia is the total 3--volume. Thus, the desired symmetry is explicitly constructed in the following way:

Let $\mathcal C$ denote the group of conformal transformations on $\Sigma$ and parametrize its elements by scalars $\phi:\Sigma\to \mathbb R$ acting as
\begin{equation}
 \phi: \left\{
 \begin{array}{rcl}
  g_{ab}(x)&\to&e^{4 \phi(x)}g_{ab}(x)\\
  \pi^{ab}(x)&\to&e^{-4\phi(x)}\pi^{ab}(x).
 \end{array} \right.
\end{equation}
Consider the one--parameter subgroup $\mathcal V$ parametrized by homogeneous $\phi: x \to \alpha$. Notice that $\mathcal V$ is normal, because $\mathcal C$ is Abelian, so we can construct the quotient $\mathcal C/\mathcal V$ by building equivalence classes w.r.t. the relation
\begin{equation}
 \mathcal C \ni \phi \sim \phi^\prime \textrm{ iff }\exists \alpha \in \mathcal V, \textrm{ s.t. }\phi=\phi^\prime+\alpha.
\end{equation}
Given a metric $g$ on $\Sigma$ and $\phi \in \mathcal C$ we can find the unique representative in the equivalence class $[\phi]_\sim \in \mathcal C/\mathcal V$ that leaves $V_g=\int_\Sigma d^3x \sqrt{|g|}(x)$ invariant using the map
\begin{equation}
 \widehat{\,\,.\,\,}_g : \phi \mapsto \phi - \frac 1 6 \ln\langle e^{6\phi}\rangle_g,
\end{equation}
where we define the mean $\langle f \rangle_g:=\frac 1 V_g \int_\Sigma d^3x \sqrt{|g|}(x) f(x)$ for a scalar $f:\Sigma \to \mathbb R$. The map $\widehat{\,\,.\,\,}_g$ allows us to  parametrize $\mathcal C/\mathcal V$ by scalars $\phi$. Note that $\hat\phi_g$ can be written more transparently by observing that it is chosen so that the volume element of the conformally transformed metric is equal to
\begin{equation}
    \sqrt{ \lf| e^{4\hat\phi} g \rt| } = \frac{e^{6\hat\phi}}{\mean{e^{6\hat\phi}}} \sqrt{|g|}.
\end{equation}
Thus, the conformal transformation
\begin{equation}\label{eq conf symmetry geo}
    g_{ab}(x) \to \exp\lf( 4\hat\phi(x) \rt) g_{ab}(x)
\end{equation}
leaves
\begin{equation}
    V_g = \int d^3 x\sqrt{|g|} = \int d^3 x \frac{e^{6\hat\phi}}{\mean{e^{6\hat\phi}}} \sqrt{|g|} = V_g \frac{\mean{e^{6\hat\phi}}}{\mean{e^{6\hat\phi}}}
\end{equation}
invariant (we will often suppress the subscript $g$ in $\hat\phi_g$ for convenience). $\mathcal C/\mathcal V$ is then precisely the symmetry we want to obtain in exchange for foliation invariance. Given the transformation properties of $R[g]$ under (\ref{eq conf symmetry geo}), it is easy to show that the ADM is \emph{not} invariant. We are, thus, dealing with a non--equivariant best matching theory.

\subsection{Perform canonical best matching}

The next step is to enlarge the phase space $\Gamma(g;\pi) \to \Gamma_\text{e}(g,\phi;\pi,\pi_\phi) =\Gamma \times T^*(\mathcal C)$ with the canonical ``product Poisson bracket'' and parametrize conformal transformations by scalar functions $\phi$ and their conjugate momentum densities by scalar densities $\pi_\phi$. To ensure that $\phi$ and $\pi_\phi$ are purely auxiliary, we must impose the trivial constraint
\begin{equation}
    \pi_\phi \approx 0.
\end{equation}
This constraint has a trivial property that we will need in a moment: it commutes with \emph{any} smooth function $f(g;\pi)$ of the original phase space variables.

We can now do a canonical transformation $T:\{ g,\phi;\pi,\pi_\phi\} \to \{ G,\Phi;\Pi,\Pi_\phi\}$ that performs the shifting required for best matching. We define the generating function
\begin{equation}
 F[\phi]:=\int d^3x \lf( g_{ab}(x) \exp\left(4\hat \phi(x)\right)\Pi^{ab}(x) + \phi(x) \Pi_\phi(x) \rt).
\end{equation}
Note that capitalized variables represent transformed variables. Then, the canonical transformation $T$ generated by $F[\phi]$ is:
\begin{equation}\label{eq can trans geo}
 \begin{array}{rcccl}
   g_{ab}(x)&\to& Tg_{ab}(x)&=& \exp\left(4 \hat \phi(x)\right) g_{ab}(x)\\
   \pi^{ab}(x)&\to& T\pi^{ab}(x) &=& e^{-4 \hat \phi(x)} \left(\pi^{ab}(x)-\frac 1 3 \langle\pi\rangle_g \left(1-e^{6\hat \phi(x)}\right)g^{ab}(x)\sqrt{|g|}(x)\right) \\
   \phi(x) &\to& T\phi(x) &=&\phi(x) \\
   \pi_\phi(x) &\to& T\pi_\phi(x) &=& \pi_\phi(x) - 4\lf( \pi(x) - \mean{\pi}\sqrt g(x) \rt),
 \end{array}
\end{equation}
where $\pi = \pi^{ab}g_{ab}$. This is indeed a volume preserving conformal transformation. Under this transformation, the constraint $\pi_\phi \approx 0$ transforms to the new constraint
\begin{equation}
    \pi_\phi \approx 0 \to C\equiv \pi_\phi - 4\lf( \pi(x) - \mean{\pi}\sqrt g(x) \rt) \approx 0.
\end{equation}
Before the canonical transformation, $C$ was first class with respect to any smooth functions $f(g,p)$. Thus, after the canonical transformation, $C$ will continue to be first class with respect to the transformed functions $f(Tg, T\pi)$. In other words, $C$ acts trivially on the image of $T$. In particular, it will continue to commute with transformed scalar and diffeomorphism constraints. The interpretation of $C$ does not change after applying $T$. Its role in the theory is to remove the auxiliary degrees of freedom introduced by $\phi$ and $\pi_\phi$. It must be present to ensure that the extended theory is dynamically equivalent to the original. Note that $C$ could have been derived as a primary constraint directly from the Lagrangian by performing conformal best matching on the BSW action followed by a Legendre transform. This equivalent approach was followed in \cite{barbour_el_al:physical_dof}.

The scalar constraints transform in the following way under the map $T$:
\begin{multline}\label{eq LY 1}
    S(x)\to TS(x)= \frac{e^{-6\hat\phi}}{\sqrt{|g|}} \left[ \pi^{ab} \pi_{ab} - \frac{\pi}{2} - \frac{\langle\pi\rangle^2_g}{6} (1 - e^{6\hat\phi})^2|g|+ \frac{\langle\pi\rangle_g}{3}\pi(1 - e^{6\hat\phi})\sqrt{|g|} \right]
		 \\ - e^{2\hat\phi}\sqrt{|g|} \left[ R[g] - 8\left(D^2\hat\phi + (D\hat\phi)^2 \right) \right]
 %  H_a(x)\to TH_a(x)&= D^a \left( Tg_{ab}(x) T\pi^{bc}(x) \right) - 4 D^a \hat\phi(x) Tg_{ab}(x) T\pi^{bc}(x),
 \end{multline}
where $v^2 \equiv g_{ab} v^a v^b$ for any vector $v^a$. To avoid technical difficulties arising from the proper treatment of the diffeomorphism constraint, we remove it for the time being and verify after the dualization is completed that it can be consistently reintroduced into the theory.

The quantity
\begin{equation}\label{eq D def}
    \mathcal D \equiv \pi(x)-\sqrt{|g|}(x)\langle \pi\rangle_g
\end{equation}
will be the first class constraint left over in the dual theory. We pause for a moment to note its important properties, which can be verified by straightforward calculations. First, $\mathcal D$ is invariant under the canonical transformation (\ref{eq can trans geo}). Second, it generates infinitesimal volume preserving conformal transformations. This can be seen by noting that
\begin{align}\label{eq local vpct}
    \delta_\theta\, g_{ab}(x) &= \lf( 4\theta(x) - \mean{4\theta} \rt) g_{ab}(x) \notag \\
    \delta_\theta\, \pi^{ab}(x) &= \lf( - 4\theta(x) + \mean{4\theta} \rt) \lf(\pi^{ab}(x) - \frac 1 2 \mean{\pi}_g g^{ab}(x) \sqrt{|g|}(x) \rt),
\end{align}
where $\delta_\theta  = \pb{\cdot}{\mathcal D(4\theta)}$, is the infinitesimal form of (\ref{eq can trans geo}). This is the key property that we require of the dual theory.

We note one final property of $C(x)$. Just as in the toy model, one of the $C(x)$'s is an identity and not an independent constraint. This can be seen by noting that
\begin{equation}
    \int_\Sigma d^3x\, C(x) = \int_\Sigma d^3x\, \pi_\phi(x)
\end{equation}
but the Poisson bracket of $\int d^3x\, \pi_\phi$ with the variables $g_{ij}$, $\pi^{ij}$, $\hat\phi$, and $\pi_\phi$ is identically zero (the only non--trivial calculation is $\pb{\rho(\hat\phi)}{\int_\Sigma d^3x\, \pi_\phi(x)}=0$ for arbitrary smearings $\rho$). This is a consequence of restricting to $\mathcal C/\mathcal V$ instead of just $\mathcal C$. It means that $\int d^3x\,C$ can be removed from the theory without affecting the theory in any way. The result of this is that we will be left with a global first class scalar constraint at the end of the procedure.

\subsection{Impose and propagate best--matching constraints}\label{sec: prop pi_phi}

To make this theory relational with respect to volume preserving conformal transformations, we must perform a best--matching variation. This involves imposing the constraints
\begin{equation}
 \pi_\phi\approx 0.
\end{equation}
One of these constraints: $\int d^3x\, \pi_\phi$, is trivial as noted above. However, singling out this constraint explicitly is more subtle now than it was in the finite dimensional toy model. For simplicity, we will work with this redundant parametrization of the constraints, keeping in mind that it is over--complete by one equation.

Imposing $\pi_\phi(x)\approx 0$ turns all but one of the original scalar constraints into second class constraints, as can be observed from the Poisson bracket
\begin{equation}
 \{TS(N),\pi_\phi(\lambda)\}= \int dx\, \lambda(x) \left[ F_N - \sqrt{|g|}(x) e^{6\hat\phi(x)} \langle F_N \rangle_g \right]
\end{equation}
where we smeared $\pi_\phi$ with a scalar $\lambda$ and $F_N$ is given by
\begin{multline}
    F_N = 8 g_{ab} D^a\left(e^{2\hat\phi}D^bN \right)\sqrt{|g|} - 8N e^{2\hat\phi}\sqrt{|g|} \left[ R[g] - 8\left(D^2\hat\phi + (D\hat\phi)^2 \right) \right]  \\
		 - 2Ne^{6\hat\phi} \sqrt{|g|} \langle \pi \rangle_g^2 - N \left[ 6 TS + 2\pi \langle \pi \rangle_g(C + \pi_\phi) \right].
\end{multline}
The last term in $F_N$ is weakly zero. This is the canonical form of the non--equivariance condition. The consistency of the dynamics requires that we fix the Lagrange multipliers to satisfy the equation
\begin{equation}\label{equ:consistency-condition}
 \begin{array}{rcl}
   \{TS(N_0),\pi_\phi(x)\}&=&0.\\
 \end{array}
\end{equation}
The trivial solution $N_0\equiv 0$ yields ``frozen dynamics.'' Fortunately, there is a unique solution $N_0(\hat\phi,g_{ij},\pi^{ij})$ to the lapse fixing equation
\begin{equation}\label{equ:lapse-fixing}
    N \left[e^{-4\hat\phi} \left[ R[g] - 8\left(D^2\hat\phi + (D\hat\phi)^2 \right) \right] + \frac 1 4 \langle \pi \rangle_g^2 \right] -  e^{-6\hat\phi} g_{ab} D^a\left(e^{2\hat\phi}D^bN \right) = \mean{\mathcal L}_g
\end{equation}
that is equivalent to the first line of (\ref{equ:consistency-condition}). In (\ref{equ:lapse-fixing}), $\mathcal L$ is defined as $e^{6\hat\phi} \sqrt{|g|}$ times the left--hand--side of (\ref{equ:lapse-fixing}).

In terms of the transformed variables $G_{ab} \equiv Tg_{ab}$ and $\Pi^{ab} \equiv T\pi^{ab}$, (\ref{equ:lapse-fixing}) takes the simple form
\begin{equation}\label{eq lfe simple}
    \lf( R[G] + \frac{\mean{\Pi}_G}{4} - D^2  \rt) N = \mean{\lf( R[G] + \frac{\mean{\Pi}_G}{4} - D^2  \rt) N}_G.
\end{equation}
This is the same lapse fixing equation obtained in \cite{barbour_el_al:physical_dof} by a similar argument using the Lagrangian formalism. To understand why (\ref{eq lfe simple}) has a one parameter family of solutions (and, thus, depends only on $\hat\phi$) note that any solution, $N_0$, to
\begin{equation}\label{eq lfe standard}
    \lf( R[G] + \frac{\mean{\Pi}_G}{4} - D^2  \rt) N_0 = - c^2,
\end{equation}
for an arbitrary constant $c$, is also a solution to (\ref{eq lfe simple}). Eq (\ref{eq lfe standard}) is the well--known lapse fixing equation for CMC foliations. It is known to have unique positive solutions\cite{brill:lfe_uniqueness,York:york_method_long,Niall_73}. This fact is vital for our approach to work as it allows us to construct the Hamiltonian of the dual theory.

The one parameter freedom in $N_0$ is exactly what we expect from the redundancy of the $\pi_\phi(x)$'s. We are one global equation short of completely fixing the $N(x)$'s. Thus, we should be left with a single first class linear combination of the $TS(x)$'s. This is precisely the global Hamiltonian
\begin{equation}H_\text{gl}=\int d^3x\, N_0(x) TS(x)=TS(N_0),\end{equation}
where both $N_0(x)$ and $TS(x)$ are functionals of $g,\pi,\phi$. $H_\text{gl}$ is constructed to be first class with respect to $\pi_\phi$. We can now state the important result that \emph{the extended system with the constraint $\pi_\phi\approx 0$ is consistent as merely a (partial) gauge fixing of the original system.}

\subsection{Eliminate second class constraints}

Let us collect and classify the constraints of our theory to recap what we have done. We started with the first class constraints $TS(x)$ and then imposed the conditions $\pi_\phi(x) = 0$. The former are second class with respect to the latter but, since $\int d^3x \pi_\phi$ plays no role in the theory, there is a corresponding constraint $H_\text{gl}$ that is still first class. We can split the first class $H_\text{gl}$ from the remaining second class $TS(x)$'s by defining the variable $\widetilde{TS}(x) \equiv TS(x) - H_\text{gl}$. This leaves us with the first class constraints
\begin{align}
    H_\text{gl} &\approx 0 & C(x) \approx 0
\end{align}
as well as the second class constraints
\begin{align}
    \widetilde{TS}(x) &\approx 0 & \pi_\phi(x) \approx 0
\end{align}
in analogy with equations (\ref{eq 1st class part}) and (\ref{eq 2nd class part}) of the toy model. Note that, for the moment, we are still relaxing the diffeomorphism constraints. In Section~(\ref{sec: prop pi_phi}), we showed that $N_0$ is consistent gauge choice for general relativity. However, we should now construct the Dirac bracket to explore the full structure of the theory.

First, notice that since $\{H_\text{gl},\pi_\phi\}\approx 0$ we know that $\{\widetilde{TS}(x),\pi_\phi(y)\} \approx \{{TS}(x),\pi_\phi(y)\}$. This leads to the weak equality
$$ \int G(x,x')\{\widetilde{TS}(x'),\pi_\phi(y)\}d^3x'\approx\int G(x,x')\{{TS}(x'),\pi_\phi(y)\}d^3x'=\delta(x,y)$$
provided the Green's function, $G$, for the differential operator acting on $N$ in (\ref{eq lfe simple}) exists. This is guaranteed because the existence and uniqueness of solutions to the CMC lapse fixing equation, for suitable initial data and boundary conditions, implies existence of the Green's function for the respective initial data and boundary value problem.\footnote{This is implied by the existence of a Green function for the so called Lichnerowicz Laplacian $D^2 +R$, an operator that can be put into the form of a Hodge Laplacian $d\delta+\delta d$.} The formal existence of this Green's function is all that we will need for the remainder of the paper.

In terms of this Green's function, the Dirac bracket is defined as:
\begin{equation}\label{Dirac geometro}
 \begin{array}{rcr}
   \{f_1,f_2\}_D&:=&\{f_1,f_2\}-\int d^3 x d ^3 y\, \{f_1,\pi_\phi(x)\}G(x,y)\{\widetilde{TS}(y),f_2\}\\
    &&+\int d^3 x d^3 y \, \{f_1,\widetilde{TS}(x)\}G(x,y)\{\pi_\phi(y),f_2\},
  \end{array}
\end{equation}
for arbitrary functions $f_1$ and $f_2$. By construction, the Dirac bracket between $\pi_\phi$ or $\wts$ and any phase space function is strongly equal to zero. We can then eliminate these constraints from the theory by following Dirac's algorithm \cite{Dirac:CMC_fixing}. First, we set $\pi_\phi = 0$ everywhere that it appears; then we find $\hat\phi_0$ such that $\wts(\hat\phi_0) = 0$ and insert $\hat\phi_0$ into all constraints. Note that it is $\hat\phi$, and not $\phi$, that we can explicitly solve for. This is analogous to what happens in the toy model and arises because of the redundancy in the $\pi_\phi(x)$. When this is done, the system will have been reduced to the original phase space $\Gamma$. The Dirac bracket between functions $f$ of $\Gamma$ and the remaining first class constraints is weakly equivalent to the Poisson bracket. This is true, just as in the toy model, because the extra terms in the Dirac bracket are either proportional to $\{f,\pi_\phi(x)\}$, which is zero, or $\{\pi_\phi(y),H_\text{gl}\}$ and $\{\pi_\phi(y),\pi_\phi(x)\}$, which are both weakly zero. Finally, because we are simply inputting the solution of a constraint back into the Hamiltonian, the equations of motion on the contraint surface will remain unchanged. Thus, $T_{\phi_0}$ is still effectively a canonical transformation. For more details, see \cite{Gomes:linking_paper} or the appendix of \cite{gryb:shape_dyn}.

The last step is to verify that $\wts(\hat\phi_0) = 0$ can be solved for $\hat\phi_0$. We can simplify (\ref{eq LY 1}) using the strong equations $\pi_\phi = 0$ and by taking linear combinations of the first class constraints $C(x)$. This leads to the equivalent constraint
\begin{equation}\label{eq LY 2}
    \frac{e^{-6\hat\phi}}{\sqrt{|g|}} \lf[ \sigma^{ab} \sigma_{ab} - \frac{\mean{\pi}^2_g}{6} e^{12\hat\phi}|g| - e^{8\hat\phi} \bar R[g] \sqrt{|g|}\rt] \approx 0,
\end{equation}
where $\sigma^{ab} \equiv \pi^{ab} - \frac 1 3 \mean{\pi}_g g^{ab} \sqrt{|g|}$ and $\bar R[g] = R[g] - 8(D^2 \hat\phi + (D\hat\phi)^2)$.
Eq (\ref{eq LY 2}) is the Lichnerowicz--York (LY) equation used for solving the initial value problem of general relativity. Its existence and uniqueness properties have been extensively studied.\footnote{See \cite{Niall_73} or, for the specific context given here, see \cite{OMurchadha:LY_cspv}.} It is known to have unique solutions when $\sigma^{ab}$ is transverse and traceless. Fortunately, these are exactly the conditions required by the diffeomorphism and conformal constraints respectively. The formal invertibility of this equation is the second vital requirement for our procedure. Without this, we would not be able to prove the existence of the dual theory. However, given that we can solve (\ref{eq LY 2}) for $\hat\phi_0$ for specified boundary and initial data, we arrive at the dual Hamiltonian
\begin{equation}
    H_\text{dual}'= \mathcal N H_\text{gl}[\phi_0] + \int_\Sigma d^3 x \lambda(x) \left(\pi(x)-\langle \pi\rangle\right),
\end{equation}
where $\mathcal N$ is a spatially constant Lagrange multiplier representing the remaining global lapse of the theory. We can now reinsert the diffeomorphism constraint. This gives the final Hamiltonian
\begin{equation}\label{eq dual ham geo}
    H_\text{dual}= \mathcal N H_\text{gl}[\phi_0] + \int_\Sigma d^3 x \left[\lambda(x)\mathcal D(x) +\xi^a(x)T_{\hat\phi_0} H_a(x)\right],
\end{equation}
using the definition (\ref{eq D def}). Note that we must use the transformed $H_a(x)$ evaluated at $\hat\phi_0$ (recall that $\mathcal D$ is invariant under $T$).  As shown above, the usual Poisson bracket over $\Gamma$ can be used to determine the evolution of the system.
 
\subsection{Construct explicit dictionary}

We can verify that the dual Hamiltonian has the required properties. First, we check that all the constraints are first class. The global fixed Hamiltonian was constructed to be first class with respect to the conformal constraints $\mathcal D$. This can be seen by observing that the $TS(x)$'s are first class with respect to the $C(x)$'s and that the $C(x)$'s are equal to the $\mathcal D(x)$'s when $\pi_\phi(x) = 0$. The $H_a$'s are easily seen to be first class with $H_\text{gl}$ because they are first class with respect to the original $S(x)$'s and because $T$ is a canonical transformation. The $\pi_\phi$'s commute with themselves because they are ultra--local canonical variables. Lastly, one can directly verify that
\begin{equation}
    \pb{T_{\hat\phi_0} \vec H(\vec v)}{\mathcal D(f)} = T_{\hat\phi_0}\pb{\vec H(\vec v)}{\mathcal D(f)} = D(\mathcal L_v f) \approx 0,
\end{equation}
where $\vec v$ and $f$ are smearings.

Secondly, the dual theory is indeed invariant under volume preserving conformal transformations. For this, recall that the $\mathcal D$'s generate the infinitesimal form of (\ref{eq can trans geo}) according to (\ref{eq local vpct}). Furthermore, the theory is also invariant under 3D diffeomorphisms generated by the $H_a$'s. The theory is \emph{not}, however, invariant under 4D diffeomorphisms. The diffeomorphism invariance is only within the spatial hypersurfaces and is, thus, foliation preserving. This means, in particular, that the theory is not Lorentz invariant because it is not invariant under boosts.

Finally, there is a gauge in which the equations of motion of the two theories are equivalent. Compare the ADM Hamiltonian to that of the dual theory.
\begin{align}
    H_\text{ADM} &=\int_\Sigma d^3x\left(N(x)S(x)+\xi^a(x)H_a(x)\right) \notag \\
    H_\text{dual} &= \mathcal N H_\text{gl}[\phi_0] + \int_\Sigma d^3 x \left[\lambda(x)\mathcal D(x) +\xi^a(x)T_{\hat\phi_0} H_a(x)\right] \notag \\
             &= \int_\Sigma d^3 x\,\lf( \mathcal N N_0(x,\hat\phi_0) T_{\hat\phi_0} S(x) + \lambda(x)\mathcal D(x) +\xi^a(x)T_{\hat\phi_0} H_a(x)\right).
\end{align}
Because $T$ is a canonical transformation, the equations of motion of both theories take the same form in the gauges
\begin{align}\label{eq dic geo}
    N(x) &= \mathcal N N_0(x,\hat\phi_0) \notag \\
    \lambda(x) &= 0.
\end{align}
To map the solutions from one side of the duality to the other, we simply need to use the explicit function $\hat\phi_0$ and the map $T_{\hat\phi_0}$. This completes the dictionary.

This dictionary is particularly straightforward to use if one can find initial data for the ADM Hamiltonian that satisfies the initial value constraints. In this case $\hat\phi_0 = 0$. Because both theories are first class, this condition will be propagated and the solutions of each theory are \emph{equal} in the gauges (\ref{eq dic geo}). To find suitable initial data, we still must solve the LY equation. However, for the purpose of using the dictionary in the classical theory, we only need to solve this for a single point on phase space.\footnote{To prove that the dual theory actually exists and to study the quantum theory, we still must be able to solve the LY equation over all of phase space.}

Given the existence of the above dictionary, we arrive at the following proposition:
\begin{proposition}
  The theory with total Hamiltonian (\ref{eq dual ham geo}) is a gauge theory of foliation preserving 3--diffeomorphisms and volume preserving 3D conformal transformations. In the gauge $\lambda = 0$, this dynamical system has the same trajectories as general relativity in CMC gauge.
\end{proposition}
We define this theory to be \emph{shape dynamics}.

\subsection{Linking theory}\label{eq:first link theory}

It is powerful to view shape dynamics and GR as originating from a \emph{linking theory} just as was done for the toy models in Section~(\ref{sec:linking toy}). This approach was first done in \cite{Gomes:linking_paper}. The linking theory is simply the extended theory on the enlarged phased space $\Gamma_\text{e}$ consisting of the constraints
\begin{align}
    S(N) &\approx 0 \\
    H(\xi) &\approx 0 \\
    \pi_\phi(\rho) &\approx 0.
\end{align}
The canonical transformation $T$ leads to
\begin{align}
    T S(N) &\approx 0 \\
    T H(\xi) &\approx 0 \\
    T \pi_\phi(\rho) = C(\rho) &\approx 0.
\end{align}
The gauge fixing condition $\phi = 0$ trivializes the $T$ map and a phase space reduction eliminates the $C$ constraint. This leads immediately to the ADM theory. Alternatively, the gauge fixing condition $\pi_\phi = 0$ and subsequent phase space reduction eliminates $\widetilde {TS} \approx 0$ by setting $\phi = \phi_0$ then sends $C(\rho) \to D(\rho)$. This is shape dynamics. The dictionary is obtained by further picking the gauge fixings $\rho = 0$ in shape dynamics and $N = N_0$ (ie, CMC gauge) in ADM. These gauge fixings are shown in Figure~(\ref{fig:geo linking}).
\begin{figure}\label{fig:geo linking}
    \begin{center}
	\includegraphics[width = 0.9\textwidth]{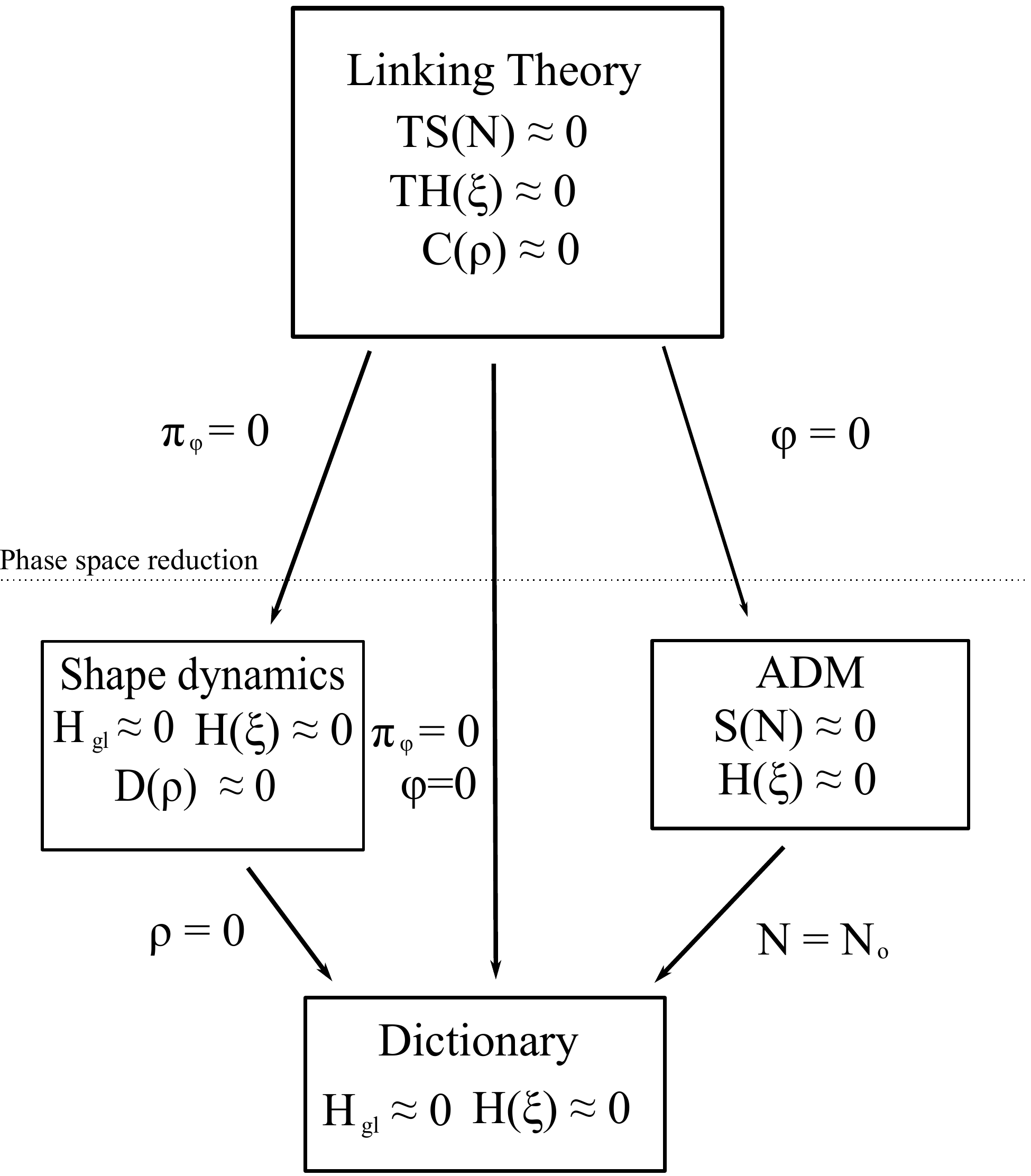}
    \end{center}
    \caption{Shape dynamics and ADM are obtained from different gauge fixings of a \emph{linking} theory.}
\end{figure}

%------------------------------ SHAPE DYNAMICS ------------------------
%----------------------------------------------------------------------
% SHAPE DYNAMICS
%----------------------------------------------------------------------
% This section explains gives the new results concerning SD. This will
% require, in some sense, the most work to get right. The plan is to put
% stuff on the V and \epsilon expansions as well as the HJ stuff and
% the dS/CFT correspondence. Finally, it would be nice to have a section
% on possible constructions principles and open questions.
%
% Jun 2: This is VERY rough... the order will change considerably!
%
%======================================================================
\chapter{Shape dynamics}\label{chap:shape_dynamics}
%======================================================================

In the last chapter we took a back door route to shape dynamics. We motivated its construction by applying a dualization procedure inspired by non--equivariant best matching. The result was a theory, equivariant with respect to volume preserving conformal transformations, with a highly non--local global Hamiltonian. The non--locality allows us to have, at the same time, conformal invariance \emph{and} equivalence to GR. Because of its global lapse, one can think of shape dynamics as a genuine geodesic principle on conformal superspace, albeit one with a complicated configuration space metric. Our choice of metric is made purely to have a dynamics that is equivalent to GR. The purpose of this chapter is to try to understand this choice better. We will try to piece together some of the basic structures of shape dynamics in the hopes of understanding it better as a theory in its own right.

The way that shape dynamics has been presented in terms of a dualization procedure hides some aspects of the simple relationship between shape dynamics and GR. In this chapter, we give a simple picture that illustrates how the shape dynamics Hamiltonian is actually constructed in practice. Because it is obtained through the solution of a non--linear partial differential equation on an arbitrary compact manifold without boundary, $\Sigma$, the global Hamiltonian is not straightforward to handle analytically in full generality. It is, thus, convenient and insightful to compute the Hamiltonian in different perturbative expansions. We will discuss two different approaches to solving for the global Hamiltonian: 1) by Taylor expanding in terms of the inverse of the volume, 2) by expanding in fluctuations about a background. The first approach reveals that shape dynamics has an intriguing behavior at large volume: it becomes a fully conformal theory. The second is useful for studying cosmological solutions of shape dynamics. After computing the Hamiltonian in these two expansions, we will calculate the Hamilton--Jacobi functional in the large volume limit. This large volume expansion contains more information than the usual derivatives expansions that can be performed in standard GR about asymptotically DeSitter (or Anti--DeSitter) backgrounds.

These calculations highlight the main advantage that shape dynamics presents over GR: that the local constraints are linear in the momenta. This has two important consequences: 1) that the constraints themselves can be solved analytically, and 2) that group representations can be formed by exponentiating the local algebra. In other words, it is straightforward to construct quantities that are invariant under volume preserving conformal transformations. Furthermore, there are special gauges, which \emph{do not} correspond to GR in CMC gauge, where certain calculations are drastically simplified. So although it is true that no calculation should, in principle, be more difficult in shape dynamics then in standard GR; in practice, having control of the gauge invariance in shape dynamics can help to organize the calculation in a more efficient way. 

The conformal behavior of the shape dynamics Hamiltonian and the computation of the Hamilton--Jacobi functional at large volume suggests a potentially powerful way of studying gauge/gravity dualities using a holographic renormalization group flow equation in shape dynamics. The immediate advantage presented by an approach based on shape dynamics is that shape dynamics is \emph{nearly} conformally invariant for from the onset. As a result, the action of the canonically quantized conformal constraints on the shape dynamics wavefunction is \emph{nearly} identical to the action of the conformal Ward identities on a CFT living on an arbitrary background. The only thing that breaks the conformal invariance is global Hamiltonian, which, as we will see, is conformal at large volume. This may provide us with an interesting option: to treat Hamiltonian flow in shape dynamics as renormalization group flow in a CFT. If this were possible, it may provide a construction principle for shape dynamics that is independent of general relativity. We will present some hints that this may be possible. A deeper exploration of this exciting option is reserved for future work. Much of the work presented in this chapter is based on \cite{gryb:gravity_cft}.

%---------------------- This basic description comes from our gravity/cft correspondence paper --------------------------

\section{Basic Construction}\label{sec:sd basic}

In Section~(\ref{sec:sd derivation}), we rigorously constructed shape dynamics using non--equivariant best matching applied to GR in ADM form. We will now present a different derivation that is less rigorous but more intuitive. The aim is to establish a clear picture that can be used to better understand the nature of the shape dynamics Hamiltonian.

Our starting point is to look for a theory that is invariant under volume preserving conformal transformations and 3--dimensional diffeomorphisms and has the same dynamical trajectories and initial value problem as GR in CMC gauge. Thus, we are looking for a theory with a Hamiltonian of the form
\begin{equation}\label{eq:sd tot ham}
   H_\text{SD}(\lambda) = \int_\Sigma d^3 x\lf[ N^a(x,\lambda) H_a + 4\rho(x,\lambda) D \rt] + \mathcal N(\lambda) \hg,
\end{equation}
where, as before, $H_a$ generates 3--diffeos and $D$ generates infinitesimal volume preserving conformal transformations. $\hg$ is some global Hamiltonian, depending only on $\lambda$, that is chosen so that this theory is equivalent to GR. Recall that the volume preserving condition is enforced by using hatted conformal factors $\hat\phi$ that are required to satisfy
\begin{equation}
   \mean{e^{6\hat\phi}} = 1.
\end{equation}
For convenience, we also recall that the metric $g_{ab}$ and its conjugate momentum density $\pi^{ab}$ transform in the following way under volume preserving conformal transformations
\begin{equation}
\begin{array}{l}
   g_{ab} \to e^{4 \hat \phi} g_{ab}, \\
  \pi^{ab} \to e^{-4 \hat \phi} \left[\pi^{ab}-\frac 1 3 \langle\pi\rangle \left(1-e^{6\hat \phi}\right)g^{ab}(x)\sqrt{g}\right].
 \end{array} \label{eq:vpct}
\end{equation}

We must find an expression for $\hg$ such that the flow of $H_\text{SD}$ when $\rho = 0$ is identical to that of the usual ADM Hamiltonian
\begin{equation}
   H_\text{ADM}(\lambda) = \int d^3 x\lf[ N^a(x,\lambda) H_a(x,\lambda) + N(x,\lambda) S(x,\lambda) \rt]
\end{equation}
in CMC gauge. To obtain $\hg$, consider the unique linear combination $\tilde S$ of $S$ such that $\det_{x,y} \pb{\tilde S (x)}{D(y)} \neq 0$. $\tilde S$ represents the part of $S$ that is gauge fixed by the CMC condition $\frac {\pi}{\sqrt g}  = \mean{\pi}$ (or $D=0$). To prove that $\tilde S$ exists and is unique it is most convenient to extend the phase space and work in the \emph{linking} theory, as was done in Section~(\ref{sec:sd derivation}). We will not dwell on the details of this proof here. For a more complete analysis, see \cite{Gomes:linking_paper, gomes:phd}. As a consequence of the invertibility of $\pb{\tilde S (x)}{D(y)}$, the configuration space surfaces $\tilde S = 0$ and $D = 0$ have a unique intersection (see Fig (\ref{fig:sd})) that we require to be the common trajectory shared by shape dynamics and GR. To fulfill this requirement, $\hg$ is defined as the part of $S$ on the intersection that is \emph{not} gauge fixed by the $D$'s. $\hg$ can be defined anywhere on the $D=0$ surface by lifting the value of $\lf. S\rt|_{\tilde S = 0, D=0}$ along the gauge orbits of $D$. Since these are the volume preserving conformal transformations \eq{vpct}, $\hg$ on $D$ is
\begin{equation}\label{eq:hg}
   \hg(g,\pi) = \frac 1 {\sqrt g} S (t_\phi g, t_\phi \pi),
\end{equation}
where we have used the $t_\phi$ map as a short hand for the volume preserving conformal transformations, \eq{vpct}. As illustrated in Fig (\ref{fig:sd}), $\phi$ needs to be found such that $t_\phi$ brings you to the surface $\tilde S = 0$, where $S$ is constant. Thus, \eq{hg} needs to be solved for $\phi$ such that the left hand side is a constant. As an immediate consequence of this definition, $\hg$ is invariant under finite volume preserving conformal transformations.
\begin{figure}
\begin{center}
\includegraphics[width=0.45\textwidth]{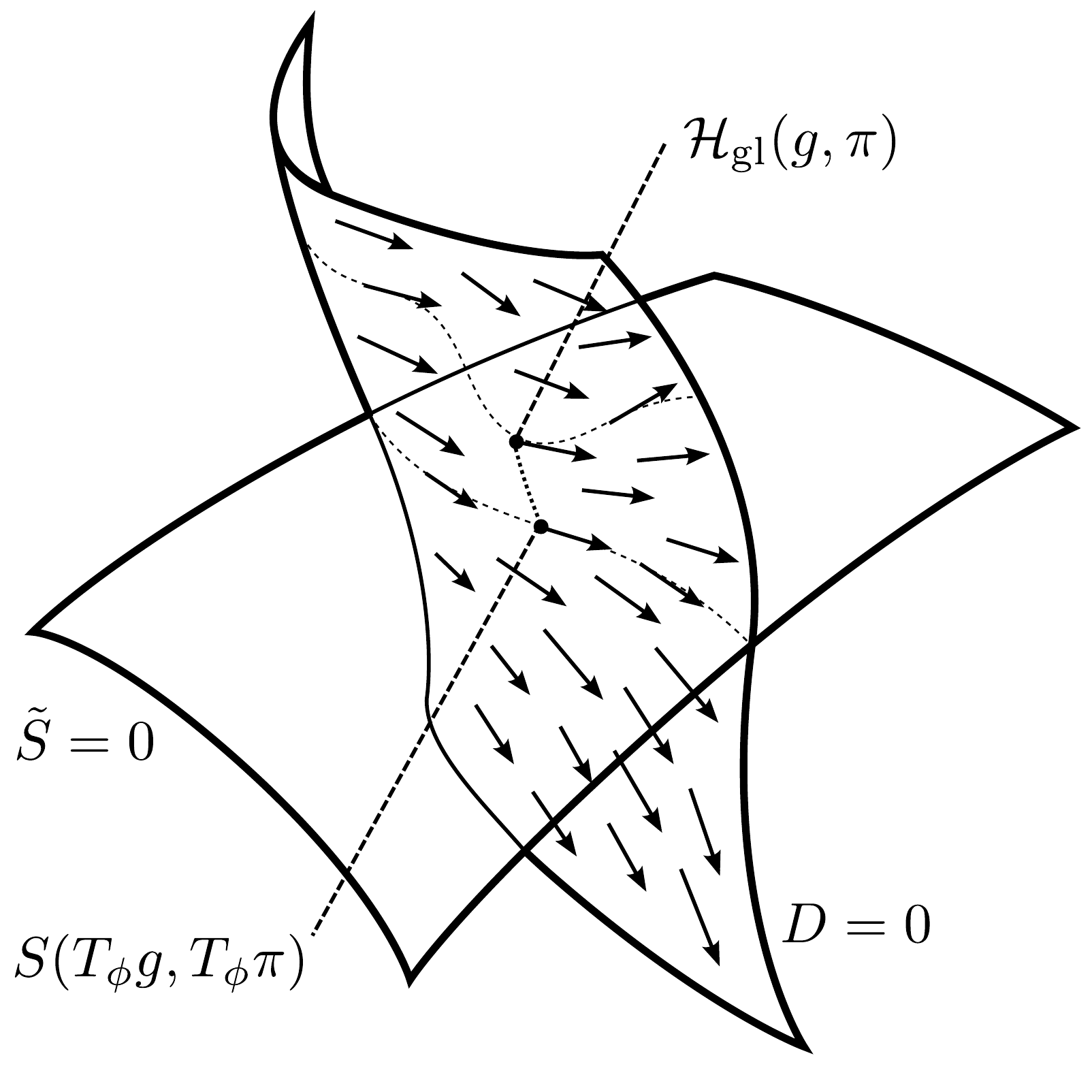}
\caption{The definition of $\hg$. The constraint surface $D\approx 0$ provides a proper gauge fixing of $\tilde S \approx 0$. $\hg$ is defined by the value of $S$ at the gauge fixed surface, represented by the dark dotted line.}\label{fig:sd}
\label{default}
\end{center}
\end{figure}

In summary, the shape dynamics Hamiltonian is given by \eq{sd tot ham}, where $\hg$ and $\hat\phi$ are given by solving the simultaneous equations
\begin{align}
    \hg(g,\pi) &= \frac 1 {\sqrt g} S (t_\phi g, t_\phi \pi) \notag \\
    \mean{e^{6\hat\phi}} &= 1. \label{eq:SD-definition}
\end{align}
This is completely equivalent to final algorithm obtained in Section~(\ref{sec:sd derivation}). However, this derivation presents a much clearer picture of how shape dynamics is constructed from GR. These are two theories living on different intersecting surfaces of the ADM phase space. The intersection of the two theories contains the physical trajectories so that any flow in directions where the surfaces do not intersect is unphysical and can always be projected back down the intersection.

\subsection{Yamabe problem}

We note that the above algorithm is analogous to that used to solve the Yamabe problem \cite{YamabeConjecture} that all compact manifolds in $d\geq 3$ are conformally constant curvature. For this, one needs to show that a constant $\tilde R$ can be found such that
\begin{equation}\label{eq:yamabe}
   \tilde R = R(e^{4 \hat \lambda[g,x)} g)
\end{equation}
for some non--local functional $\hat \lambda[g,x)$ of $g$. The restriction $\mean{e^{6\hat\lambda}} = 1$ selects a unique value of $\tilde R$. This suggests an intriguing connection between the Yamabe problem and shape dynamics. Indeed, $\tilde R$ will play an important role in the solution of the  equation of shape dynamics in the large $V$ expansion.

\subsection{Matter coupling}

The inclusion of matter into the shape dynamics formalism has recently been achieved in \cite{MatterPaper}. The procedure is relatively straightforward and follows from the results of \cite{Isenberg:matter_coupling1,Isenberg:matter_coupling2}. The procedure can be implemented for any matter for which the initial value problem can be solved in GR. We will only consider the pure gravity sector of the correspondence, for which interesting conclusions can already be drawn. The more complete analysis that includes matter is currently underway.

\section{Large volume expansion}

Constructing the shape dynamics Hamiltonian involves solving the system of equations \eq{SD-definition}, which constitute a non--linear partial differential equation under a global integral condition. This task is not straightforward, even when considering simple topologies for $\Sigma$. However, it is possible to find simple expansions that lead to recursion relations with which $\hg$ can be solved. The first expansion we will consider is a Taylor expansion in terms of $1/V^{2/3}$. The $2/3$ power is chosen because this is how the metric changes in terms of the volume under global rescalings. The large volume expansion is useful, not just as a computational tool, but also for providing insight into the nature of $\hg$. As we will see, it may provide a link between shape dynamics and CFT.

\subsection{Canonical transformation}

To perform the large $V$ expansion, it will be convenient to take advantage of the fact that $R$ is conformally constant by shifting
\begin{equation}
  \hat  \phi \to  \hat \phi +  \hat \lambda,
\end{equation}
so that $R \to \tilde R$. Recall that, by \eq{yamabe}, $\hat\lambda$ is the conformal factor used to construct a conformally constant curvature metric, $\tilde g_{ab}$, through $\tilde g_{ab} = e^{4\hat\lambda[g]} g_{ab}$ and $\tilde R$ is the constant value of the curvature when $\tilde g$ has unit volume. Since they are solutions to \eq{yamabe}, both $\tilde R$ and $\hat\lambda$ are non--local functionals of the metric. Note that this freedom drastically reduces the amount of work that needs to be done to perform the large $V$ expansion. This is evidence that knowing how to construct gauge invariant quantities in shape dynamics can help organize certain calculations in a more efficient way.

To expand $\hg$ in powers of $V^{-2/3}$, the explicit $V$ dependence of $\hg$ must be isolated. This can be done using the change of variables $(g_{ab};\pi^{ab}) \to (V,\bar g_{ab}; P, \bar \pi^{ab})$ given by
\begin{align}
   \bar g_{ab} &= \lf(\frac{V}{V_0} \rt)^{-\frac{2}{3}} g_{ab}, \qquad \qquad V = \int d^3x \sqrt{g}, \\
   \bar \pi^{ab} &= \lf(\frac{V}{V_0}\rt)^{\frac{2}{3}}\lf(\pi^{ab} - \frac 1 3 \mean{\pi} g^{ab} \sqrt{g} \rt),  ~ P= \frac 2 3 \mean{\pi},
\end{align}
where $V_0 = \int d^3 x \sqrt{\bar g}$ is a fixed reference volume. The quantity $P$ is the York time in CMC gauge, which is spatial constant by definition. This canonical transformation is motivated by the fact that the York time is canonically conjugate to the volume in CMC gauge. Thus, $\pb{V}{P} = 1$. Our goal is to completely extract the volume dependence from $g_{ab}$ and $\pi^{ab}$. This can be achieved for the metric by simply dividing by the appropriate power of the volume as is done in the construction of $\bar g_{ab}$. It is straightforward to verify that $\pb{P}{\bar g_{ab}} = 0$ so that
\begin{equation}
    \diby{\bar g_{ab}}{V} = 0.
\end{equation}
For the momenta $\pi^{ab}$, we must extract the trace part, which is canonically conjugate to $V$, before dividing by the appropriate power of the volume. This leads to $\bar \pi^{ab}$, which indeed satisfies $\pb{P}{\bar \pi^{ab}} = 0$ so that
\begin{equation}
    \diby{\bar \pi^{ab}}{V} = 0.
\end{equation}
Thus, the barred quantities and $(V,P)$ can be used to extract the volume dependence of $g_{ab}$ and $\pi^{ab}$. However, care must be taken because $\bar g$ and $\bar \pi$ are \emph{not} canonically conjugate so that this is \emph{not} a canonical transformation. As a result, we will have to exercise caution later when solving the Hamilton--Jacobi equation or if one would like to quantize shape dynamics using these variables.

\subsection{Expansion}

If we use these variables and call $\Omega = e^{\hat\phi+ \hat\lambda}$, we can insert (\ref{eq:vpct}) into (\ref{eq:SD-definition}) and, using $D = 0$, see that the equations \eq{SD-definition} transform into
\begin{align}\label{eq:main hg}
&\begin{array}{r}
   \hg = \lf( 2\Lambda - \frac{3}{8}P^2 \rt) + \frac{\lf( 8 \tilde{\nabla}^2 - \tilde R \rt) \Omega}{(V/V_0)^{2/3}\Omega^5}
       - \frac{\bar \pi^{ab} \bar \pi_{ab} }{(V/V_0)^2\Omega^{12}\tilde g}
     \end{array}
   \\
 &  \mean{\Omega^6} = 1, \label{eq:Omega}
\end{align}
where tilded quantities and means are calculated using $\tilde g_{ab}$. We will solve these equations by inserting the expansion ansatz
\begin{equation}
   \hg = \sum_{n=0}^\infty  \left(\frac{V}{V_0} \right)^{-2n/3} \hn{n}, ~~ \Omega^6 = \sum_{n=0}^\infty  \left(\frac{V}{V_0} \right)^{-2n/3}  \wn{n}
\end{equation}
and solving order by order in $1/V^{2/3}$. The reason for expanding $\Omega^6$ instead of $\Omega$ is that the restriction \eq{Omega} is trivially solved by
\begin{equation}
    \mean{\wn{n}} = 0 \label{eq:mean omega}
\end{equation}
for $n \neq 0$ and $\mean{\wn{0}} = 1$.

We can solve for the $\hn{n}$'s by inserting the expansion, taking the mean, and using the fact that $\tilde R$ is constant. We will demonstrate the procedure by working out the first couple of terms. For $n= 0$, we have trivially.
\begin{equation}
    \hn{0} = 2\Lambda - \frac 3 8 P^2.
\end{equation}
For $n=1$, we get
\begin{align}\label{eq:hn1 prelim}
    \hn{1} &= - \frac{1}{\wn{0}^{5/6} V^{2/3}} \left( 8 {\tilde \nabla}^2  - \tilde R \right) \wn{0}^{1/6} = - R \left( \wn{0}^{2/3} \tilde g \right) \notag \\
           &= - R(\wn{0}^{2/3} \tilde g).
\end{align}
If we take the mean of both sides, we obtain
\begin{equation}\label{eq:hn1 inter}
    \hn{1} = - \mean{ R(\wn{0}^{2/3} \tilde g) }.
\end{equation}
Inserting this into \eq{hn1 prelim} gives
\begin{equation}
    R(\wn{0}^{2/3} \tilde g) = \mean{R(\wn{0}^{2/3} \tilde g)}.
\end{equation}
This simply says that $\wn{0}^{2/3} \tilde g_{ab}$ must be the unit volume metric with constant curvature. In other words, $\wn{0}^{2/3} \tilde g_{ab} = \tilde g_{ab}$, or $\wn{0} = 1$. Inserting this result into \eq{hn1 inter}
\begin{equation}
    \hn{1} = - \tilde R.
\end{equation}
Because the solution to the Yamabe problem is \emph{unique}, we know that this is the \emph{only} solution for $\hn{1}$ and $\wn{0}$. 

For $n=2$, we get
\begin{equation}\label{eq:hn2 prelim}
    \hn{2} = -\frac{2}{3} \left( \tilde R + 2 {\tilde \nabla}^2 \right)  \wn{1}.
\end{equation}
Taking the mean and using integration by parts to drop boundary terms (we are on compact without boundary $\Sigma$) we get
\begin{align}
    \hn{2} &= -\frac {2\tilde R} 3 \mean{\wn{1}} \notag \\
	    &= 0.
\end{align}
In the last line, we used the fact that $\mean{\wn 1} = 0$ because of \eq{mean omega}. Inserting this into \eq{hn2 prelim} gives
\begin{equation}
    \lf( \tilde R + 2 \tilde \nabla^2 \rt) \wn 1 = 0.
\end{equation}
This equation has the solution
\begin{equation}
    \wn 1 = 0.
\end{equation}
This solution is unique provided $R$ is not in the discrete spectrum of $2\nabla^2$. We will ignore the measure zero case where $R$ happens to be in this spectrum. For more details on the uniqueness of this procedure, see \cite{gomes:phd}.

The subsequent orders can be worked out in a similar fashion. In general, the solution for $\hn{n}$ can be used to solve recursively for $\wn{n}$. Collecting the first three terms, we get
\begin{eqnarray}\label{eq:hg v exp}
   \hg &=& \lf( 2\Lambda - \frac 3 8 P^2 \rt) - \lf(\frac{V_0}{V}\rt)^{2/3}  \mean{\tilde R} \nonumber
           \\
&&           + \lf(\frac{V_0}{V}\rt)^2 \mean{\frac{\bar \pi^{ab} \bar \pi_{ab}}{\bar g}} + \order{\lf(V/V_0\rt)^{-8/3}}.
\end{eqnarray}
The next order terms are significantly more complex and non--local as they involve the inverse of the operator $R + 2\nabla^2$.

The expression \eq{hg v exp} combined with the conformal constraints allows us to make a classical connection with CFT. Using $P = \frac 2 3 \mean{\pi}$, solving $\hg = 0$ for $\mean{\pi}$, and adding the result to $D\approx 0$, we obtain
\begin{equation}
   \pi / \sqrt{g} = \pm \, c.
\end{equation}
$c$ is a spatial constant and commutes with all conformal constraints. It is, thus, a \emph{central charge} of the conformal algebra. We see that, at the classical level, shape dynamics is equivalent to \emph{two} classical CFTs with specific central charges. In the quantum theory, we will see that it may be possible to treat the terms contributing to $c$ as a 1--loop renormalization of the conformal anomaly. This observation may have some interesting applications in defining shape dynamics directly from a pair of CFTs.

%---------------------------- Now a discussion to the pert theory stuff --------------------------------

\section{Perturbative shape dynamics}\label{sec:pert sd}

In this section, we introduce a perturbative expansion of fluctuations, $h_{ab}$ and $p^{ab}$, about some background metric, $g_{ab}$, and momenta, $\pi^{ab}$,
\begin{align}
    g_{ab} &\to g_{ab} + \epsilon h_{ab}\\
    \pi^{ab} &\to \pi^{ab} + \epsilon p^{ab}.
\end{align}
Our goal is to solve the system of equations (\ref{eq:SD-definition}) order by order in $\epsilon$ for the expansion parameters, $\hn{n}$ and $\wn{n}$, of $\hd$ and $\Omega$ such that
\begin{align}
    \hd &= \sum_{n=0}^\infty \epsilon^n \hn{n} \\
    \Omega^6 &= \sum_{n=0}^\infty \epsilon^n \wn{n} .
\end{align}

It is straightforward to compute (\ref{eq:SD-definition}) to zeroth order. The result is
\begin{align}
    \mean{\omega_{(0)}} &= 1, \\
    \hn{0} &= \frac{1}{\wn{0}^{2} g} \lf( \pi\cdot\pi - \frac{\mean{\pi}^2}{3} \rt) - \frac{\mean{\pi}^2}{6} + 2\Lambda - \frac{R}{\wn{0}^{4/6}} + 8 \frac{\nabla^2 \wn{0}}{\wn{0}^{5/6}}.\label{eq:LY 0}
\end{align}
Equation \eq{LY 0} has the same form as (\ref{eq:SD-definition}), which is known to have a unique solution. However, \eq{LY 0} is simpler because it is in terms of the background metric only. If we can pick a background for which $g_{ab}$ and $p^{ab}$ are spatial constants, then it is easy to verify that the unique solution is
\begin{align}
    \wn{0} &=1 \\
    \hn{0} &=  \frac{1}{ g} \pi\cdot\pi - \frac{\mean{\pi}^2}{2} + 2\Lambda - R = 0.
\end{align}
As expected, we've recovered the undensitized Hamiltonian constraint of the ADM theory on CMC surfaces in terms of the background fields. Since we're expanding about a solution, $\hn{0}$ is satisfied by assumption. In appendix~(\ref{Appendix-useful formula}), we give $g_{ab}$ and $\pi^{ab}$ for a background deSitter spacetime in CMC gauge.

To deal with the higher orders, we need the following propositions.
\begin{proposition}
    The equation
    \begin{equation}\label{eq:main sd equation}
	\ts = \hd
    \end{equation}
    can be expanded order by order in perturbation theory, for $n>0$, by solving
    \begin{equation}\label{eq:main pert}
	H_{(n)} = \lf( \nabla^2 + m^2_{(n)} \rt) \omega_{(n)} - f_{(n)},
    \end{equation}
    where
    \begin{itemize}
	\item $\nabla^2$ and $m_{(n)}$ depend only on the backgrounds $g_{ab}$ and $\pi^{ab}$.
	\item $f_{(n)}$ depends on $\omega_{(i)}$ only for $i<n$ and all perturbations of the metric and momenta up to $n^\text{th}$ order.
    \end{itemize}
\end{proposition}
\begin{proof}
    $\wn{n}$ \emph{must} appears in \eq{main sd equation} with a factor of $\epsilon^n$. Thus, any product of it with the operator $\nabla^2 + m^2_{(n)}$ will be higher order in epsilon unless this operator depends only on the background. Furthermore, $f_{(n)}$ will contain higher orders of $\epsilon$ unless it is composed of $\wn{i}$ for $i<n$. The case $i=n$ can be absorbed by $m_{(n)}$.
\end{proof}
\begin{proposition}
    The condition
    \begin{equation}
	\mean{\Omega^6} = 1
    \end{equation}
    can be expressed order by order in perturbation theory, for $n>0$, as
    \begin{equation}\label{eq:omega pert}
	\mean{\wn{n}} = W_{(n)}
    \end{equation}
    where $W_{(n)}$ depends on $\omega_{(i)}$ only for $i<n$ and all perturbations of the metric and momenta up to $n^\text{th}$ order.
\end{proposition}
\begin{proof}
    Because $\wn{n}$ necessarily introduces a factor of $\epsilon^n$, $\wn{n}$ can only enter $\mean{\Omega^6}$ through $\mean{\wn{n}}$ calculated with the background metric. The other terms of $\mean{\Omega^6}$ are necessarily products of lower order expansions of the background and of $\Omega$.
\end{proof}

If we can find a background for which $m_{(n)}$ is a spatial constant depending only on time (the deSitter background given in Appendix~(\ref{Appendix-useful formula}) is an example of such a background) then we can solve \eq{omega pert} and \eq{main pert} for $n>0$ explicitly. Taking the mean of \eq{main pert}, we get
\begin{align}
    H_{(n)} &= \mean{\lf( \nabla^2 + m^2_{(n)} \rt) \omega_{(n)} - f_{(n)}}, \\
	    &= m^2_{(n)} \mean{\omega_{(n)}} - \mean{f_{(n)}} \\
	    &= m^2_{(n)} W_{(n)}- \mean{f_{(n)}},
\end{align}
where, in the second line, we used integration by parts and the constancy of $m$ and, in the third line, we used \eq{omega pert}. In general, $f_{(n)}$ will depend on $\omega_{(i)}$ for $i<n$. We can compute these by inserting our solution for $H_{(n)}$ into \eq{main pert}. This leads to
\begin{equation}\label{eq:omega}
    \omega_{(n)} = W_{(n)} + \lf( \nabla^2 + m_{(n)}^2(t) \rt)^{-1} \lf[ f_{(n)} - \mean{f_{(n)}} \rt].
\end{equation}
This equation can be solved by first finding the solution $\tilde{\omega}_{(n)}$ of
\begin{equation}
    \tilde{\omega}_{(n)} = \lf( \nabla^2 + m_{(n)}^2(t) \rt)^{-1} f_{(n)}
\end{equation}
then setting
\begin{equation}
    \omega_{(n)} = W_{(n)} + \tilde{\omega}_{(n)} - \mean{\tilde{\omega}_{(n)}}.
\end{equation}
That this ansatz solves \eq{omega} follows directly from the linear action of the operator $\nabla^2 + m_{(n)}^2(t)$ on the mean. For certain backgrounds, the operator $\nabla^2 + m_{(n)}^2(t)$ can be inverted straightforwardly. For the deSitter background presented in Appendix~(\ref{Appendix-useful formula}), the Laplacian is simply the Laplacian on the sphere. The explicit calculation of $\wn{n}$ arbitrary order is, therefore, straightforward and left for future work (some variations useful for this calculation are given in the Appendix~(\ref{Appendix-useful formula})).

%--------------------------------------- And back to the gravity/cft paper stuff --------------------------------------------

\section{Hamilton--Jacobi (HJ) equation}

Using the $1/V^{2/3}$ expansion, we can solve the HJ equation for shape dynamics with symmetric boundary conditions. This is useful both for solving the classical theory and for later drawing a connection with CFT. Remarkably, gauge invariance will allow us to solve for \emph{all} local constraints in the large $V$ expansion. This provides more information at each order than the usual derivative expansions used to solve the HJ equation in GR \cite{Freidel,Verlinde}. The HJ equation can be obtained from \eq{hg v exp} by making the substitutions
\begin{align}
   P &\to \ddiby{S}{V} & \pi^{ab} &\to \ddiby{S}{ g_{ab}},
\end{align}
where $S= S(g_{ab}, \alpha^{ab})$ is the HJ functional that depends on the metric $g_{ab}$ and parametrically on the separation constants $\alpha^{ab}$, which are symmetric tensor densities of weight 1. We can express $\bar \pi^{ab}$ in terms of $\ddiby{S}{ g_{ab}}$
\begin{equation}
    \bar \pi^{ab} \to V^{2/3} \lf( \ddiby S {g_{ab}} - \frac 1 3 \mean{ \ddiby S {g_{cd}} g_{cd}} g^{ab} \sqrt g \rt)
\end{equation}
then use the chain rule
\begin{equation}
    \ddiby S {g_{ab}(x)} = \ddiby S V \ddiby V {\bar g_{ab}(x)} + \int d^3y\, \ddiby S {\bar g_{ab}(y)} \ddiby {\bar g_{ab}(y)} {g_{ab}(x)}
\end{equation}
to write the result in terms of $\ddiby{S}{V}$ and $\ddiby{S}{\bar g_{ab}}$. Remarkably, the $V$ derivatives drop out of the final expression:
\begin{equation}\label{eq:hj sub}
   \bar \pi^{ab} \to \ddiby{S}{\bar g_{ab}} - \frac 1 3 \mean{ \bar g_{ab} \ddiby{S}{\bar g_{ab}}} \bar g^{ab} \sqrt{\bar g}.
\end{equation}

The strategy will be to expand $S$ in powers of $(V/V_0)^{-2/3}$,
\begin{equation}
   S = \sum_{n=0}^\infty \lf( \frac{V}{V_0}\rt)^{(3-2n)/3} \sn{n}
\end{equation}
then insert this expansion into the HJ equation obtained using the substitution \eq{hj sub}. To obtain a complete integral of the HJ equation, $\sn{0}$ can be taken of the form $\sn{0}=\int d^3 x \, \alpha^{ab}g_{ab}$. The linear constraints determine $\alpha^{ab}$ to be transverse and covariantly constant trace. The leading order HJ equation determines the value of the trace of $\alpha^{ab}$. This restricts the freely specifiable components of $\alpha^{ab}$ precisely to the freely specifiable momentum data in York's approach \cite{York:cotton_tensor}.

When $\alpha^{ab}$ have a vanishing transverse--traceless part, it is possible to solve for the HJ functional exactly using a recursion relation. These conditions are compatible with asymptotic (in time) dS space, which has maximally symmetric CMC slices. The treatment of general $\alpha$'s is currently under investigation. The vanishing transverse--traceless condition implies that
\begin{equation}
   \sn{0} = \pm \sqrt{ \frac {16} 3 \Lambda }\, V_0.
\end{equation}

To obtain the remaining $\sn{n}$'s, we can use the gauge invariance of $\hg$ under the action of the $D$'s to work in a gauge where $R$ is constant. In this gauge, $\tilde R = R$ and the variations of $\tilde R$ can be found using the standard variations of $R$. The $\sn{n}$'s can be found recursively using our solution for $\sn{0}$ and by collecting powers of $(V/V_0)^{-2/3}$. The first terms are
\begin{align}
   \sn{1} &= \mp \sqrt{\frac 3 \Lambda }  \tilde R \, V_0 = \mp \sqrt{\frac 3 \Lambda }  \int d^3x \sqrt {\bar g} \tilde R, \\
   \sn{2} &= \pm \lf(\frac 3 \Lambda\rt)^{3/2} \int d^3 x \sqrt {\bar g} \lf( \frac 3 8 \tilde R^2 - \tilde R^{ab} \tilde R_{ab} \rt).
\end{align}
Note that $\sn{0}$ and $\sn{1}$ are the only terms with positive dimension. The higher order terms can be obtained straightforwardly but become increasingly more involved because of the non--local terms appearing in the V expansion of $\hg$. Gauge invariant solutions can be obtained by restoring the $\lambda[g,x)$ dependence of the tilded variables. This solves the local HJ constraints of shape dynamics in asymptotic dS space. In this calculation, the gauge invariance has allowed us to construct a general solution to \emph{all} the local constraints of the theory in the large $V$ expansion. This would \emph{not} be possible in GR since the local quadratic constraints are considerably more complicated, providing further evidence that the local constraints of shape dynamics are helpful in simplifying certain calculations.

\section{The semiclassical correspondence}

We will now use our solution of the HJ equation to establish a semiclassical correspondence between shape dynamics and CFT. In the semiclassical approximation, the phase of the wavefunctional is given by the solution to the HJ equation. Thus,
\begin{equation}
   \Psi_\text{sd} = \Psi_\text{sd}^+ + \Psi_\text{sd}^- =a_+ e^{\frac i \hbar S_+} + a_- e^{\frac i \hbar S_-},
\end{equation}
where $S_\pm$ represent the two solutions of the HJ equation. Using the gauge invariance of our solutions under volume preserving conformal transformations, it follows, by differentiating with respect to the volume, that the $\Psi_\text{sd}^\pm$ obey the conformal Ward identities
\begin{equation}
   i\hbar g_{ab} \frac 1 {\sqrt g} \ddiby{}{g_{ab}} \Psi_\text{sd}^\pm(g) = \mp A(g) \Psi_\text{sd}^\pm(g),
\end{equation}
where
\begin{eqnarray} \label{eq:conf_ano}
&&   A(g) = \sqrt{ \frac {16} 3 \Lambda } - \sqrt{\frac 1 {3\Lambda} } \mean{\tilde R} ~ \lf( V/V_0\rt)^{-2/3} \\
&&- \frac 1 3 \lf(\frac 3 \Lambda\rt)^{3/2} \mean{\frac 3 8 \tilde R^2 - \tilde R^{ab} \tilde R_{ab}} ~ \lf( V/V_0\rt)^{-4/3}+ \cdots \nonumber
\end{eqnarray}
is a conformal anomaly. The advantage that shape dynamics has over GR is that the local constraints are \emph{linear} in the momenta. The shape dynamics constraints can be quantized unambiguously as vector fields on configuration space leading to linear Ward identities on the CFT side of the correspondence. It follows that the wavefunctional of shape dynamics is invariant under diffeomorphisms and volume preserving conformal transformations. The correspondence thus implies that the CFT partition function is also invariant under diffeomorphisms and volume preserving conformal transformations at all RG times and not just at the conjectured fixed point (ie, the infinite volume limit).

This correspondence suggests an interesting possibility: the potential of developing a construction principle for shape dynamics that does not rely on having GR at our disposal. Such a construction principle might be obtained through the correspondence by trying to implement Barbour's understanding of time \cite{barbour:eot}. The identification of ``volume time'' with ``RG time'' suggests that time is identified with the level of coarse graining of a CFT. Coarse graining is roughly a restriction of the complexity of configuration space. If true, this would imply that time is given by complexity. Thus, this construction principle for the shape dynamics Hamiltonian would provide a realization of Barbour's idea that the flow of time enters a timeless theory through a measure of complexity. He calls this measure the abundance of ``time capsules.''

Our derivation is close in spirit to \cite{Verlinde} and is particularly inspired by Freidel \cite{Freidel}. As mentioned, being able to impose all local constraints is not technically viable in GR, and thus provides an enormous advantage of the shape dynamics approach. Our central charge is a genuine central charge, even away from the fixed point. This is in contrast to GR where the constraints are no longer first class with respect to the conformal constraints away from the fixed point.

Our assumptions for the construction of the large volume shape dynamics Hamiltonian are compact CMC slices and the existence of trajectories that reach the large volume limit. Thus, we have shown that the correspondence from shape dynamics does not need to assume asymptotic (A)dS space, but is a generic large CMC volume gravity/CFT correspondence. To obtain the particularly simple HJ functional, we furthermore assume late time homogeneity. In light of these advantages, we believe that shape dynamics may be the natural framework for further exploring the connection between gravity in the large CMC volume limit and boundary CFT.

%-------------------------------- Conclusions -------------------------
%----------------------------------------------------------------------
% CONCLUSIONS / OUTLOOK
%----------------------------------------------------------------------
%
%----------------------------------------------------------------------
\chapter{Conclusions / Outlook}
%----------------------------------------------------------------------

We have shown that it is possible to derive gravity starting from two simple principles:
\begin{enumerate}
    \item that all measurements of length reduce to local comparisons.
    \item that duration should be a measure of the relative change in the universe.
\end{enumerate}
The first principle implies that it is the local \emph{shape} degrees of freedom that are physically relevant for the evolution. This motivates that the transformation
\begin{equation}\label{eq:main symmetry}
    g_{ab}(x) \to e^{4\phi(x)} g_{ab}(x)
\end{equation}
should be a gauge symmetry of a theory of gravity, where $g_{ab}$ is a dynamical spatial metric and $\phi$ is an arbitrary conformal factor. The second principle implies that the dynamics of the theory should be given by a geodesic principle on shape space. We then reviewed a procedure, called \emph{best matching}, that implements these principles simultaneously.

The idea behind best matching, we saw, was to shift the redundant configuration variables used in the theory (in shape dynamics this is the 3--metric) along the symmetry directions until the difference between two different snapshots of the configurations, calculated with some choice of metric, is minimized. This minimum distance gives the value of the metric on shape space. We showed that this procedure was equivalent to treating configuration space as a fibre bundle over shape space and then choosing a particular section on this fibre bundle. Best matching is, thus, a way of implementing Mach's principles by doing gauge theory on configuration space. In phase space, we showed that best matching is equivalent to performing a particular canonical transformation then imposing an extra condition. For certain choices of metric on configuration space, this leads to the standard Gauss constraints encountered in gauge theories. For others, this extra condition provided a gauge fixing of one of the other first class constraints of the theory.

Once the structure of the fibre bundle has been identified, the only ambiguity in the procedure is a choice of metric on configuration space. We showed that there exists a unique choice of metric that leads to a theory, shape dynamics, whose dynamics is equivalent to GR, even with the volume preserving restriction. The relationship between shape dynamics and GR is most easily understood by noting that both theories are different gauge fixings of a larger \emph{linking} theory. In this picture, it is possible to ``trade'' one symmetry for another by lifting to the linking theory and performing the appropriate gauge fixing. However, in order to get non--trivial dynamics in the theory, it was necessary to trade \emph{all but one} of the local Hamiltonian constraints of GR for constraints that generate local conformal transformations. For this reason, there is a global restriction on the transformation \eq{main symmetry} in shape dynamics. This restriction is that the conformal transformation must preserve the 3--volume of the universe when the spacial topology is compact and without boundary.

The conformal symmetry \eq{main symmetry} that is obtained from this procedure is technically and conceptually much simpler than foliation invariance in GR. The cost of this simplification is non--locality in the Hamiltonian. This non--locality originates from the fact that we must perform a phase space reduction in the linking theory to obtain shape dynamics. This phase space reduction involves the inversion of a partial differential equation. The challenge in shape dynamics is, thus, to compute this non--local Hamiltonian explicitly. We provided two expansions where this is possible. The first, was an expansion in large volume. This expansion reveals that shape dynamics is purely conformal at large volume. The solution to the Hamilton--Jacobi equation about the large volume limit suggests an intriguing connection between Hamiltonian flow away from the large volume limit and RG in a boundary CFT. The second expansion was in terms of small perturbations about a background. This expansion may be particularly useful for cosmology.

We have shown that there exists an equivalent formulation of GR that has a different symmetry group and is motivated from simples principles. Obviously, there are many interesting possible ways to extend this work. First, it is important to understand the structure of this new theory. For example, there are many interesting conceptual questions like how to explain the twins paradox and length contraction in a theory with an absolute notion of simultaneity. Alternatively, it would be interesting to study how highly symmetric solutions could be derived directly from shape dynamics without reference the analogous solution in GR in CMC gauge. More interesting, though, would be the potential for doing new calculations in cosmology. Because gauge fixing is straightforward in shape dynamics, certain calculations may be easier than the analogous calculation in GR.

However, the most interesting implications are those for the quantum theory. The non--locality at the classical level should not allow for any predictions different from that of GR. But, at the quantum level, the non--locality might manifest itself differently. This would be an interesting possibility to explore. The local symmetries of shape dynamics are identical to those of Ho\v rava--Lifshitz gravity. However, the non--locality of the dynamics makes it impossible to use the power counting arguments normally used for perturbative renormalizability. Nevertheless, it may still be true that this new theory will change the RG flow of the theory into the UV. It would be interesting to compare the non--perturbative behavior of quantum shape dynamics with that of GR. Finally, the potentially most interesting extension of this work would be to identify a principle for selecting the metric on configuration space. This would have to pick out the precise form of the shape dynamics Hamiltonian without having to refer directly to GR. One possible mechanism would be to take inspiration from the AdS/CFT correspondence and try to equate Hamiltonian flow away from large volume with RG flow in a boundary CFT. If such a principle could be found, it would potentially explain how to quantize a field theory that is fundamentally non--local. It would also provide the missing link necessary for deriving gravity directly from a simple set of well motivated principles.

% The \appendix statement indicates the beginning of the appendices.
\appendix

% Add a title page before the appendices and a line in the Table of Contents
\chapter*{APPENDICES}
\addcontentsline{toc}{chapter}{APPENDICES}

%------------------------------- USEFUL FORMULA ----------------------
%======================================================================
\chapter{Useful Formula}
\label{Appendix-useful formula}
%======================================================================

\section{De Sitter background in CMC}

In this section, we give the metric, $g_{ab}$, and its conjugate momentum density, $\pi^{ab}$, for a deSitter spacetime in CMC gauge in arbitrary dimension, $d$. We also calculate useful quantities like the spatial curvature, $R$, and other important phase space functions.

\begin{align}
g_{\mu\nu}(t,r,\theta_1, \theta_2) &= \lf( \begin{array}{c c}
                                -1 & \\ & \alpha^2\ct{2} \Omega_{ab}
                               \end{array} \rt) & \Omega_{ab}(r,\theta_1, \theta_2) = \lf( \begin{array}{c c c} 
                               1 & & \\ & \sin^2 \theta_1 & \\ &  & \sin^2 \theta_1 \sin^2 \theta_2
                               \end{array}
                               \rt)
\end{align}

\begin{equation}
    \begin{array}{c c c}
	\sqrt\Omega = \sin^2 \theta_1 \sin \theta_2 & V_0 = \int d^3 x \sqrt\Omega & V = \int d^3x \sqrt g = \alpha^3 \ct{3} V_0\\ & & \\
	
	g_{ab} = \alpha^2 \ct{2} \Omega_{ab} & g^{ab} = \frac{\Omega^{ab}}{\alpha^2 \ct{2}} & \sqrt g = \alpha^3 \ct{3} \sqrt\Omega \\ & & \\
	\pi^{ab} = -\frac{(d-1)}{\alpha} \tht{} g^{ab} \sqrt g & \pi = - \frac{d(d-1)}{\alpha} \tht{} \sqrt g & \mean{\pi} = -\frac{d(d-1)}{\alpha} \tht{} \\ & & \\
	R = \frac {d(d-1)} {\alpha^2 \ct{2}} & \Lambda = \frac{d(d-1)}{2\alpha^2} & \frac{1}{\sqrt g} \lf[ \pi\cdot\pi - \frac{\pi^2}{2} \rt] = -6 \tht{2} \sqrt g \\ & & \\
	G_{abcd} = g_{ac} g_{bd} - \frac 1 {d-1} g_{ab} g_{cd} & & 
    \end{array}
\end{equation}

\section{Useful Variations}\label{variations}

In this section, we give some useful variations of phase space quantities for calculating $\hn{1}$.

The expansion is:
\begin{itemize}
    \item $  g_{ab} \to g_{ab} + \epsilon h_{ab} $
    \item $ \pi^{ab} \to \pi^{ab} + \epsilon p^{ab} $
    \item $ \hd = \hn{0} + \epsilon\hn{1} + \epsilon^2\hn{2} +\hdots$
    \item $ \Omega^6 = \wn{0} + \epsilon\wn{1} + \epsilon^2\wn{2} +\hdots$
\end{itemize}

In terms of these, we find:
\begin{itemize}
    \item $ \Omega^n = 1 + \epsilon\, n\wn{1} + \epsilon^2 \lf( n \wn{2} + \frac{n!}{2!(n-2)!}\wn{1} \rt) + \ordn{2} $
    \item $ g^{ab} \to g^{ab} - \epsilon h^{ab} + \ordn{2} $
    \item $ \sqrt g \to \sqrt g \lf( 1 + \epsilon h/2 + \ordn{2} \rt) $
    \item $ \pi \to \pi + \epsilon (p + \pi\cdot h) + \ordn{2} $
    \item $ \mean{\pi} \to \mean{\pi} + \epsilon \lf( \mean{\pi\cdot h} - 1/2 \mean{\pi}\mean{h} + \mean{p} \rt) $
    \item $ \sigma^{ab} = \pi^{ab} - \frac 1 3 \mean{\pi} g^{ab} \sqrt g \to \sigma^{ab} + \epsilon \lf[ p^{ab} - \frac 1 3 \lf( \lf( \delta \mean{\pi} - \mean{\pi} \frac h 2 \rt) g^{ab} - \mean{\pi} h^{ab} \rt) \sqrt g \rt] + \ordn{2} $
    \item $ \sigma\cdot\sigma \to \sigma\cdot\sigma + 2\epsilon \lf( \sigma\cdot \delta\sigma + \sigma\cdot g \cdot h \cdot \sigma \rt) + \ordn{2} $
    \item $ \pi\cdot \pi \to \pi\cdot \pi + 2 \epsilon \lf( \pi \cdot p + \pi\cdot g \cdot h \cdot \pi \rt) + \ordn{2} $
    \item $  R \to R + \epsilon \lf( G\cdot h + h^{ab}_{;ab} + \nabla^2 R \rt) + \ordn{2}$
\end{itemize}

%----------------------------------------------------------------------
% END MATERIAL
%----------------------------------------------------------------------

% B I B L I O G R A P H Y
% -----------------------

% The following statement selects the style to use for references.  It controls the sort order of the entries in the bibliography and also the formatting for the in-text labels.
\bibliographystyle{utphys}
% This specifies the location of the file containing the bibliographic information.  
% It assumes you're using BibTeX (if not, why not?).
\cleardoublepage % This is needed if the book class is used, to place the anchor in the correct page,
                 % because the bibliography will start on its own page.
                 % Use \clearpage instead if the document class uses the "oneside" argument
\phantomsection  % With hyperref package, enables hyperlinking from the table of contents to bibliography             
% The following statement causes the title "References" to be used for the bibliography section:
\renewcommand*{\bibname}{References}

% Add the References to the Table of Contents
\addcontentsline{toc}{chapter}{\textbf{References}}

\bibliography{mach}
% Tip 5: You can create multiple .bib files to organize your references. 
% Just list them all in the \bibliogaphy command, separated by commas (no spaces).

% The following statement causes the specified references to be added to the bibliography% even if they were not 
% cited in the text. The asterisk is a wildcard that causes all entries in the bibliographic database to be included (optional).
%\nocite{*}

\end{document}